\pdfoutput=1
\documentclass[11pt]{article}
\usepackage{fullpage}
\usepackage{mathtools,amsfonts,amsthm,amssymb,bbm}
\usepackage{thm-restate,thmtools}
\usepackage{xcolor}
\usepackage[backref,colorlinks,citecolor=blue,linkcolor=magenta,bookmarks=true]{hyperref}
\usepackage[nameinlink]{cleveref}

\usepackage[pdftex]{graphicx} 
\usepackage[english]{babel} 
\usepackage{calc} 
\usepackage{enumitem} 

\usepackage[all]{nowidow} 
\usepackage[protrusion=true,expansion=true]{microtype} 

\usepackage{graphicx}
\usepackage[algo2e,ruled,linesnumbered,noend]{algorithm2e}
\Crefname{algocf}{Algorithm}{Algorithms}

\usepackage{algpseudocode}
\newtheorem{theorem}{Theorem}
\newtheorem{problem}{Problem}
\newtheorem{lemma}{Lemma}
\newtheorem*{theorem*}{Theorem}

\Crefname{remark}{Remark}{Remark}
\newtheorem{proposition}{Proposition}
\newtheorem*{proposition*}{Proposition}
\newtheorem{definition}{Definition}

\newtheorem*{corollary*}{Corollary}

\newcommand{\1}[1]{\mathbbm{1}\left\{#1\right\}} 
\newcommand{\A}{\mathcal{A}} 
\newcommand{\Bern}{\mathsf{Bernoulli}} 
\newcommand{\Binomial}{\mathsf{Binomial}} 
\newcommand{\choosetopk}{\textsf{Choose-Top-}k}
\newcommand{\counter}{\mathrm{counter}}
\newcommand{\CR}{\mathsf{CR}}
\newcommand{\eps}{\varepsilon} 
\newcommand{\event}{E}
\newcommand{\Ex}[2]{\operatorname*{\mathbb{E}}_{#1}\left[#2\right]} 
\newcommand{\findkth}{\textsf{Find-}k\textsf{-th}}
\newcommand{\Geo}{\mathsf{Geo}}
\newcommand{\goodevent}{\event^{\textsf{good}}} 
\newcommand{\OPT}{\mathsf{OPT}}
\newcommand{\poly}{\operatorname*{poly}} 
\newcommand{\polylog}{\operatorname*{polylog}} 
\newcommand{\pr}[2]{\Pr_{#1}\left[#2\right]} 
\newcommand{\rank}{\mathrm{Rank}} 
\newcommand{\rmd}{\mathrm{d}}
\newcommand{\SOL}{\mathsf{SOL}}
\newcommand{\Var}[2]{\operatorname*{\mathrm{Var}}_{#1}\left[#2\right]} 

\title{Optimal $k$-Secretary with Logarithmic Memory}
\author{
Mingda Qiao\thanks{University of Massachusetts Amherst. Email: \texttt{mingda.qiao.cs@gmail.com}.}
\and
Wei Zhang\thanks{Massachusetts Institute of Technology. Email: \texttt{weiz41798@gmail.com}.}
}
\date{}

\begin{document}
\maketitle

\begin{abstract}%
We study memory-bounded algorithms for the $k$-secretary problem. The algorithm of Kleinberg (SODA 2005) achieves an optimal competitive ratio of $1 - O(1/\sqrt{k})$, yet a straightforward implementation requires $\Omega(k)$ memory. Our main result is a $k$-secretary algorithm that matches the optimal competitive ratio using $O(\log k)$ words of memory. We prove this result by establishing a general reduction from $k$-secretary to (random-order) \emph{quantile estimation}, the problem of finding the $k$-th largest element in a stream. We show that a quantile estimation algorithm with an $O(k^{\alpha})$ expected error (in terms of the rank) gives a $(1 - O(1/k^{1-\alpha}))$-competitive $k$-secretary algorithm with $O(1)$ extra words. We then introduce a new quantile estimation algorithm that achieves an $O(\sqrt{k})$ expected error bound using $O(\log k)$ memory. Of independent interest, we give a different algorithm that uses $O(\sqrt{k})$ words and finds the $k$-th largest element exactly with high probability, generalizing a result of Munro and Paterson (1980).
\end{abstract}

\section{Introduction}
In the classical secretary problem, $n$ numbers are presented to an online decision-maker (``player'') in random order. Upon seeing each number, the player has to make an irrevocable decision on whether to accept it---if the player accepts, the game terminates immediately; otherwise, the player moves on to the next element in the sequence. The goal of the player is to maximize the chosen number. It is well-known that there is a strategy for the player that chooses the largest element with probability at least $1/e$ for any value of $n$, which is essentially optimal.

The $k$-secretary problem is a natural extension of the above---instead of choosing a single element, the player is now allowed to accept up to $k$ elements. For this problem, Kleinberg~\cite{Kleinberg05} gave a $(1 - O(1/\sqrt{k}))$-\emph{competitive} algorithm. Here, an algorithm is $\alpha$-competitive if the expected sum of the accepted elements is at least an $\alpha$-fraction of the maximum possible result---the sum of the $k$ largest numbers. It was also shown that the competitive ratio of any algorithm is at best $1 - \Omega(1/\sqrt{k})$.

In this work, we revisit the $k$-secretary problem under an additional memory constraint on the player's algorithm. One potential use case of the $k$-secretary model is the routing setting~\cite{nakano2011message}, in which $n$ data packets arrive at a router in a sequential order, while the goal is to filter out certain low-quality packets and only forward $k$ out of the $n$ packets for downstream processing. In this setup, as the values of both $n$ and $k$ can be huge compared to the storage of the router, we need the algorithm to be both optimal (in terms of the competitive ratio) and memory-efficient (using a memory sublinear, or even logarithmic, in $k$).

Unfortunately, Kleinberg's algorithm, while being optimal, is not designed to be memory-efficient. A straightforward implementation of the algorithm requires $\Omega(k)$ words of memory. The bottleneck of the memory usage is for finding the $k$-th largest element in the sequence---a problem known as \emph{quantile estimation}---for which the straightforward method requires $\Omega(k)$ space.

One might hope to improve the memory bound using quantile estimation algorithms that are more space-efficient. We first note that there exist strong $\poly(k)$ memory lower bounds if we restrict ourselves to \emph{exact} algorithms that find the $k$-th largest element exactly (with high probability): a lower bound of Munro and Paterson~\cite{MP80} suggests that this requires $\Omega(\sqrt{k})$ space. Naturally, one might hope to use algorithms that find the $k$-th largest element approximately (e.g., returning an element with rank $k \pm o(k)$). Guha and McGregor~\cite{guha2009stream} gave an algorithm that uses $O(1)$ words of memory\footnote{We assume a word size of $\Theta(\log n)$.} and finds an element of rank $k \pm O(\sqrt{k}\cdot \log^2n)$. Combining this result with Kleinberg's algorithm, however, would lead to extra $\polylog(n)$ factors in the competitive ratio.

In this work, we aim to answer the following questions:
\begin{quote}
    \emph{In general, does an accurate quantile estimator lead to a competitive $k$-secretary algorithm while maintaining its memory usage? Concretely, to match the optimal competitive ratio for $k$-secretary, how much memory is needed?}
\end{quote}

\subsection{Problem Setup}
\paragraph{Notations.} Throughout the paper, $s = (s_1, s_2, \ldots, s_n)$ denotes a random-order sequence, and we shorthand $s_{i_1:i_2}$ for the subsequence $(s_{i_1}, s_{i_1 + 1}, \ldots, s_{i_2})$. When it is clear from the context, we sometimes abuse the notation and let $s_{i_1:i_2}$ denote the set formed by the elements in the subsequence. For instance, $s_{i_1:i_2} \cap (-\infty, a)$ is used as a shorthand for $\{x \in \{s_{i_1}, s_{i_1 + 1}, \ldots, s_{i_2}\}: x < a\}$.

We will frequently use the \emph{rank} of an element $x$ within a contiguous subsequence, when only the elements that are strictly smaller than a threshold $x'$ are counted.
\begin{definition}[Rank]\label{def:rank}
For sequence $(s_1, s_2, \ldots, s_n)$ and $1 \le i_1 \le i_2 \le n$, we define
\[
    \rank_{i_1:i_2}(x; x')\coloneqq  \left|\left\{t\in \{s_{i_1}, \ldots, s_{i_2}\}| x' > t \ge x\right\}\right|.
\]
We also write $\rank_{i_1:i_2}(x)$ as a shorthand for $\rank_{i_1:i_2}(x; +\infty)$.
\end{definition}

\paragraph{The $k$-Secretary problem.}
We formally define the $k$-secretary problem as follows.
\begin{problem}[$k$-Secretary]\label{def.ksecretary}
    Let $x_1 > x_2 > \cdots > x_n \ge 0$ be $n$ numbers that are non-negative, distinct, and unknown to the algorithm. The algorithm reads a uniformly random permutation of the $n$ elements one by one. Upon seeing each element, the algorithm decides whether to accept it. The algorithm may accept at most $k$ elements, and aims to maximize their sum.
\end{problem}

The assumption that all the elements are distinct is for the simplicity of the exposition. It comes without loss of generality by a standard tie-breaking argument.\footnote{For example, we may augment each element with a random identifier with $\Theta(\log n)$ bits. This ensures distinctness with high probability, and only requires one additional word for storing each element.}

We say that a $k$-secretary algorithm is competitive, if it is guaranteed to secure a certain fraction of the maximum possible outcome, namely, the sum of the $k$ largest elements.

\begin{definition}[Competitive ratio]\label{def.cr}
    A $k$-secretary algorithm is $\alpha$-competitive if, on any $k$-secretary instance $x_1 > x_2 > \cdots > x_n$, the expected sum of the accepted elements is at least $\alpha \cdot \sum_{i=1}^{k}x_i$. Here, the expectation is over the randomness both in the permutation and in the algorithm.
\end{definition}

\paragraph{Quantile estimation.} Next, we formally define quantile estimation in a random-order stream.

\begin{problem}[Quantile estimation]\label{def.quantile_estimation}
Let $x_1 > x_2 > \cdots > x_n$ be $n$ distinct numbers that are unknown to the algorithm. The algorithm takes $n$ and $k$ as inputs, and reads a uniformly random permutation of the $n$ numbers one by one. The goal is to output the (approximately) $k$-th largest element $x_k$. When the output is $x^* \in \{x_1, x_2, \ldots, x_n\}$, the algorithm incurs an error of  $|\rank_{1:n}(x^*) - k|$.
\end{problem}

Again, the assumption that the $n$ elements are distinct is for simplifying the notations. The results for distinct elements can be extended to the general case via tie-breaking and a more careful definition of the error.

\paragraph{Memory model.} We informally introduce the model of computation and the memory usage of an algorithm. All our algorithms can be implemented in the word RAM model, assuming that each word can store either an integer of magnitude $\poly(n)$ or an element in the stream (in either $k$-secretary or quantile estimation). In particular, when all stream elements are bounded by $\poly(n)$, a word size of $\Theta(\log n)$ is sufficient. Moreover, all our algorithms are comparison-based: they only access the stream elements via pairwise comparisons, and do not perform any arithmetic operations on the elements. Therefore, a word size of $\Theta(\log n)$ would suffice if: (1) every stream element is presented as a unique identifier (which takes $O(\log n)$ bits); (2) the algorithm has access to an oracle that compares two elements specified by their identifiers.

\subsection{Our Results}
Our main result shows that logarithmic memory is sufficient for being optimally-competitive in the $k$-secretary problem.
\begin{theorem}\label{thm:k-secretary}
    There is a $k$-secretary algorithm that uses $O(\log k)$ words and achieves a competitive ratio of $1 - O(1/\sqrt{k})$.
\end{theorem}

\paragraph{Reduction from $k$-secretary to quantile estimation.} Towards proving \Cref{thm:k-secretary}, we show that a quantile estimation algorithm leads to a competitive $k$-secretary algorithm with almost the same memory usage. Formally, the reduction applies to all \emph{comparison-based} quantile estimators. Roughly speaking, an algorithm is comparison-based if it only accesses the elements in the stream via pairwise comparisons. As a result, the output of a comparison-based algorithm is invariant (up to the renaming of elements) under any order-preserving (i.e., monotone) transformations. See \Cref{def:comparison-based} for a more formal definition.

\begin{restatable}{proposition}{reduction}\label{prop:reduction}
    Suppose that, for some $\alpha \in [1/2, 1]$, there is a comparison-based quantile estimation algorithm that uses $m$ words of memory and has an error of $O(k^{\alpha})$ in expectation. Then, there is a $k$-secretary algorithm that uses $m + O(1)$ words and achieves a competitive ratio of $1 -O(1/k^{1 - \alpha})$.
\end{restatable}

The proposition follows from a reduction that is essentially implicit in \cite{Kleinberg05}---Kleinberg's algorithm can be viewed as a special case in which the quantile estimator is exactly correct (and thus, we can take $\alpha = 1/2$). Our proof of \Cref{prop:reduction} shows that this reduction is actually robust to the inaccuracy in the quantile estimator, so the error bound in quantile estimation can be smoothly translated to the sub-optimality in $k$-secretary.

As we explain in \Cref{sec:overview-reduction}, directly following the reduction of~\cite{Kleinberg05} would increase the memory usage by a factor of $\log k$. We avoid this multiplicative increase by carefully modifying the reduction, so that the $k$-secretary algorithm would run $\log k$ copies of the quantile estimator \emph{sequentially}, rather than in parallel.

\paragraph{New results for quantile estimation.} In light of \Cref{prop:reduction}, the only missing piece towards proving \Cref{thm:k-secretary} is a memory-efficient quantile estimator with $O(\sqrt{k})$ error.

\begin{restatable}{theorem}{quantileapprox}\label{thm:quantile-approx}
    There is a quantile estimation algorithm that, when finding the $k$-th largest element, uses $O(\log k)$ words and incurs an $O(\sqrt{k})$ error in expectation over the randomness in both the algorithm and the order of the stream.
\end{restatable}

We note that both the memory usage and the error bound are independent of $n$. This is because we address the \emph{expected} error directly, rather than conditioning on a good event involving the entire length-$n$ stream, thus avoiding a potential $\polylog(n)$ dependence.

The key idea behind the algorithm for \Cref{thm:quantile-approx} is to reduce the problem of finding the $k$-th largest element to finding the $k'$-th largest in a subsequence, for some $k' \ll k$. By shrinking the value of $k$ fast enough, we ensure that the error incurred in the subproblems can be controlled. In \Cref{sec:overview-approx,sec:overview-details}, we sketch the algorithm and its analysis in more detail.

Our second algorithm finds the $k$-th largest element exactly with high probability, albeit with a looser memory bound of $O(\sqrt{k})$.

\begin{restatable}{theorem}{quantileexact}\label{thm:quantile-exact}
    For any $m \ge 1$, there is a quantile estimation algorithm that uses $O(m)$ words of memory and incurs zero error (i.e., returns the exactly $k$-th largest element) with probability at least
    \[
        1 - 12\lfloor\log_2k\rfloor\cdot\exp\left(-\frac{m}{12}\right) - 2\sum_{i=0}^{\lfloor\log_2k\rfloor-1}\exp\left(- \frac{m^2}{32(k/2^i)}\right).
    \]
\end{restatable}

In particular, setting $m = O(\sqrt{k\log(1/\delta)} + \log(1/\delta))$ is sufficient for succeeding with probability $\ge 1 - \delta$.

\Cref{thm:quantile-exact} follows from the technique of \cite{MP80}: the algorithm maintains $m$ consecutive elements among the stream that has been observed so far, in the hope that the $k$-th largest element is among these $m$ elements at the end. \cite{MP80} showed that this strategy finds the median of a length-$n$ random-order stream (i.e., the $k = n/2$ case) using $O(\sqrt{n})$ memory, and proved a matching $\Omega(\sqrt{n})$ lower bound. Our result extends the result of~\cite{MP80} to general $k$. In \Cref{sec:overview-exact}, we introduce the technique of~\cite{MP80} in more detail, and sketch our strategy for handling biased quantiles (i.e., the $k \ll n$ case).

\subsection{Related Work}
\paragraph{The secretary problem and its variants.} 
The classical secretary problem is often attributed to~\cite{dynkin1963optimum}, though the exact origin of the problem is obscure~\cite{freeman1983secretary,Ferguson1989}; different versions of the problem were studied in the contemporary work of~\cite{Lindley61,CMRS64,GM66}. 

The natural extension of selecting up to $k$ elements (i.e., the $k$-secretary problem) was studied by~\cite{GM66,AMW01,Kleinberg05,BIKK07} under various objectives (e.g., the probability of selecting the $k$ largest elements exactly, or the largest element being among the chosen ones). We followed the formulation of~\cite{Kleinberg05} in terms of the competitive ratio. \cite{Kleinberg05} showed that the optimal competitive ratio is $1 - \Theta(1/\sqrt{k})$ as $k \to \infty$. \cite{BIKK07} subsequently gave $1/e$-competitive algorithms for all $k \ge 1$, improving the result of~\cite{Kleinberg05} in the small-$k$ regime. A more recent work of~\cite{AL21} also focused on the non-asymptotic regime, and gave a deterministic algorithm with a competitive ratio $> 1/e$ for all $k \ge 2$.

A further extension of the $k$-secretary problem considers selecting multiple elements subject to a more general combinatorial constraint. \cite{BIKK07} studied the knapsack secretary problem. \cite{BIK07,BIKK18} introduced the \emph{matroid secretary problem}, in which the player is asked to select elements that form an independent set of a given matroid. This problem has been extensively studied~\cite{chakraborty2012improved,soto2013matroid,jaillet2013advances,dinitz2014matroid,lachish2014log, feldman2014simple,huynh2020matroid,soto2021strong,AKKO23}, and it remains a long-standing open problem whether an $O(1)$-competitive algorithm exists in general. Other variants of the secretary problem consider alternative models for how the player accesses the elements and makes decisions, including models of interview costs~\cite{bartoszynski1978secretary}, shortlists~\cite{ASS19}, reservation costs~\cite{burjons2021secretary}, and a ``pen testing'' variant~\cite{QV23,ganesh2023combinatorial}.

\paragraph{Quantile estimation.} For the quantile estimation problem in random-order streams, \cite{MP80} gave an exact selection algorithm for the $k = \lfloor n/2\rfloor$ case (i.e., finding the median) with $O(\sqrt{n})$ memory, and proved a matching $\Omega(\sqrt{n})$ memory lower bound. Our \Cref{thm:quantile-exact} generalizes their result to general values of $k$, and their lower bound implies that the $O(\sqrt{k})$ memory usage cannot be improved in general (specifically, in the $n = 2k$ case).

For approximate selection, \cite{guha2009stream}  proposed an algorithm that uses $O(1)$ words of memory and finds the $k$-th largest element up to an error of $O(\sqrt{k}\log^2 n\cdot\log(1/\delta))$ with probability $1 - \delta$. In comparison, \Cref{thm:quantile-approx} bounds the expected error and avoids the extra $\polylog(n)$ factor, both of which are crucial for the optimality of the resulting $k$-secretary algorithm. A subsequent work of~\cite{MV12} solved median estimation with an $n^{1/3 + o(1)}$ error using $O(1)$ words of memory. In comparison, \Cref{thm:quantile-approx} applies to all values of $k$ and makes the error dependent only on $k$ and not $n$, at the cost of a larger exponent of $1/2$ and a logarithmic memory usage. We note that, even if we could improve the error to $O(k^{1/3})$, the reduction from $k$-secretary still introduces an $\Omega(\sqrt{k})$ error, which would dominate the quantile estimation error.

While we focus on the random-order setup of quantile estimation, many prior work also addressed the more challenging setup in which elements arrive in an arbitrary order~\cite{MP80,manku1998approximate,KLL16,masson2019ddsketch,gupta2024optimal}. The multi-pass setting, in which the algorithm may scan the stream multiple times, was also considered~\cite{MP80,guha2009stream}.

\paragraph{Learning and decision making under memory constraints.} More broadly, our work is part of the endeavor to understand the role of memory in learning, prediction, and decision-making. Prior work along this line studied the memory bounds for parity learning in a streaming setting~\cite{valiant2016information,steinhardt2016memory,kol2017time,raz2018fast, garg2018extractor}, for the experts problem in online learning~\cite{srinivas2022memory,PZ23,peng2023near}, as well as the fundamental problem of linear regression~\cite{sharan2019memory, marsden2022efficient, blanchard2023memory, blanchard2024gradient}.

\section{Proof Overview}
In this section, we give high-level sketches of our proofs and highlight some of the technical challenges. We recommend the readers to read this section before delving into the formal proofs in the appendix, as the technicalities might obscure some of the simple intuitions behind our algorithms.


We introduce our algorithm for \Cref{thm:quantile-approx} in \Cref{sec:overview-approx}. The key idea is to consider a sub-problem obtained by filtering out the elements that are above a certain threshold, which reduces the scale of the problem rapidly enough. In \Cref{sec:overview-details}, we discuss a few technical challenges in turning an idealized analysis of the algorithm into a formal proof.

In \Cref{sec:overview-exact}, we sketch the proof of \Cref{thm:quantile-exact}, which is based on a different idea of~\cite{MP80}. The algorithm maintains a block of $m$ consecutive elements in the stream that has been observed so far. Whenever a new element comes, the algorithm tries to ``drift'' the length-$m$ block, in the hope that the $k$-th largest element is always within the block. \cite{MP80} framed this as a random walk problem, for which we give a new solution for general $k$ and in the $m = \Omega(\sqrt{k})$ regime.

We end the section by sketching the proof of \Cref{prop:reduction} (in \Cref{sec:overview-reduction}), which reduces the $k$-secretary problem to quantile estimation using a variant of the algorithm of~\cite{Kleinberg05}.

\subsection{An Approximate Algorithm via Iterative Conditioning}\label{sec:overview-approx}
Now, we introduce our algorithm for \Cref{thm:quantile-approx}. The algorithm is a recursive one: we repeatedly reduce the problem of finding the $k$-th largest element to finding the $k'$-th largest (for some $k' \ll k$) in a subsequence of the stream. Before presenting the actual algorithm, we start by explaining why a na\"ive attempt fails, after which the key idea of our algorithm becomes more transparent.

\paragraph{A na\"ive reduction.} How should we solve quantile estimation for a specific value of $k$, if we already have an algorithm (denoted by $\A_{k/2}$) that solves the $k' = k/2$ case with an $O(\sqrt{k'}) = O(\sqrt{k})$ error? The answer is simple: we run $\A_{k/2}$ on the first $n/2$ elements of the sequence and return its answer. Assuming that the output of $\A_{k/2}$, denoted by $x^*$, is exactly the $(k/2)$-th largest element among $s_{1:(n/2)}$, a simple concentration argument would show that the rank of $x^*$ among $s_{1:n}$ is roughly $k \pm O(\sqrt{k})$. Now, let $l \coloneqq \rank_{1:(n/2)}(x^*)$ denote the actual rank of $x^*$ among the first half of the sequence. Our assumption on $\A_{k/2}$ implies that $l \approx k/2 \pm O(\sqrt{k})$, which further gives
\[
    \rank_{1:n}(x^*)
\approx 2l \pm O(\sqrt{l})
\approx k \pm O(\sqrt{k}).
\]

At first glance, the above seems to suggest an extremely simple solution to quantile estimation: we repeatedly apply the reduction above, and reduce the value of ``$k$'' to $k/2, k/4, \ldots, 1$. For the base case that ``$k$'' equals $1$, we can find the largest element exactly using $O(1)$ memory. However, this approach is equivalent to outputting the largest one among the first $\approx n/k$ elements, and its error can be easily shown to be $\Omega(k)$. What goes wrong here? The issue is that we ignored the blow-up in the error during the reduction. Define random variables $X_1, X_2, X_4, \ldots, X_{k/2}, X_k$ such that $X_{k'}$ denotes the error in the rank when we solve the instance with $k = k'$. Our discussion from the last paragraph shows that, for some universal constant $C > 0$, $(X_k)$ satisfies the dynamics
\begin{equation}\label{eq:dynamics-naive}
    X_1 = 0,
\quad
    X_k \approx 2X_{k/2} + C\sqrt{k}.
\end{equation}
Expanding the above gives
\[
    X_k
\approx C\sqrt{k} + 2C\sqrt{k/2} + 4C\sqrt{k/4} + \cdots + (k/2)\cdot C\sqrt{2}
=   \Theta(k),
\]
since the error incurred at $k' = O(1)$ would be doubled roughly $\log_2 k$ times through the iteration $X_k \approx 2X_{k/2} \pm C\sqrt{k}$, resulting in the dominant term of $\Omega(k)$ in the error.

\paragraph{The actual reduction.} We would get an $O(\sqrt{k})$ error bound if we could shrink the value of $k$ faster. Imagine that, instead of reducing to an instance with $k' = k/2$, we can reduce to an instance with $k' = k / 100$, while the error still propagates in the same way. Then, the recursion in \Cref{eq:dynamics-naive} would be replaced with $X_k \approx 2X_{k/100} + C\sqrt{k}$, which leads to
\[
    X_k \approx C\sqrt{k} + 2C\sqrt{k/100} + 2^2 C\sqrt{k/100^2} + \cdots + 2^{\log_{100} k}C\sqrt{1}
=   O(\sqrt{k}).
\]

Therefore, the key is to reduce the scale of the problem (measured by parameter $k$) at a slightly faster pace. Our algorithm does this in a fairly simple way. First, we divide the length-$n$ sequence $s$ into two halves $s_{1:B}$ and $s_{(B+1):n}$ for some $B \approx n / 2$. Then, we subsample a random $p$-fraction of the first half, where $p = \Theta(m/k)$, by drawing $B_1 \sim \Binomial(B, p)$ and taking the first $B_1$ elements. We find the top $m$ elements among $s_{1:B_1}$, denoted by $M[1], M[2], \ldots, M[m]$, using $O(m)$ memory (in the straightforward way). After reading the remaining $B - B_1$ elements in $s_{1:B}$, we can compute the rank $\rank_{1:B}(M[i])$ for every $i \in [m]$. Then, we find an index $i^*$ such that
\[
    \rank_{1:B}(M[i^*])
<   k/2
<   \rank_{1:B}(M[i^* + 1]).
\]
Let $a' \coloneqq M[i^*]$ and $k' \coloneqq k/2 - \rank_{1:B}(a')$. We know that the $(k/2)$-th largest element among $s_{1:B}$ is simply the $k'$-th largest element among $s_{1:B} \cap (-\infty, a')$. Therefore, we recursively find the $k'$-th largest element among $s_{(B+1):n} \cap (-\infty, a')$, in the hope that it will be a good approximation for the $k$-th largest element overall.

To make this work, we need the parameter $k'$ for the recursive call to be smaller than $k$ be a sufficiently large factor. This is indeed the case, since the gaps in the ranks of $M[1], M[2], \ldots, M[m]$ among $s_{1:B}$ roughly follow the geometric distribution $\Geo(p)$, so each of them has an expectation of $\le 1/p = O(k/m)$. By setting $m = \Omega(\log k)$, we can ensure that every gap is smaller than $C_0 \cdot k$ (for any small constant $C_0$), with probability $1 - 1/\poly(k)$.

The reduction sketched above actually reduces quantile estimation to a slightly different problem: instead of finding the $k$-th largest element among $s_{1:n}$, we find the $k$-th largest among $s_{1:n} \cap (-\infty, a')$. This does not pose a problem as the reduction works for this generalized version as well.

\subsection{A Few Technical Details}\label{sec:overview-details}

It should be noted that, while the intuition behind the algorithm is simple, it turns out to be highly non-trivial to state and analyze the algorithm rigorously.

\paragraph{Handling edge-cases.} Formally, the recursive algorithm takes parameters $n$, $k$, and $a$ as inputs, and aims to find the $k$-th largest element among $s_{1:n} \cap (-\infty, a)$, where $s_{1:n}$ is the upcoming random-order sequence of length $n$. The quantile estimation instance is invalid if $k > n' \coloneqq |s_{1:n} \cap (-\infty, a)|$. While the outermost call to the algorithm (where $a = +\infty$) is always valid, the algorithm might eventually lead to a recursive call in which the parameter $k$ becomes invalid. In this case, our algorithm always returns $-\infty$ as the answer. Then, in the previous level of recursion---which makes this invalid call---the algorithm translates this $-\infty$ to the smallest element in the stream.

The second edge case is that we fail to find a good choice of $a'$ that significantly reduces the value of $k$. This can, in turn, happen for three different reasons: (1) $|s_{1:B} \cap (-\infty, a)| < k/2$, in which the first half of the sequence does not have a $(k/2)$-th largest element; (2) the $(k/2)$-th largest element exists, but none of the elements with ranks in $[k/2 - C_0\cdot k, k/2]$ get subsampled into $s_{1:B_1}$, so none of them is present in the array $M$; (3) Some element with rank in $[k/2 - C_0\cdot k, k/2]$ gets subsampled into $s_{1:B_1}$, but more than $m$ elements with ranks in $[1, k/2]$ get subsampled, so that the one with a rank closest to $k/2$ is not kept in $M$.

For the latter two cases, we show that for some appropriate choice of $m = \Omega(\log k)$, their total probability is at most $1/\poly(k)$. Thus, if we detect that either of them happens, we can choose to return the largest element, which leads to a rank error of $k-1$ and becomes negligible after multiplying with the $1/\poly(k)$ probability. The first case is trickier since it \emph{can} happen with a constant probability.\footnote{Consider the case that $k \approx n' \coloneqq |s_{1:n} \cap (-\infty, a)|$, in which case $|s_{1:B} \cap (-\infty, a)|$ roughly follows $\Binomial(n', 1/2)$, and can be below $k/2$ with probability $\approx 1/2$.} In that case, we return the smallest element (with a rank error of $n' - k$), and it turns out that its contribution to the overall error is still under control.

\paragraph{Towards a rigorous analysis.} In addition to the need of handling the edge cases outlined above, the analysis of our algorithm is necessarily complicated due to the complex dependence between different parts of the algorithm.

For simplicity, we assume that the algorithm proceeds in a ``typical'' way, i.e., none of the edge cases happen. Furthermore, we analyze the outermost level of recursion, where the threshold parameter is $a = +\infty$ (i.e., no element gets ignored). Then, roughly speaking, our algorithm has the following three steps:
\begin{enumerate}
    \item[\textbf{Step 1.}] Compute a threshold $a' \in s_{1:B}$ and let $k' \coloneqq k/2 - \rank_{1:B}(a')$.
    \item[\textbf{Step 2.}] Run the algorithm recursively on $s_{(B+1):n}$ with parameters $n - B$, $k'$ and $a'$.
    \item[\textbf{Step 3.}] Return the output $x^*$ of the recursive call as the answer.
\end{enumerate}

Then, the straightforward approach to analyzing the above would be an inductive one: Let $l \coloneqq \rank_{(B+1):n}(x^*; a')$ denote the actual rank of $x^*$ among $s_{(B+1):n} \cap (-\infty, a')$. By the inductive hypothesis, $|l - k'| = O(\sqrt{k'})$ in expectation (denoted by $l \approx k' \pm O(\sqrt{k'})$ informally). If $B$ is sampled from $\Binomial(n, 1/2)$, $\{s_1, s_2, \ldots, s_B\}$ would be uniformly distributed among all subsets of $s_{1:n}$. Then, a concentration argument suggests
\begin{equation}\label{eq:overview-concentration}
    \rank_{1:B}(x^*) \approx (k/2 - k') + l \pm O(\sqrt{l})
\quad\text{and}\quad
    \rank_{(B+1):n}(x^*) \approx (k/2 - k') + l \pm O(\sqrt{k}).
\end{equation}
In total, we would obtain $\rank_{1:n}(x^*) = \rank_{1:B}(x^*) + \rank_{(B+1):n}(x^*) \approx k \pm O(\sqrt{k})$.

Unfortunately, the analysis above is technically incorrect. We were analyzing the \emph{conditional} concentration of $\rank_{1:B}(x^*)$ and $\rank_{(B+1):n}(x^*)$ given the values of $a'$, $k'$ and $l$. Since $(a', k', l)$ is determined by the value of $B$ as well as ordering of $s$, after the conditioning, $s_{1:B}$ is no longer uniformly distributed, which invalidates the concentration argument.

Towards a rigorous analysis, we need to carefully untangle the randomness in the three steps above. Crucially, we note that our algorithm is comparison-based, i.e., the behavior of the algorithm is unchanged if we replace $s$ with a different sequence $s'$ with the same ordering as $s$. This motivates the following order in which we ``realize'' the randomness:
\begin{itemize}
    \item First, we sample $B \sim \Binomial(n, 1/2)$.
    \item Then, towards analyzing Step~1, we realize the \emph{relative ordering} of the elements $s_1, s_2, \ldots, s_B$ (out of the $B!$ possibilities). This ordering is sufficient for determining $i \coloneqq \rank_{1:B}(a')$, the rank of $a'$ among $s_{1:B}$ as well as the value of $k'$.
    \item We also realize the value of $i_1 \coloneqq \rank_{1:n}(a')$. Conditioning on $(B, i)$, this rank $i_1$ is identically distributed as the $i$-th smallest element in a size-$B$ subset of $[n]$ chosen uniformly at random.
    \item After that, we realize the \emph{relative ordering} of the last $(n-B)$ elements $s_{(B+1):n}$ (out of the $(n-B)!$ possibilities). This, in turn, determines the rank of $x^*$ (the output of the recursive call) among $s_{(B+1):n} \cap (-\infty, a')$, denoted by $l \coloneqq \rank_{(B+1):n}(x^*; a')$.
    \item At this point, even after conditioning on the values of $(B, i, i_1, l)$, the subset $\{s_1, s_2, \ldots, s_B\}$ is still uniformly distributed among all size-$B$ subsets $S \subseteq \{s_1, s_2, \ldots, s_n\}$, subject to the constraint that the $i_1$-th largest element among $s_{1:n}$ (namely, $a'$) is the $i$-th largest element among $s_{1:B}$. 
\end{itemize}
The uniformity that we retain in the last step above allows us to rigorously prove bounds that are qualitatively similar to \Cref{eq:overview-concentration}.

\subsection{An Exact Algorithm via Maintaining Consecutive Elements}\label{sec:overview-exact}
We sketch our proof of~\Cref{thm:quantile-exact}, which is based on a very different algorithm. The key idea of the algorithm is to maintain \emph{consecutive} elements among the elements that have been observed so far, which was used by~\cite{MP80} to solve the median selection problem (i.e., the special case that $k = n / 2$).

\paragraph{The algorithm as a random walk (with limited control).} At any time $t$, after reading the first $t$ elements $s_1, s_2, \ldots, s_t$ in the random-order stream, the algorithm maintains $m$ consecutive elements out of them. Here, ``consecutive'' is with respect to the sorted order of the elements, rather than the order in which they arrive. Formally, suppose that the first $t$ elements of $s$, when sorted in decreasing order, are given by $x_1 > x_2 > \cdots > x_t$. The algorithm stores the elements $x_l, x_{l+1}, \ldots, x_r$ as well as the values of $l$ and $r$, where $r - l + 1 = m$.

Then, what happens when the next element $x' \coloneqq s_{t+1}$ arrives? By the random-order assumption, $x'$ is equally likely to fall into the $t+1$ intervals below:
\[
    (-\infty, x_t), (x_t, x_{t-1}), \ldots, (x_2, x_1), (x_1, +\infty).
\]
In particular, we have the three cases below:
\begin{enumerate}
    \item[\textbf{Case 1.}] $x' > x_l$. This happens with probability $l / (t + 1)$. In this case, we cannot add $x'$ to the array without breaking the invariant that the elements are consecutive. Therefore, we have to keep the $m$ stored elements unchanged, and the parameters $(l, r)$ become $(l', r') = (l + 1, r + 1)$ in the next step.
    \item[\textbf{Case 2.}] $x' < x_r$. Similarly, with probability $(t - r + 1) / (t + 1)$, the condition $x' < x_r$ holds. Again, we cannot update the $m$ elements, and the parameters in the next step remain unchanged, i.e., $(l', r') = (l, r)$.
    \item[\textbf{Case 3.}] $x' \in (x_r, x_l)$. The most interesting case is that $x'$ falls into one of the $r - l = m-1$ ``gaps'' among $x_l > x_{l+1} > \cdots > x_r$, which happens with probability $(m - 1) / (t + 1)$. In this case, we would obtain $m + 1$ consecutive elements among $s_{1:(t+1)}$. Then, we need to kick out one of the two elements at the ends. This gives us some freedom in deciding whether $(l', r') = (l, r)$ (by kicking out the smallest element) or $(l', r') = (l + 1, r + 1)$ (by kicking out the largest element).
\end{enumerate} 
At the end of the game, we successfully find the $k$-th largest element if $l \le k \le r$.

Following this strategy, the quantile estimation problem becomes a control problem: We start at time $t = m$ with $(l, r) = (1, m)$. At each of the following steps $t = m + 1, m + 2, \ldots, n$, the value of $(l, r)$ transitions according to the three cases above---either it transitions deterministically (in Cases 1~and~2), or the algorithm may specify whether $l$ and $r$ get incremented by $1$ in Case~3. The goal is to design our strategy in Case~3 (which may vary for different time step $t$ and the values of $(l, r)$ at that step), so that the probability of $k \in [l, r]$ is maximized at time $n$.

The above was exactly the idea in the work of~\cite{MP80}, who studied the special case of $k = n/2$. For that special case, their strategy for Case~3 is to choose $(l', r')$ such that $\left|\frac{l' + r'}{2} - \frac{t + 2}{2}\right|$ is minimized. Intuitively, this makes sure that the median of the stream that has been observed so far is as close to the center of the length-$m$ array as possible. They analyzed the resulting random walk, and showed that  a memory of $m = \Theta(\sqrt{n})$ is sufficient. For the general $k$ case, analyzing the random walk associated with the analogue of the algorithm of~\cite{MP80} becomes more difficult.\footnote{As~\cite{MP80} remarked in their work, this random walk is ``difficult to analyze exactly since the transition probabilities vary with [the value of $(l, r)$] and with time''.}

\paragraph{The actual algorithm.} Our algorithm behind \Cref{thm:quantile-exact} can be viewed as a \emph{staged} strategy for the control problem of~\cite{MP80}, which admits a slightly simpler analysis.

The algorithm is divided into $\approx \log_2 k$ stages. Roughly speaking, the goal of each stage~$i \in \{0, 1, \ldots, \log_2 k\}$ is to make sure that, after reading the first $(n/k)\cdot 2^i$ elements, we have $l \le 2^i \le r$. What happens when we go from stage~$i$ to stage~$i+1$? Suppose that, at the end of stage~$i$, the $2^i$-th largest element so far (denoted by $x^*$) is exactly in the middle of the array, i.e., $2^i = \frac{l + r}{2}$. As we read the additional elements in stage~$(i+1)$, we maintain the $m$ consecutive elements, such that $x^*$ is kept in the middle of the array.

By a concentration argument, after stage~$i+1$, the rank of $x^*$ is bounded between $2^{i+1} \pm O(\sqrt{2^i}) = 2^{i+1} \pm O(\sqrt{k})$. Then, as long as $m \gg \sqrt{k}$, the actual element with rank $2^{i+1}$, denoted by $\tilde x$, must be among the $m$ consecutive elements maintained by the algorithm. Then, our algorithm shortens the length-$m$ array that contains consecutive elements, so that $\tilde x$ becomes the center of the array. Here, we deviate from the plan outlined earlier, as the value of $r - l + 1$ can drop below $m$ from time to time. Fortunately, by another concentration argument, as we go from stage~$i+1$ to stage~$i+2$, we are going to ``absorb'' additional elements in to the array, so that the length of the array returns to $m$ before we shorten the array again.

The proof sketch above can be formalized into a high-probability guarantee, which states that, as long as $m = \Omega(\sqrt{k})$, the algorithm proceeds in the hoped-for manner except with a tiny probability.

\paragraph{Another perspective.} We remark that our algorithm can alternatively be viewed as a remedy of the ``na\"ive reduction'' from \Cref{sec:overview-approx}. In \Cref{eq:dynamics-naive}, by reducing to the $k' = k/2$ case, the error doubles and increases by an $O(\sqrt{k})$ amount, resulting in a trivial error bound of $O(k)$. In the algorithm outlined above, however, by utilizing the strategy of~\cite{MP80}, we can offset the error by $O(m)$ at each iteration. In particular, as long as $m \gg \sqrt{k}$, the errors accumulated in all stages can be canceled out. This leads to our exact selection guarantee.

\subsection{From Quantile Estimation to Secretary Problem}\label{sec:overview-reduction}
We sketch how an accurate quantile estimator leads to a competitive $k$-secretary algorithm via a reduction similar to the algorithm of~\cite{Kleinberg05}. Kleinberg's algorithm is a recursive one: to solve the $k$-secretary problem on a length-$n$ stream, we divide the stream into two halves, each of length $\approx n/2$. We run the algorithm recursively for the $(k/2)$-secretary instance formed by the first half. In parallel, we find the $(k/2)$-th largest element, denoted by $x^*$, among the first half. Then, we read the second half of the sequence, and accept every element that is larger than $x^*$, until at least $k/2$ elements among the second half have been accepted.

To gain some intuition about the $1 - O(1/\sqrt{k})$ competitive ratio, consider a $k$-secretary instance that consists of $k$ copies of $1$ and $n - k$ copies of $0$. To be exactly optimal, the algorithm needs to accept every single element of value $1$. When running Kleinberg's algorithm, however, it holds with probability $\Omega(1)$ that the second half of the stream only contains $k/2 - \Omega(\sqrt{k})$ copies of $1$, in which case we lose an $\Omega(1/\sqrt{k})$ term in the competitive ratio.

Towards proving \Cref{prop:reduction}, we apply Kleinberg's algorithm with the quantile estimation subroutine replaced by the (hypothetical) memory-efficient algorithm (denoted by $\A$) with an expected error of $O(k^{\alpha})$. Then, we revisit the instance considered before. For clarity, we break ties among the $k$ copies of $1$ by replacing them with $1 - \eps, 1 - 2\eps, \ldots, 1 - k\eps$ for some tiny value $\eps \ll 1/k$. Then, the first half of the sequence contains a random subset of $\{1 - i\eps: i \in [k]\}$. When we call the quantile estimator $\A$ to find the $k/2$-th largest element among the first half, it might hold with probability $\Omega(1)$ that $\A$ returns an element $x^* = 1 - i^*\eps$ for some $i^* = k - \Omega(k^{\alpha})$. In this case, the algorithm might reject $\Omega(k^{\alpha})$ non-zero elements among the second half by mistake, resulting in a competitive ratio of at best $1 - \Omega(1/k^{1-\alpha})$.

While the argument above suggests why \Cref{prop:reduction} gives the ``right'' dependence of the competitive ratio on $\alpha$, the actual proof is slightly more complicated, since we need to argue that the impact of the error in quantile estimator on the competitive ratio is smooth, so that an upper bound on the \emph{expected error} also translates into a competitive ratio (after considering the randomness in quantile estimator).

The reduction outlined above would not give the desired memory bound of $m + O(1)$ in \Cref{prop:reduction}. Kleinberg's algorithm involves running $\approx \log_2 k$ copies of the quantile estimator: roughly speaking, the $i$-th copy ($i \le \log_2 k$) is for finding the $k/2^i$-th largest element among $s_{1:(n/2^i)}$. Since these copies run on overlapping prefixes of the stream, they must run in parallel, and thus increasing the memory usage by a factor of $\log k$. The actual memory bound $m + O(1)$ is obtained from a slight modification to Kleinberg's algorithm: instead of finding the $(k/2)$-th largest element among $s_{1:(n/2)}$ and using it as the threshold $x^*$ for the second half, we choose $x^*$ as the $(k/4)$-th largest element among $s_{(n/4+1):(n/2)}$. This change only perturbs the rank of $x^*$ by $O(\sqrt{k})$ in expectation, and would be dominated by the $k^{\alpha}$ error in the quantile estimator $\A$. As a result, we would run $\log k$ copies of $\A$ on \emph{disjoint} intervals in $s$: $s_{(n/4+1):(n/2)}, s_{(n/8+1):(n/4)}, \ldots$, so it suffices to allocate $m$ words of memory for all the $\log k$ calls to $\A$, and we only use an $O(1)$ extra memory.

\section{Quantile Estimation in Logarithmic Memory}\label{section.2}
In this section, we prove \Cref{thm:quantile-approx} by giving an algorithm that uses $O(\log k)$ memory and finds the approximately $k$-th largest element in a random-order stream of $n$ elements, up to an $O(\sqrt{k})$ error in the rank in expectation.

\subsection{The Algorithm}

Our main algorithm (\Cref{algorithm.main}) calls a recursive procedure $\findkth$ (\Cref{algorithm.recurse}) with parameters $n$, $m$, $k$ and $+\infty$ to find the $k$-th largest element (among those that are smaller than $+\infty$) in the next $n$ elements. We sketch how \Cref{algorithm.recurse} works in the following, in the special case that $a = +\infty$ (i.e., the first level of the recursion):
\begin{itemize}
    \item First, it samples $B \sim \Binomial\left(n, 1/2\right)$ and divides the $n$ upcoming elements (denoted by $s_1, s_2, \ldots, s_n$) into two halves: $s_{1:B}$ and $s_{(B+1):n}$.
    \item Then, among the first half, it further sub-samples a small fraction of $B_1 \sim \Binomial\left(B, \frac{2m}{3k}\right)$ elements. The hope is to find an element $a'$ such that its rank among $s_{1:B}$ is slightly below $k/2$. (Concretely, $\rank_{1:B}(a'; a) = \lfloor k/2\rfloor - k'$ for some $k'$ between $1$ and $\delta \coloneqq C_0\cdot k$.)
    \item At this point, we know that the $(k/2)$-th largest element among $s_{1:B}$ is the $k'$-th largest element among $s_{1:B} \cap (-\infty, a')$.
    \item Finally, we call \Cref{algorithm.recurse} recursively with parameters $n' = n - B$, $k'$ and $a'$ to find the $k'$-th largest element among $s_{(B+1):n} \cap (-\infty, a')$, in the hope that it will be approximately the overall $k$-th largest among $s_{1:n}$.
\end{itemize}

In the more general case that $a \ne +\infty$, the algorithm essentially does the same thing, except that all elements that are larger than or equal to $a$ are ignored.

\begin{algorithm2e}
    \caption{Main Algorithm} \label{algorithm.main}
    \KwIn{Stream length $n$, memory size $m$, target rank $k$.}
    \KwOut{An approximately $k$-th largest element among the $n$ elements.}
    Call $\findkth(n, m, k, +\infty)$ in \Cref{algorithm.recurse}\;
\end{algorithm2e}

\begin{algorithm2e}
    \caption{$\findkth(n, m, k, a)$}
    \label{algorithm.recurse}
    \KwIn{Stream length $n$, memory size $m$, target rank $k$, element threshold $a$, and access to random-order sequence $s = (s_1, s_2, \ldots, s_n)$}
    \KwOut{The (approximately) $k$-th largest element among $s \cap (-\infty, a)$, or $-\infty$ if $|s \cap (-\infty, a)| < k$}
    Use $O(1)$ memory to keep track of: (1) the smallest element, denoted by $\underline{x} = \min\{s_1, s_2, \ldots, s_n\}$; (2) the number of elements that are strictly smaller than $a$, denoted by $n' = |s \cap (-\infty, a)|$\;
    \uIf{$k \le m$} {
        Use $O(m)$ memory to find the $k$ largest elements among $s_{1:n} \cap (-\infty, a)$\;\label{line:naive} 
        \Return the $k$-th largest element, or $-\infty$ if $n' < k$\;
    }
    Draw $B \sim \Binomial\left(n, \frac{1}{2}\right)$ and $B_1 \sim \Binomial\left(B, \frac{2m}{3k}\right)$\; \label{line:draw-B-1}
    Use a length-$m$ array $M$ to find the $m$ largest elements among $s_{1:B_1} \cap (-\infty, a)$\; \label{line:first-half}
    Read until the $B$-th element in the string to compute the value of $\rank_{1:B}(x; a)$ for each element $x$ in array $M$\;
    \uIf{$|s_{1:B} \cap (-\infty, a)| < \lfloor k/2\rfloor$}{
        Read the remaining $n - B$ elements to compute $\underline{x}$ and $n'$\;
        \Return $\underline{x}$, or $-\infty$ if $n' < k$\; \label{line.case1}
    }
    $\delta \leftarrow C_0\cdot k$, where $C_0$ is a sufficiently small constant parameter\;
    Find the largest element $a'$ in $M$ such that $\rank_{1:B}\left(a'; a\right) \in [\lfloor k/2 \rfloor - \delta, \lfloor k/2 \rfloor - 1]$\;\label{line:find-a'}
    \uIf{no such element $a'$ exists}{
        Read the remaining $n - B$ elements\;
        \Return the largest element in the array below $a$, or $-\infty$ if $n' < k$\;\label{line:no-a'}
    }
    $x \gets \findkth\left(n - B, m, \left\lfloor k / 2\right\rfloor - \rank_{1:B}\left(a'; a\right), a'\right)$\;\label{line:recursive-call}
    \lIf{$n' < k$}{\Return $-\infty$}
    \lIf{$x = -\infty$}{\Return $\underline{x}$}\label{line:translate-infinity-to-smallest}
    \Return $x$\;\label{line.end}
\end{algorithm2e}

\paragraph{Comparison-based algorithms.} We note that our algorithm is \emph{comparison-based} in the sense that it can be implemented such that the algorithm only accesses the elements in the sequence in the following three ways: (1) compare a pair of elements; (2) return an element as output; (3) pass an element as a parameter of a recursive call. One desirable property of such comparison-based algorithms is that, roughly speaking, the output of the algorithm only depends on the \emph{relative ordering} of the elements, rather than their exact identities. We formalize this property and use it as the definition of comparison-based algorithms in the following.

\begin{definition}[Comparison-based algorithms]\label{def:comparison-based}
    A quantile estimation algorithm $\A$ is comparison-based if, for any $n$ and $k$, when finding the $k$-th largest element in a random-order sequence $s$ of $n$ distinct elements, the distribution of
    \[
        \rank_{1:n}(\A(n, k, s))
    \]
    is the same regardless of the choice of the $n$ elements in $s$.

    More generally, suppose that the quantile estimation problem has an additional threshold parameter $a$. An algorithm $\A$ is comparison-based if, as long as $a \notin \{s_1, s_2, \ldots, s_n\}$, the distribution of
    \[
        \rank_{1:n}(\A(n, k, a, s); a)
    \]
    only depends on $n$, $k$, and $\rank_{1:n}(a)$, and does not depend on $a$ and the $n$ elements in $s$.
\end{definition}

\subsection{Overview of the Analysis}\label{sec:approx-analysis-overview}
The rest of this section is devoted to the analysis of \Cref{algorithm.recurse}. We start by observing the behavior of $\findkth$ in several different edge cases. We then sketch how we analyze the randomness in both the random-order stream and the algorithm, so that the randomness in different parts can be significantly decoupled.

\paragraph{Corner cases.} Let $n' \coloneqq |\{s_1, s_2, \ldots, s_n\} \cap (-\infty, a)|$ denote the number of elements that are strictly smaller than the threshold $a$. When $n' < k$, the quantile estimation problem is ill-defined, in which case $\findkth$ always returns $-\infty$. Note that the $n' < k$ case has the highest priority among all edge cases in the sense that, before the algorithm tries to output anything (possibly as a result of handling other corner cases), it makes sure to read through the end of the stream $s$ to check whether $n' < k$ (and outputs $-\infty$ if so).

On Line~\ref{line:naive}, we choose to use the straightforward algorithm for quantile estimation when $k \le m$. This serves as the boundary condition of the recursion, and also ensures that the parameter $\frac{2m}{3k}$ on Line~\ref{line:draw-B-1} is indeed in $[0, 1]$, and thus valid.

On Line~\ref{line.case1}, if it turns out that $|s_{1:B} \cap (-\infty, a)| < \lfloor k/2\rfloor$, we return the smallest element in the entire stream. To see why we do this, recall that there are $n'$ elements in $s_{1:n} \cap (-\infty, a)$. If we are not in the edge case that $n' < k$, we should expect that there are $\approx n' / 2 \ge k / 2$ elements in $s_{1:B} \cap (-\infty, a)$, since $s_{1:B}$, as a set, is uniformly distributed among all subsets of $\{s_1, s_2, \ldots, s_n\}$. Therefore, whenever the edge case $|s_{1:B} \cap (-\infty, a)| < \lfloor k/2\rfloor$ happens, we are sure that $n'$ only exceeds $k$ by a small amount, in which case outputting the $n'$-th largest element in $s \cap (-\infty, a)$ (namely, the smallest element in $s$) is accurate enough.

When the above does not happen, we aim to find an element $a'$ in the first half of the stream, $s_{1:B}$, so that $a'$ is close to the $\lfloor k/2\rfloor$-th largest among $s_{1:B} \cap (-\infty, a)$. On Line~\ref{line:no-a'}, if we fail to find such an $a'$, we return the largest element among $s_{1:n} \cap (-\infty, a)$, which has a rank of $1$. While doing so leads to an error of $k - 1$ in the rank, as we show in \Cref{lemma:prob-find-a'}, this edge case only happens with probability $1/\poly(k)$, so the contribution to the expected error is negligible.

If none of the edge cases above happen, $\findkth$ calls itself recursively and gets an output $x$. The final edge case is that this $x$ might take value $-\infty$, as a result of encountering the first edge case in the recursive call. If this happens, we translate $-\infty$ to the smallest element $\underline{x}$ in the stream.

\paragraph{Unravel the randomness.} We will prove the error bound of $\findkth$ by induction. As mentioned in \Cref{sec:overview-details}, a rigorous analysis is complicated because, for the inductive step, we need to show that the accuracy of the recursive call
\[
    x^* \gets \findkth(n - B, m, \lfloor k/2\rfloor - \rank_{1:B}(a'; a), a')
\]
(conditioning on all the parameters) implies the accuracy of the original procedure. For this purpose, we might hope that $s_{(B+1):n}$ (as a set) is uniformly distributed among all size-$(n-B)$ subsets of $s_{1:n}$, so that we can translate the rank of $x^*$ among $s_{(B+1):n}$ to its rank among the entire stream $s_{1:n}$. Unfortunately, this uniformity might not hold, since the conditioning on $a'$ and $\rank_{1:B}(a'; a)$ would bias the conditional distribution of $s_{(B+1):n}$.

To make the analysis valid, we need to ``realize'' the randomness of $s_{1:n}$ in a specific order. In the following, we sketch the analysis assuming that none of the edge cases above happen:
\begin{itemize}
    \item First, over the randomness in $B \sim \Binomial(n, 1/2)$, $s_{1:B}$ is uniformly distributed among all subsets of $s_{1:n}$ (\Cref{lemma:subsample}). In particular, $s_{1:B} \cap (-\infty, a)$ is a uniformly random subset of $s_{1:n} \cap (-\infty, a)$. It follows that: (1) $B' \coloneqq |s_{1:B} \cap (-\infty, a)|$ follows the distribution $\Binomial(n', 1/2)$, where $n' \coloneqq |s_{1:n} \cap (-\infty, a)|$; (2) Conditioning on the realization of $B'$, $s_{1:B} \cap (-\infty, a)$ is uniformly distributed among all size-$B'$ subsets of the size-$n'$ set $s_{1:n} \cap (-\infty, a)$. 
    \item We condition on the realization of $B'$. The algorithm finds an element $a' \in s_{1:B}$ on Line~\ref{line:find-a'}. Let $i \coloneqq \rank_{1:B}(a'; a)$ denote its rank among the first half (when only elements below $a$ are considered). Note that the distribution of $i$ is solely determined by the value of $B'$, since the algorithm is comparison-based (\Cref{def:comparison-based}).
    \item Conditioning on the realization of $(B', i)$, $a'$ is identically distributed as the $i$-th largest element in a uniformly random size-$B'$ subset of $s_{1:n} \cap (-\infty, a)$. Let $i_1 \coloneqq \rank_{1:n}(a'; a)$ denote the rank of $a'$. Later, we can analyze the concentration of $i_1 \mid B', i$ using \Cref{lemma:subset-rank-concentration}. 
    \item Conditioning on the realization of $(B', i, i_1)$, the algorithm makes a recursive call (on Line~\ref{line:recursive-call}) to find the $k'$-th largest among $s_{(B+1):n} \cap (-\infty, a')$, where $k' \coloneqq \lfloor k/2\rfloor - i$. Let $x^*$ denote the element returned by the recursive call, and $l \coloneqq \rank_{B+1:n}(x^*; a')$ be its actual rank. During the inductive proof, we will use the induction hypothesis that $l \mid B', i, i_1$ concentrates around $k'$.
    \item Conditioning on $(B', i, i_1, l)$, $x^*$ is the $l$-th largest element among $s_{(B+1):n} \cap (-\infty, a')$. Furthermore, we can verify that $s_{(B+1):n} \cap (-\infty, a')$ is uniformly distributed among all size-$(n - i_1 - (B - i))$ subsets of $s_{1:n} \cap (-\infty, a')$. Applying \Cref{lemma:subset-rank-concentration} again shows that $\rank_{1:n}(x^*; a')$ concentrates, which in turn implies the concentration of $\rank_{1:n}(x^*; a) = \rank_{1:n}(x^*; a') + i_1$.
    \item Finally, our goal is to show that the rank of $x^*$, $\rank_{1:n}(x^*; a)$, concentrates around $k$. For this purpose, we will start with the conditional concentration bound above, and take an expectation over the joint distribution of $(B', i, i_1, l)$.
\end{itemize}

\subsection{The Analysis}

We will frequently use the following simple fact, which we prove in \Cref{sec:technical-lemmas}: a random prefix of a random-order stream is distributed as an independently subsampled subset.
\begin{restatable}{lemma}{lemmasubsample}\label{lemma:subsample}
    Suppose that $s_1, s_2, \ldots, s_n$ is a uniformly random permutation of a size-$n$ set $S$, and $p \in [0, 1]$. Then, for $B$ independently drawn from $\Binomial\left(n, p\right)$, $\{s_1, s_2, \ldots, s_B\}$ is distributed as a random subset of $S$ that includes every element with probability $p$ independently.
\end{restatable}

Ideally, we want there to be an element $a'$ such that the rank of it among $x_{1:B'}$ is in the range of $\left[\left\lfloor k/2 \right\rfloor - \delta, \left\lfloor k/2 \right\rfloor - 1\right]$ to conduct the recursion call on Line~\ref{line:recursive-call}. Recall that $n' \coloneqq |s_{1:n} \cap (-\infty, a)|$ denotes the number of elements in $s_{1:n}$ that are strictly smaller than $a$. The following lemma shows that, in a call to procedure $\findkth(n, m, k, a)$ with $n' \ge k$, it holds with high probability that we successfully find an element $a'$ on Line~\ref{line:find-a'}, unless the first half of the sequence, $s_{1:B}$, contains strictly fewer than $\lfloor k/2 \rfloor$ elements below $a$.

\begin{lemma}\label{lemma:prob-find-a'}
    Consider a call to the procedure $\findkth(n, m, k, a)$ that satisfies $n' \ge k \ge 2$. With probability at least 
    \[
        1 - \exp\left(-\frac{2}{9}\cdot m\right) - \exp\left(1-\frac{2C_0}{3}\cdot m\right),
    \]
    one of the following two is true:
    \begin{enumerate}
        \item $B' \coloneqq |s_{1:B} \cap (-\infty, a)| < \lfloor k / 2\rfloor$.
        \item On Line~\ref{line:find-a'}, the algorithm successfully finds an element $a'$ in the length-$m$ array such that $\rank_{1:B}(a'; a) \in [\lfloor k /2 \rfloor - C_0k, \lfloor k /2 \rfloor - 1]$.
    \end{enumerate}
\end{lemma}

\begin{proof}
    It is sufficient to prove the following: Conditioning on the realization of $B' \ge \lfloor k/2 \rfloor$ and the set $s_{1:B'} \cap (-\infty, a)$ (but not their order), the second condition (that we successfully find $a'$) holds with probability at least
    \[
        1 - \exp\left(-\frac{2}{9}\cdot m\right) - \exp\left(1-\frac{2C_0}{3}\cdot m\right).
    \]
    The lemma would then follow from the law of total probability.

    Conditioning on the value of $B'$ and the size-$B'$ set $s_{1:B} \cap (-\infty, a)$, the ordering in which these $B'$ elements appear in $s$ is still uniformly distributed. Then, since $B_1$ is sampled independently from $\Binomial(B, \frac{2m}{3k})$, by \Cref{lemma:subsample}, the elements that we examine on Line~\ref{line:first-half} constitute a random subset of $s_{1:B} \cap (-\infty, a)$ that includes every element with probability $\frac{2m}{3k}$ independently.

    Formally, suppose that the $B'$ elements in $s_{1:B} \cap (-\infty, a)$, when sorted in descending orders, are $x_1 > x_2 > \cdots > x_{B'}$. For each $i \in [B']$, let binary random variable $Y_i \coloneqq \1{x_i \in \{s_1, s_2, \ldots, s_{B_1}\}}$ denote whether $x_i$ is among the first $B_1$ elements (and thus examined on Line~\ref{line:first-half}). Then, as we discussed above, $Y_1$ through $Y_{B'}$ are independently distributed as $\Bern(\frac{2m}{3k})$.

    Then, we re-write the event that the algorithm fails to find a good $a'$ (on Line~\ref{line:find-a'}) as a condition regarding $Y_1$ through $Y_{B'}$. We claim that this can happen only if at least one of the following two events happen:
    \begin{itemize}
        \item Event $\event_a$: $Y_1 + Y_2 + \cdots + Y_{\lfloor k/2\rfloor} > m$.
        \item Event $\event_b$: $Y_i = 0$ holds for every integer $i \in [\lfloor k/2\rfloor - C_0k, \lfloor k/2\rfloor - 1]$.
    \end{itemize}
    To see this, we first note that the negation of $\event_b$ implies that there exists $i \in [\lfloor k/2\rfloor - C_0k, \lfloor k/2\rfloor - 1]$ such that $x_i$ is among $s_{1:B_1}$. Then, by the negation of $\event_a$, $x_i$ must be among the top $m$ elements that get examined on Line~\ref{line:first-half}. Then, $x_i$ would be a valid choice for $a'$ on Line~\ref{line:find-a'}.

    Since $Y_1 + Y_2 + \cdots + Y_{\lfloor k/2 \rfloor}$ follows $\Binomial(\lfloor k/2 \rfloor, \frac{2m}{3k})$, a Chernoff bound (relative error form) gives
    \begin{align*}
            \pr{}{\event_a}
    &=      \pr{X\sim \Binomial\left(\lfloor k/2\rfloor, \frac{2m}{3k}\right)}{X \ge m}\\
    &\le    \pr{X\sim \Binomial\left(\lfloor k/2\rfloor, \frac{2m}{3k}\right)}{X \ge 3\Ex{}{X}}\\
    &\le    e^{-\Ex{}{X}}
    =       \exp\left(-\left\lfloor\frac{k}{2}\right\rfloor\cdot\frac{2m}{3k}\right)
    \le     e^{-2m/9}. \tag{$\lfloor k/2 \rfloor \ge k/3$ for all $k \ge 2$}
    \end{align*}
    For the other event $\event_b$, since there are $\lfloor C_0 k\rfloor$ integers in $[\lfloor k/2 \rfloor - C_0 k, \lfloor k/2 \rfloor - 1]$, we have
    \[
        \pr{}{\event_b}
    =   \left(1 - \frac{2m}{3k}\right)^{\lfloor C_0 k\rfloor}
    \le \exp\left(-\frac{2m}{3k}\cdot \lfloor C_0 k\rfloor\right)
    \le \exp\left(-\frac{2m}{3k}\cdot (C_0 k - 1)\right)
    \le e^{1-(2C_0 / 3)\cdot m}.
    \]
    The lemma then follows from a union bound.
\end{proof}
  
The following lemma, which we prove in \Cref{sec:technical-lemmas}, considers a random size-$k$ subset of $[n]$, and shows how well the $i$-th smallest element in the subset concentrates.

\begin{restatable}{lemma}{subsetrankconcentration}\label{lemma:subset-rank-concentration}
    Suppose that $n \ge k \ge i \ge 1$. Let $x$ be the $i$-th smallest element in a size-$k$ subset of $\{1, 2, \ldots, n\}$ chosen uniformly at random. Then,
    \[
        \Ex{}{\left|x - i\cdot \frac{n}{k}\right|} \le 2\cdot \sqrt{i\cdot \frac{n}{k}}\cdot \frac{n}{k}+\frac{n^2}{k^2}.
    \]
\end{restatable}

Recall the overall plan of our analysis from \Cref{sec:approx-analysis-overview}. The following lemma shows that, conditioning on the values of $(B', i, i_1, l)$, the $l$-th largest element among $s_{(B+1):n} \cap (-\infty, a')$, denoted by $x^*$, satisfies a certain concentration of $\rank_{1:n}(x^*; a)$. To simplify the notation, the lemma is stated for the special case that $a = +\infty$, so that $n'$ and $B'$ coincide with $n$ and $B$.

\begin{lemma}\label{lemma:rank-in-second-half-to-overall}
    Suppose that $s = (s_1, s_2, \ldots, s_n)$ is a uniformly random permutation of elements $x_1 > x_2 > \cdots > x_n$. Then, the following holds for all integers $B$, $i$, and $i_1$ that satisfy $1 \le i \le B \le n$ and $i \le i_1 \le i + (n - B)$: Let $\event$ denote the event that $x_{i_1}$ is the $i$-th largest element among $s_{1:B}$. Conditioning on event $\event$, for any $l \in \{1, 2, \ldots, (n - B) - (i_1 - i)\}$, the $l$-th largest element among $s_{(B+1):n} \cap (-\infty, x_{i_1})$, denoted by $x^*$, satisfies
    \begin{align*}
        &~\Ex{}{\left|\rank_{{1:n}}(x^*) - \left(i_1+l\cdot \frac{n-i_1}{(n - i_1) - (B-i)}\right)\right|}\\
    \le &~2 \cdot \left(\frac{n-i_1}{n-i_1-(B-i)}\right)\sqrt{l\cdot \frac{n-i_1}{(n-i_1)-(B-i)}} + \left(\frac{n-i_1}{(n-i_1)-(B-i)}\right)^2,
    \end{align*}
    where the expectation is over the randomness in $s_1, s_2, \ldots, s_n$ after conditioning on $\event$.
\end{lemma}
\begin{proof}
    Let $a' \coloneqq x_{i_1}$. Since $a'$ is the $i$-th largest element in $s_{1:B}$, exactly $B - i$ elements in $s_{1:B}$ are strictly smaller than $a'$. Among $s_{1:n}$, exactly $n - i_1$ elements (namely, $x_{i_1 + 1}$ through $x_n$) are strictly smaller than $a'$. Therefore, we have
    \[
        |s_{(B+1):n} \cap (-\infty, a')|
    =   |s_{1:n} \cap (-\infty, a')| - |s_{1:B} \cap (-\infty, a')|
    =   (n - i_1) - (B - i).
    \]
    In other words, exactly $(n - i_1) - (B - i)$ elements in $s_{(B+1):n}$ are strictly smaller than $a'$.

    Without the conditioning on event $\event$, the elements in $s_{(B + 1): n}$ form a size-$(n-B)$ subset of $s_{1:n}$ chosen uniformly at random. Then, after conditioning on $\event$, $s_{(B+1):n} \cap (-\infty, a')$ is a uniformly random subset of $\{x_{i_1+1}, \ldots, x_n\}$ of size $(n - i_1) - (B - i)$. Applying \Cref{lemma:subset-rank-concentration} with parameters
    \[
        \tilde n = n - i_1, \quad \tilde k = (n - i_1) - (B - i), \quad \tilde i = l
    \]
    shows that
    \begin{align*}
        &~\Ex{}{\left|\rank_{{1:n}}(x^*; a') - l\cdot \frac{n-i_1}{(n-B)-(i_1-i)}\right|}\\
    \le &~2 \cdot \left(\frac{n-i_1}{n-i_1-(B-i)}\right)\sqrt{l\cdot \frac{n-i_1}{n-i_1-(B-i)}} + \left(\frac{n-i_1}{n-i_1-(B-i)}\right)^2,
    \end{align*}
    where $x^*$ is the unique element in $s_{(B+1):n}$ that satisfies $\rank_{B+1:n}(x^*; a') = l$ (namely, the $l$-th largest element among $s_{(B+1):n} \cap (-\infty, a)$).
   
    Finally, the lemma follows from the observation that, for any $x^* < a' = x_{i_1}$,
    \[
        \rank_{1:n}(x^*)
    =   \rank_{1:n}(a') + \rank_{1:n}(x^*; a')
    =   \rank_{1:n}(x^*; a') + i_1.
    \]
\end{proof}

The following technical lemma shows how \Cref{lemma:rank-in-second-half-to-overall} is going to be applied in the analysis of \Cref{algorithm.recurse}: Conditioning on any realization of $B' \in (n'/4, 3n'/4)$ and $i \coloneqq \rank_{1:B}(a'; a)$, over the remaining randomness in $i_1\coloneqq \rank_{1:n}(a'; a)$ and $l\coloneqq \rank_{B+1:n}(x^*; a')$, the outer procedure is accurate assuming that the recursive call is accurate enough. In particular, \Cref{lemma:rank-in-second-half-to-overall} is used to translate the concentration of $l$ into the concentration of $\rank_{1:n}(x^*; a)$.

\begin{lemma}\label{lemma:rank-in-second-half-to-overall-simplified}
    Consider a call to the procedure $\findkth(n, m, k, a)$ on a random-order stream $s_1, s_2, \ldots, s_n$. Let $n' \coloneqq |s_{1:n} \cap (-\infty, a)|$ and $B' \coloneqq |s_{1:B} \cap (-\infty, a)|$ denote the numbers of elements that are strictly smaller than the threshold $a$ among $s_{1:n}$ and $s_{1:B}$, respectively. Let $i \coloneqq \rank_{1:B}(a'; a)$ denote the rank of element $a'$ (defined on Line~\ref{line:find-a'}) among $s_{1:B} \cap (-\infty, a)$, and $i_1 \coloneqq \rank_{1:n}(a'; a)$ be its rank among $s_{1:n} \cap (-\infty, a)$.
    
    Suppose that: (1) $B' \in (n'/4, 3n'/4)$; (2) conditioning on the values of $B'$, $i$, and $i_1$, the element $x^* \in s_{(B+1):n}$ returned by the recursive call (on Line~\ref{line:recursive-call}) satisfies
    \[
        \Ex{}{\left|\rank_{{B+1:n}}(x^*; a')- k'\right|} \le C_2\sqrt{k'},
    \]
    where $k' \coloneqq \lfloor k/2\rfloor - i \ge 1$ and $C_2 = 262$. Then, we have
    \[
    \Ex{}{\left|\rank_{{1:n}}(x^*; a) - \left(i_1+l\cdot \frac{n'-i_1}{n'-i_1-(B'-i)}\right)\right|} \le C_1\cdot\sqrt{C_2}\cdot\sqrt{k'},
    \]
    where $C_1 = 914$, and the expectation is over the randomness in the ordering of $s_1, s_2, \ldots, s_n$ after the conditioning on $B'$ and $i$.
\end{lemma}
\begin{proof}
Recall that $B' \coloneqq |s_{1:B} \cap (-\infty, a)|$, $i \coloneqq \rank_{1:B}(a'; a)$ and $i_1 \coloneqq \rank_{1:n}(a'; a)$. In addition, we introduce the shorthands
\[
    j \coloneqq B'-i+1,
\quad
    j_1 \coloneqq n'-i_1 +1.
\]
Then, $a'$ is both the $j$-th smallest element in $s_{1:B} \cap (-\infty, a)$ and the $j_1$-th smallest element in $s_{1:n} \cap (-\infty, a)$. It follows that, conditioning on the values of $B'$ and $i$ (and thus $j$), the value of $j_1$ is identically distributed as the rank of the $j$-th smallest number in a uniformly random size-$B'$ subset of $\{1, 2, \ldots, n'\}$, and thus can be analyzed in a similar way to the analysis in \Cref{lemma:subset-rank-concentration}.

In the rest of the proof, we will upper bound the expectation by considering two cases separately: $j_1 \ge (1+\eps)\cdot j$ and $j_1 < (1+\eps)\cdot j$, where $\eps = 1/16$.
\paragraph{Case 1: $j_1 \ge (1+\eps)\cdot j$.}
In this case, we have
\[
    \frac{B'-i}{n'-i_1-(B'-i)} = \frac{j-1}{j_1 - j} \le \frac{j-1}{(1+\eps)\cdot j - j} \le \frac{1}{\eps}.
\]
It follows that
\[
    \frac{n' - i_1}{n' - i_1 - (B' - i)}
=   \frac{B' - i}{n' - i_1 - (B' - i)} + 1
\le \frac{1}{\eps} + 1
\le \frac{2}{\eps}.
\]
Then, by \Cref{lemma:rank-in-second-half-to-overall}, we have
 \begin{equation}\begin{split}\label{eq:case-1-intermediate-bound}
    &~\Ex{i_1, l}{\1{j_1 \ge (1+\eps)\cdot j}\cdot\left|\rank_{{1:n}}(x^*; a) - \left(i_1+l\cdot \frac{n'-i_1}{n'-i_1-(B'-i)}\right)\right|}\\
\le &~\Ex{i_1}{\1{j_1 \ge (1+\eps)\cdot j}\cdot\left(\frac{4\sqrt{2}}{\eps^{3/2}}\cdot \Ex{l}{\sqrt{l}\mid i_1} + \frac{4}{\eps^2}\right)}.
\end{split}\end{equation}
To deal with the term $\Ex{l}{\sqrt{l}\mid i_1}$, we note that, conditioning on the realization of $i_1$, 
\begin{equation}\label{eq:sqrt-l-bound}
    \Ex{l}{\sqrt{l} \mid i_1}
\le \sqrt{\Ex{}{l \mid i_1}}
\le \sqrt{k' + \Ex{}{|l - k'| \mid i_1}}
\le \sqrt{k' + C_2\sqrt{k'}}
\le \sqrt{2C_2}\cdot\sqrt{k'},
\end{equation}
where the first step applies Jensen's inequality, the second step applies the triangle inequality, and the third step follows from our assumption on the concentration of $l$.

Plugging the above into \Cref{eq:case-1-intermediate-bound} and setting $\eps = 1/16$ gives
\begin{equation}\begin{split}\label{eq:case-1-final-bound}
    &~\Ex{i_1, l}{\1{j_1 \ge (1+\eps)\cdot j}\cdot\left|\rank_{{1:n}}(x^*; a) - \left(i_1+l\cdot \frac{n'-i_1}{n'-i_1-(B'-i)}\right)\right|}\\
\le &~ \frac{4\sqrt{2}}{\eps^{3/2}}\cdot \sqrt{2C_2}\cdot\sqrt{k'} + \frac{4}{\eps^2}\\
=   &~512\sqrt{C_2}\cdot\sqrt{k'} + 1024
\le 768\sqrt{C_2}\cdot\sqrt{k'},
\end{split}\end{equation}
where the last step follows from $\sqrt{C_2} \cdot \sqrt{k'} \ge \sqrt{C_2} \ge 4$.

\paragraph{Case 2: $j_1 < (1+\eps)\cdot j$.} We use a similar method as in the proof of \Cref{lemma:subset-rank-concentration} to upper bound the probability of this case. As discussed earlier, this case corresponds to the event that, when a size-$B'$ subset $S \subseteq [n']$ is chosen uniformly at random, the $j$-th smallest element in $S$ is strictly smaller than $(1 + \eps)\cdot j$. This, in turn, implies the following event for $\tilde j \coloneqq \lceil (1 + \eps)\cdot j \rceil - 1 \le (1 + \eps)\cdot j$:
\begin{itemize}
    \item Event $\event_a$: $S$ contains at least $j$ elements among $1, 2, \ldots, \tilde j$.
\end{itemize}

To upper bound the probability of event $\event_a$, we use a similar method to the proof of \Cref{lemma:subset-rank-concentration}: Defining a sequence of binary random variables $X_1, \ldots, X_{n'}$ such that each $X_{t}$ takes value $1$ if and only if $t \in S$. Then, $X_1$ through $X_{n'}$ are sampled without replacement from a size-$n'$ population that consists of $B'$ copies of $1$ and $n' - B'$ copies of $0$. Using \Cref{lemma.cite}, we can relate the moment generating function of $\sum_t X_t$ to that of $\sum_t Y_t$, where each $Y_t$ is independently drawn from $\Bern(B'/n')$. It then follows from a Chernoff bound that
\begin{align*}
        \pr{}{\event_a}
&=      \pr{X}{\sum_{t=1}^{\tilde j}X_t \ge j}\\
&\le    \exp\left(-2\cdot\frac{\left(j - \tilde j\cdot\frac{B'}{n'}\right)^2}{\tilde j}\right)\\
&\le    \exp\left(-2\cdot\frac{\left(j - (1 + \eps)\cdot j\cdot\frac{B'}{n'}\right)^2}{(1 + \eps)\cdot j}\right).
\end{align*}
The third step above holds since $\tilde j \le (1 + \eps)j$ and the function $x \mapsto -\frac{(a - bx)^2}{x}$ is monotone increasing when $a \ge bx$; in this case, the assumption that $B' < 3n'/4$ indeed guarantees
\[
    a
=   j
\ge \frac{3}{4}\cdot (1 + 1/16)j
\ge \frac{B'}{n'}\cdot (1 + \eps)j
=   bx.
\]
Plugging $\eps = 1 / 16$ and $B' / n' < 3/4$ into the above gives the upper bound
\[
    \pr{}{\event_a}
\le \exp\left(-2\cdot\frac{\left(1 - (1 + 1/16)\cdot\frac{3}{4}\right)^2}{1 + 1/16}\cdot j\right)
\le \exp(-0.077\cdot j).
\]

Recall that $j_1 = n' - i_1 + 1$, $j = B' - i + 1$, and
\[
    \frac{n' - i_1}{n' - i_1 - (B' - i)}
=   \frac{B' - i}{n' - i_1 - (B' - i)} + 1
=   \frac{j - 1}{j_1 - j} + 1
\le j.
\]
By \Cref{lemma:rank-in-second-half-to-overall}, the contribution of this case to the overall expectation is
\begin{align*}
    &~\Ex{i_1, l}{\1{j_1 < (1+\eps)\cdot j}\cdot\left|\rank_{{1:n}}(x^*; a) - \left(i_1+l\cdot \frac{n'-i_1}{n'-i_1-(B'-i)}\right)\right|}\\
    \le &~\Ex{i_1}{\1{j_1 < (1+\eps)\cdot j}\cdot\left(2\cdot j^{3/2}\cdot \Ex{l}{\sqrt{l}\mid i_1} + j^2\right)}.
\end{align*}

Plugging the bound $\Ex{l}{\sqrt{l} \mid i_1} \le \sqrt{2C_2}\cdot\sqrt{k'}$ from \Cref{eq:sqrt-l-bound} into the above gives an upper bound of
\[
    (2\sqrt{2}\cdot j^{3/2}\cdot\sqrt{C_2}\cdot\sqrt{k'} + j^2)\cdot \pr{}{j_1 < (1 + \eps)\cdot j}
\le (2\sqrt{2}\cdot j^{3/2} + j^2)\cdot \pr{}{j_1 < (1 + \eps)\cdot j}\cdot \sqrt{C_2}\cdot\sqrt{k'}.
\]
Since we have shown that the $j_1 < (1 + \eps)\cdot j$ case happens with probability at most $e^{-0.077\cdot j}$, we conclude that
\begin{equation}\begin{split}\label{upper-bound.2}
    &~\Ex{i_1, l}{\1{j_1 < (1+\eps)\cdot j}\cdot\left|\rank_{{1:n}}(x^*; a) - \left(i_1+l\cdot \frac{n'-i_1}{n'-i_1-(B'-i)}\right)\right|}\\
\le &~(2\sqrt{2}\cdot j^{3/2} + j^2)\cdot e^{-0.077j}\cdot\sqrt{C_2}\cdot\sqrt{k'}\\
\le &~ 146\sqrt{C_2}\cdot\sqrt{k'},
\end{split}\end{equation}
where the last step follows from
\begin{align*}
    \max_{x \ge 0}\left(2\sqrt{2}\cdot x^{3/2} + x^2\right)\cdot e^{-0.077x}
&\le 2\sqrt{2}\cdot\max_{x\ge0}x^{3/2}e^{-0.077x} + \max_{x\ge0}x^2e^{-0.077x}\\
&=  2\sqrt{2}\cdot \left(\frac{3/2}{0.077}\right)^{3/2}\cdot e^{-3/2} + \left(\frac{2}{0.077}\right)^{2}\cdot e^{-2}
<   146.
\end{align*}

\paragraph{Combine two upper bounds.}
Finally we combine the contributions from both cases (\Cref{eq:case-1-final-bound,upper-bound.2}) and obtain
\begin{align*}
    &~\Ex{i_1, l}{\left|\rank_{{1:n}}(x^*; a) - \left(i_1+l\cdot \frac{n'-i_1}{n'-i_1-(B'-i)}\right)\right|}\\
=   &~\Ex{i_1, l}{\1{j_1 \ge (1+\eps)\cdot j}\cdot\left|\rank_{{1:n}}(x^*; a) - \left(i_1+l\cdot \frac{n'-i_1}{n'-i_1-(B'-i)}\right)\right|}\\
+   &~\Ex{i_1, l}{\1{j_1 < (1+\eps)\cdot j}\cdot\left|\rank_{{1:n}}(x^*; a) - \left(i_1+l\cdot \frac{n'-i_1}{n'-i_1-(B'-i)}\right)\right|}\\
\le &~768\sqrt{C_2}\cdot\sqrt{k'} + 146\sqrt{C_2}\cdot\sqrt{k'}
=   914\sqrt{C_2}\cdot\sqrt{k'}.
\end{align*}
\end{proof}

Recall that \Cref{lemma:rank-in-second-half-to-overall-simplified} analyzed the concentration of $\rank_{1:n}(x^*; a)$---the rank of the element $x^*$ returned by the recursive call on Line~\ref{line:recursive-call}. However, that concentration is conditional on the realization of $(B', i)$, and is around a quantity that depends on $B'$ and $i$. In the following lemma, we further analyze the randomness in $B'$ and $i$, and show a similar concentration bound around $k$, the rank that we aim for.

\begin{lemma}\label{lemma:inductive-step}
    Let $C_2 = 262$ be a universal constant. As in the setup of \Cref{lemma:rank-in-second-half-to-overall-simplified}, consider a call to the procedure $\findkth(n, m, k, a)$ on a random-order stream $s_1, s_2, \ldots, s_n$. Define $n' \coloneqq |s_{1:n} \cap (-\infty, a)|$, $B' \coloneqq |s_{1:B} \cap (-\infty, a)|$, $i \coloneqq \rank_{1:B}(a'; a)$ and $i_1 \coloneqq \rank_{1:n}(a'; a)$ in the same way. Suppose that, conditioning on the realization of $(B', i, i_1)$, the following is true for the recursive call on Line~\ref{line:recursive-call}: (1) If $|s_{(B+1):n} \cap (-\infty, a')| \ge k' \coloneqq \lfloor k/2 \rfloor - i$, it returns an element $x'$ that satisfies
    \[
        \Ex{}{\left|\rank_{B+1:n}(x'; a') - k'\right|} \le C_2 \cdot \sqrt{k'};
    \]
    (2) If $|s_{(B+1):n} \cap (-\infty, a')| < k'$, the recursive call returns $-\infty$.
    
    Then, as long as $n' \ge k$, the current call to $\findkth$ returns an element $x^*$ such that
    \[
        \Ex{}{\left|\rank_{{1:n}}(x^*; a) - k\right|\cdot \1{\goodevent}} \le (C_2-9)\cdot \sqrt{k},
    \]
    where the expectation is over the (unconditional) randomness in $s_1, s_2, \ldots, s_n$ (and thus in the realization of $(B', i, i_1)$ and the recursive call), and $\goodevent$ is the event that the following two both hold: (1) $B' \ge \lfloor k/2\rfloor$; (2) $a'$ is successfully found on Line~\ref{line:find-a'}.
\end{lemma}

The proof uses the following lemma, which we prove in \Cref{sec:technical-lemmas}: If $B$ follows $\Binomial(n, 1/2)$, the ratio $n/B$ concentrates around $2$ up to an error of $O(1/\sqrt{n})$ in expectation (ignoring the pathological case of $B = 0$).
\begin{restatable}{lemma}{expectednoverb}\label{lemma.integral1}
    It holds for every integer $n \ge 1$ that
    \[
        \Ex{B \sim \Binomial(n, 1/2)}{\left|\frac{n}{B} - 2\right|\cdot\1{B \ne 0}}  \le \frac{14}{\sqrt{n}},
    \]
    assuming that $\left|n/B - 2\right|\cdot\1{B \ne 0}$ evaluates to $0$ when $B = 0$.
\end{restatable}

\begin{proof}[Proof of \Cref{lemma:inductive-step}]
Let $n'' \coloneqq |s_{(B+1):n} \cap (-\infty, a')| = (n' - i_1) - (B' - i)$ denote the number of elements below $a'$ in the second half of the stream, which plays the role of ``$n'$'' in the recursive call. We upper bound the expectation of interest---$\Ex{}{\left|\rank_{{1:n}}(x^*; a) - k\right|\cdot \1{\goodevent}}$---by separately controlling the contribution from the case that $B' \in (n'/4, 3n'/4)$, and that from the $B' \notin (n'/4, 3n'/4)$ case. The former case is better-behaved. The latter case is unlikely to happen, so its contribution will be negligible. We further divide the former case into two sub-cases, depending on whether $n'' < k'$ (so that the recursive call returns $-\infty$) or $n'' \ge k'$ (in which case the recursive call is accurate by our assumptions).

\paragraph{Case 1: $n'/4<B'<3n'/4$ and $n'' < k'$.} In this case, the recursive call on Line~\ref{line:recursive-call} tries to find the $k'$-th largest element among $n''$ elements. By our second assumption, the recursive call always returns $-\infty$. It follows that the output of the current recursion, $x^*$, is simply the smallest element in $s_{1:n}$, i.e., $\rank_{1:n}(x^*; a) = n'$. 

Note that we can re-write $n' - k$ as follows:
\begin{align*} 
    n' - k
=   &~i_1+\left(n'-i_1-(B'-i)\right)\cdot \frac{n'-i_1}{n'-i_1-(B'-i)}  - k\\
\le &~i_1+\left(n'-i_1-(B'-i)\right)\cdot \left(2 + \left|\frac{n'}{n'-B'}-2\right|\right)\\
&\quad + \left(n'-i_1-(B'-i)\right)\cdot \left|\frac{n'-i_1}{n'-i_1-(B'-i)} - \frac{n'}{n'-B'}\right|  - k\\
\le &~i_1+\left(\left\lfloor\frac{k}{2}\right\rfloor - i\right)\cdot \left(2 + \left|\frac{n'}{n'-B'}-2\right|\right)\\
&\quad + \left(n'-i_1-(B'-i)\right)\cdot \left|\frac{n'-i_1}{n'-i_1-(B'-i)} - \frac{n'}{n'-B'}\right|  - k\\
\le &~i_1 - 2i + (k/2)\cdot \left|\frac{n'}{n'-B'}-2\right| + \left|i_1 - i\cdot \frac{n'}{B'}\right|\cdot \left| \frac{B'}{n'-B'}\right|,
\end{align*}
where the second step applies the triangle inequality
\[
    \frac{n'-i_1}{n'-i_1-(B'-i)}
\le \left|\frac{n'-i_1}{n'-i_1-(B'-i)} - \frac{n'}{n'-B'}\right| + \left|\frac{n'}{n'-B'}-2\right| + 2,
\]
the third step holds since 
\[
    n'-i_1-(B'-i)
=   n''
<   k'
=   \left\lfloor\frac{k}{2}\right\rfloor - i,
\]
and the last step follows from $\left\lfloor\frac{k}{2}\right\rfloor - i \le k/2$, as well as
\begin{align*}
    \left|\frac{n'-i_1}{n'-i_1-(B'-i)} - \frac{n'}{n'-B'}\right|
=   &~\left|\frac{n'-i_1 - [n'-i_1-(B'-i)]\cdot \frac{n'}{n'-B'}}{n'-i_1-(B'-i)}\right| \\
=   &~\left|\frac{-i_1 + (i_1-i)\cdot \frac{n'}{n'-B'}}{n'-i_1-(B'-i)}\right| \\
=   &~ \frac{1}{|n'-i_1-(B'-i)|}\cdot\left| \frac{B'}{n'-B'}\right|\cdot \left|i_1 - \frac{n'}{B'}\cdot i\right|.
\end{align*}
Then, after multiplying the two indicators
\begin{align*}
   \mathbbm{1}_{b}&\coloneqq \1{\goodevent \wedge B' \in (n'/4, 3n'/4)},\\
\mathbbm{1}_{i_1}& \coloneqq  \1{n'' < k'}
\end{align*}
and taking an expectation, we obtain:
\begin{align*}
            A_1
&\coloneqq  \Ex{}{\mathbbm{1}_{b}\mathbbm{1}_{i_1}|\rank_{1:n}(x^*; a) - k|}\\
&\le        \Ex{}{\mathbbm{1}_{b}\mathbbm{1}_{i_1}(i_1 - 2i)} + (k/2)\cdot \Ex{}{\mathbbm{1}_{b}\mathbbm{1}_{i_1}\left|\frac{n'}{n'-B'}-2\right|} + \Ex{}{\mathbbm{1}_{b}\mathbbm{1}_{i_1}\left|i_1 - i\cdot \frac{n'}{B'}\right|\cdot \left| \frac{B'}{n'-B'}\right|}\\
&\le        \Ex{}{\mathbbm{1}_{b}\mathbbm{1}_{i_1}(i_1 - 2i)} + (k/2)\cdot \Ex{}{\mathbbm{1}_{b}\mathbbm{1}_{i_1}\left|\frac{n'}{n'-B'}-2\right|} + 3\cdot \Ex{}{\mathbbm{1}_{b}\mathbbm{1}_{i_1}\left|i_1 - i\cdot \frac{n'}{B'}\right|
}\\
&\le        4\Ex{}{\mathbbm{1}_{b}\mathbbm{1}_{i_1}\left|i_1 - i\cdot \frac{n'}{B'}\right|} + (k/2)\cdot \Ex{}{\mathbbm{1}_{b}\mathbbm{1}_{i_1}\left|\frac{n'}{n'-B'}-2\right|} +  \Ex{}{\mathbbm{1}_{b}\mathbbm{1}_{i_1}i\cdot \left|2 - \frac{n'}{B'}\right|}\\
&\le        4\Ex{}{\mathbbm{1}_{b}\mathbbm{1}_{i_1}\left|i_1 - i\cdot \frac{n'}{B'}\right|} + (k/2)\cdot \Ex{}{\mathbbm{1}_{b}\mathbbm{1}_{i_1}\left|\frac{n'}{n'-B'}-2\right|}+ (k/2)\cdot \Ex{}{ \mathbbm{1}_{b}\mathbbm{1}_{i_1}\left|2 - \frac{n'}{B'}\right|},
\end{align*}
The third step holds since $B' \in (n'/4, 3n'/4)$ implies $B' / (n' - B') \le 3$. The fourth step applies the triangle inequality again: $i_1 - 2i \le |i_1 - i\cdot(n' / B')| + |2i - i\cdot(n' / B')|$. The last step applies $i \le \lfloor k / 2\rfloor \le k /2$, which holds whenever the algorithm successfully finds a valid element $a'$ on Line~\ref{line:find-a'}, which is in turn guaranteed by the event $\goodevent$.

\paragraph{Case 2: $n' / 4 < B' < 3n'/4$ and $n'' \ge k'$.} Define
    \[l \coloneqq\rank_{B+1:n}(x'; a').\]
By our assumptions, conditioning on any realization of $(B', i, i_1)$ that leads to the recursive call with $n'' \ge k'$, we have
\begin{equation}\label{eq:lemma-assumption}
    \Ex{l}{\left|l - k'\right|} \le C_2 \cdot \sqrt{k'}. 
\end{equation}
In the following, we write $\mathbbm{1}_{i_1}' \coloneqq \1{n'' \ge k'}$ as the complement of $\mathbbm{1}_{i_1}$.

Recall that our goal is to control the expectation
\[
    \Ex{}{\mathbbm{1}_b\mathbbm{1}'_{i_1}|\rank_{1:n}(x^*; a) - k|}.
\]
By the triangle inequality, $|\rank_{1:n}(x^*; a) - k|$ is upper bounded by
\begin{equation}\label{eq:case-2-three-terms}
    \left|l\cdot \frac{n'-i_1}{n'-i_1 - (B'-i)} - 2l\right|
+   |i_1 + 2l - k|
+   \left|\rank_{1:n}(x^*; a) - \left(i_1+l\cdot \frac{n'-i_1}{n'-i_1-(B'-i)}\right)\right|.
\end{equation}
Let $A_{2,1}$, $A_{2,2}$ and $A_{2,3}$ denote the expectation of the three terms above, after multiplying the indicators $\mathbbm{1}_b\mathbbm{1}'_{i_1}$. We will upper bound these three terms separately in the following.

We start with $A_{2,1}$. Conditioning on the realization of $B' \ge \lfloor k/2 \rfloor$, we have
\begin{align*}
    &~\Ex{i_1}{\mathbbm{1}_{i_1}'\Ex{l}{\left|l\cdot \frac{n'-i_1}{n'-i_1 - (B'-i)} - 2l\right|}}\\
\le &~\Ex{i_1}{\mathbbm{1}_{i_1}'\Ex{l}{\left|l\cdot \frac{n'-i_1}{n'-i_1 - (B'-i)} - l\cdot \frac{n'}{n'-B'}\right| + \left|2l- l\cdot \frac{n'}{n'-B'}\right|}} \tag{triangle inequality}\\
=   &~\Ex{i_1}{\mathbbm{1}_{i_1}' \Ex{l}{l\cdot\frac{\left|i_1 - i\cdot \frac{n'}{B'}\right|\cdot \left|\frac{B'}{n'-B'}\right|}{n''}}} + \Ex{}{l}\cdot \left|\frac{n'}{n'-B'} - 2\right|\\
\le &~\Ex{i_1 }{\mathbbm{1}_{i_1}'\left|i_1 - i\cdot \frac{n'}{B'}\right|\cdot \left|\frac{B'}{n'-B'}\right|} +\left(k' + \Ex{}{\left|l - k'\right|}\right)\cdot \left|\frac{n'}{n'-B'} - 2\right|\tag{$l \le n''$, $l \le k' + |l - k'|$} \\
\le &~\Ex{i_1}{\mathbbm{1}_{i_1}'\left|i_1 - i\cdot \frac{n'}{B'}\right|\cdot \left|\frac{B'}{n'-B'}\right|} +\left(k' + C_2\sqrt{k'}\right)\cdot \left|\frac{n'}{n'-B'} - 2\right| \tag{\Cref{eq:lemma-assumption}}\\
\le &~\Ex{i_1}{\mathbbm{1}_{i_1}' \left|i_1 - i\cdot \frac{n'}{B'}\right|\cdot \left|\frac{B'}{n'-B'}\right|} +\left(\delta + C_2\sqrt{\delta}\right)\cdot \left|\frac{n'}{n'-B'} - 2\right|. \tag{$k' \le \delta$}
\end{align*}

Then taking an expectation over the randomness in $B'$ gives:
\begin{align*}
    A_{2,1}
&\coloneqq  \Ex{}{\mathbbm{1}_b\mathbbm{1}_{i_1}'\left|l\cdot \frac{n'-i_1}{n'-i_1 - (B'-i)} - 2l\right|}\\
&\le        \left(\delta + C_2\sqrt{\delta}\right)\cdot \Ex{}{\mathbbm{1}_b\mathbbm{1}_{i_1}'\left|\frac{n'}{n'-B'} - 2\right|}+ \Ex{}{\mathbbm{1}_b\mathbbm{1}_{i_1}' \left|i_1 - i\cdot \frac{n'}{B'}\right|\cdot \left|\frac{B'}{n'-B'}\right|}\\
&\le        \left(\delta + C_2\sqrt{\delta}\right)\cdot  \Ex{}{\mathbbm{1}_b\mathbbm{1}_{i_1}'\left|\frac{n'}{n'-B'} - 2\right|} +  3\cdot \Ex{}{ \mathbbm{1}_b\mathbbm{1}_{i_1}'\left|i_1 - i\cdot \frac{n'}{B'}\right|},
\end{align*}
where the third step holds since $\mathbbm{1}_b \ne 0 \implies B' \in (n'/4, 3n'/4) \implies B' / (n' - B') \le 3$.

Now we turn to the $A_{2,2}$ term, i.e., the expectation of the second term $|i_1 + 2l - k|$ in \Cref{eq:case-2-three-terms}. Conditioning on the value of $B'$ and taking an expectation over $(i, i_1, l)$ gives
\begin{align*}
    &~\Ex{}{\mathbbm{1}_{i_1}'\left|(i_1 + 2l)-k\right|}\\
\le &~\Ex{}{\mathbbm{1}_{i_1}'\left|i_1 + 2k' - k\right|} + 2\Ex{}{\left|l - k'\right|} \tag{triangle inequality}\\
\le &~\Ex{}{\mathbbm{1}_{i_1}'\left|i_1 - 2i + 2\lfloor k/2\rfloor -k\right|} + 2C_2 \cdot \sqrt{k'} \tag{$k' = \lfloor k/2\rfloor - i$, \Cref{eq:lemma-assumption}}\\
\le &~\Ex{}{\mathbbm{1}_{i_1}'|i_1 - 2i|} + 2C_2\sqrt{\delta}+ 1. \tag{$|2\lfloor k/2\rfloor - k| \le 1$, $k' \le \delta$}
\end{align*}
Multiplying with the indicator $\mathbbm{1}_b$ and taking another expectation over $B'$ gives
\begin{align*}
    A_{2,2}
&       \coloneqq\Ex{}{\mathbbm{1}_b\mathbbm{1}_{i_1}'\left|(i_1 + 2l)-k\right|}\\
&\le    \Ex{}{\mathbbm{1}_b\mathbbm{1}_{i_1}'|i_1 - 2i|} + 2C_2\sqrt{\delta}+ 1\\
&\le    \Ex{}{\mathbbm{1}_b\mathbbm{1}_{i_1}'\left(\left|i_1 - \frac{n'}{B'}\cdot i\right| + i\cdot\ \left|\frac{n'}{B'} - 2\right|\right)} + 2C_2\sqrt{\delta}+ 1 \tag{triangle inequality}\\
&\le    \Ex{}{\mathbbm{1}_b\mathbbm{1}_{i_1}'\left(\left|i_1 - \frac{n'}{B'}\cdot i\right| + \frac{k}{2}\cdot\ \left|\frac{n'}{B'} - 2\right|\right)} + 2C_2\sqrt{\delta} + 1,
\end{align*}
where the last step holds since, whenever $\mathbbm{1}_b \ne 0$, event $\goodevent$ happens, which guarantees $i \le \lfloor k/2\rfloor \le k/2$.

Finally, for the last term $A_{2,3}$, \Cref{lemma:rank-in-second-half-to-overall-simplified} implies that, conditioning on fixed $B'$ and $i$, we have
\[
    \Ex{}{\mathbbm{1}_{i_1}'\left|\rank_{1:n}(x^*; a) - \left(i_1+l\cdot \frac{n'-i_1}{n'-i_1-(B'-i)}\right)\right|} \le 914\sqrt{C_2}\cdot\sqrt{k'},
\]
where $k' = \lfloor k/2\rfloor - i \le \delta$ for all possible values of $i$. Then, taking an expectation over $B'$ and $i$ gives
\begin{align*}
            A_{2,3}
&\coloneqq  \Ex{}{\mathbbm{1}_b\mathbbm{1}_{i_1}'\left|\rank_{1:n}(x^*; a) - \left(i_1+l\cdot \frac{n'-i_1}{n'-i_1-(B'-i)}\right)\right|}\\
&\le       914\sqrt{C_2}\cdot\sqrt{\delta}.
\end{align*}

\paragraph{Case 3: Either $B' \le n'/4$ or $B' \ge 3n'/4$.} By a Chernoff bound,
\[
    \pr{B' \sim \Binomial\left(n', 1/2\right)}{B' \le \frac{n'}{4} \vee B' \ge \frac{3n'}{4}}  \le 2\cdot\exp\left(-2\cdot n' \cdot (1/4)^2\right) = 2e^{-n'/8}.
\]
Note that, regardless of the choice of $x^* \in s_{1:n} \cap (-\infty, a)$, its rank $\rank_{1:n}(x^*; a)$ is always between $1$ and $n'$. Thus, we have the upper bound
\[
    A_3 \coloneqq \Ex{}{\1{B' \notin (n'/4, 3n'/4)}\cdot\left|\rank_{1:n}(x^*; a) - k\right|}
\le 2n'\cdot \exp\left(-n'/8\right) \le 6,
\]
where the last step follows from $\max_{x \ge 0}2xe^{-x/8} = 16 / e \le 6$.

\paragraph{Put everything together.} Now, we upper bound the total contribution from the three cases above, namely, 
\[
    \Ex{}{\left|\rank_{1:n}(x^*; a) - k\right| \cdot \1{\goodevent}}
\le A_1 + A_{2,1} + A_{2,2} + A_{2,3} + A_3.
\]

Recall that $\mathbbm{1}_{b} = \1{\goodevent \wedge B' \in (n'/4, 3n'/4)}$ is the indicator for the following three to hold simultaneously: (1) $B' \ge \lfloor k/2\rfloor$; (2) $\findkth$ successfully finds $a'$ on Line~\ref{line:find-a'}; (3) $B' \in (n'/4, 3n'/4)$. For brevity, we shorthand the terms that frequently appear in these upper bounds: 
\begin{align*}
    T_1& \coloneqq \Ex{}{\mathbbm{1}_b|n'/B' - 2|},\\
    T_2 &\coloneqq \Ex{}{\mathbbm{1}_b|n'/(n' - B') - 2|},\\
    T_3 & \coloneqq \Ex{}{\mathbbm{1}_b\left|i_1 - i\cdot \frac{n'}{B'}\right|}.
\end{align*}
Then, by dropping the indicators $\mathbbm{1}_{i_1}$ and $\mathbbm{1}'_{i_1}$, the upper bounds that we have obtained so far can be simplified into:
\begin{align*}
    A_1 &\le \frac{k}{2}(T_1 + T_2) + 4T_3,\\
    A_{2, 1} & \le \left(\delta + C_2\sqrt{\delta}\right)\cdot T_2 + 3T_3,\\
    A_{2, 2} & \le \frac{k}{2}\cdot T_1 + T_3 + \left(2C_2\sqrt{\delta} + 1\right),\\
    A_{2, 3} & \le  914\sqrt{C_2}\cdot\sqrt{\delta}\\
    A_3 & \le 6.
\end{align*}
Taking a sum gives an overall upper bound of
\begin{equation}\label{eq:overall-bound-unsimplified}
    kT_1 + \left(\delta + C_2\sqrt{\delta} + \frac{k}{2}\right)\cdot T_2 + 8T_3 + \left(914\sqrt{C_2}\cdot\sqrt{\delta} + 2C_2\sqrt{\delta} + 7\right).
\end{equation}

Now we upper bound the terms $T_1$, $T_2$ and $T_3$. Recall that $\mathbbm{1}_b \ne 0$ implies $B' \in (n'/4, 3n'/4)$, which further implies $B' \ne 0$ and $n' - B' \ne 0$. Thus, by \Cref{lemma.integral1},
\[
    T_1, T_2
\le \frac{14}{\sqrt{n'}}
\le \frac{14}{\sqrt{k}}.
\]
For $T_3$, we apply \Cref{lemma:subset-rank-concentration}, $n' / B' < 4$ and $i \le \lfloor k/2\rfloor \le k/2$ to obtain
\[
    T_3
\le \Ex{}{\mathbbm{1}_b\left(2\cdot\frac{n'}{B'}\cdot\sqrt{i\cdot \frac{n'}{B'}} + \left(\frac{n'}{B'}\right)^2\right)}
\le 16\sqrt{k/2} + 16 
\le 28\sqrt{k}.
\]

Then, \Cref{eq:overall-bound-unsimplified} can be further simplified into
\begin{align*}
    &~14\sqrt{k} + (\delta + C_2\sqrt{\delta})\cdot\frac{14}{\sqrt{k}} + 7\sqrt{k} + 224\sqrt{k} + (914\sqrt{C_2}\cdot\sqrt{\delta} + 2C_2\sqrt{\delta} + 7)\\
\le &~252\sqrt{k} + \left[(\delta + C_2\sqrt{\delta})\cdot\frac{14}{\sqrt{k}} + 914\sqrt{C_2}\cdot\sqrt{\delta} + 2C_2\sqrt{\delta}\right].
\end{align*}
Recall that $\delta$ is set to $C_0\cdot k$ in \Cref{algorithm.recurse}. When $C_0$ is sufficiently small, all the terms that depend on $\delta$ in the above can be made smaller than $\sqrt{k}$ in total. It follows that
\[
    \Ex{}{\left|\rank_{1:n}(x^*; a) - k\right| \cdot \1{\goodevent}}
\le 253\sqrt{k}
=   (C_2 - 9)\cdot\sqrt{k}.
\]
\end{proof}

Finally, we prove our main theorem below.

\quantileapprox*

\begin{proof}
We prove the following statement regarding the $\findkth$ procedure (\Cref{algorithm.recurse}) by induction on $k$: 
\paragraph{Induction Hypothesis:} The procedure $\findkth(n, m, k, a)$, when running on a random-order sequence $s_1, s_2, \ldots, s_n$ that satisfies $n' \coloneqq |s_{1:n} \cap (-\infty, a)| \ge k$, returns an element $x^*$ that satisfies
\[
    \Ex{}{\left|\rank_{1:n}(x^*; a) - k\right|} \le C_2\sqrt{k},
\]
where $C_2 = 262$ and the expectation is over the randomness both in the ordering of the $n$ elements and in the algorithm.

The induction hypothesis above, when applied to the call to $\findkth(n, m, k, +\infty)$ in \Cref{algorithm.main}, directly proves the theorem. Also note that the proof by induction is valid since, whenever $\findkth$ is called recursively (on Line~\ref{line:recursive-call}), the parameter ``$k$'' for the recursive call is set to $\lfloor k/2 \rfloor - i$, which is strictly smaller than $k$.

\paragraph{The base case.} We start with the base case that $k \in \{1, 2, \ldots, m\}$. Since $k \le m$, $\findkth$ would always choose to use the straightforward method on Line~\ref{line:naive}, and find the $k$-th largest element exactly. This ensures that $\Ex{}{|\rank_{1:n}(x^*; a) - k|} = 0 \le C_2\sqrt{k}$ holds in the base case.

\paragraph{The inductive step.} Now we deal with the inductive step for $k > m$, in which case \Cref{algorithm.recurse} does not use the straightforward solution on Line~\ref{line:naive}. Note that we may also assume that $n' \ge k$; otherwise, the induction hypothesis would be vacuously true.

We upper bound the expectation by considering the following three cases:
\begin{itemize}
    \item \textbf{Case 1:} $s_{1:B}$ contains strictly fewer than $\lfloor k/2 \rfloor$ elements that are strictly smaller than $a$. In other words, $B' \coloneqq |s_{1:B} \cap (-\infty, a)| < \lfloor k/2 \rfloor$.
    \item \textbf{Case 2:} $B' \ge \lfloor k/2 \rfloor$, but we fail to find $a'$ on Line~\ref{line:find-a'}.
    \item \textbf{Case 3:} $B' \ge \lfloor k/2 \rfloor$ and we successfully find $a'$ on Line~\ref{line:find-a'}. 
\end{itemize}
Formally, we can decompose the expectation in question into
\[
    \Ex{}{\left|\rank_{1:n}(x^*; a) - k\right|} = A_1 + A_2 + A_3,
\]
where for each $i \in \{1, 2, 3\}$,
\[
    A_i \coloneqq \Ex{}{\left|\rank_{1:n}(x^*; a) - k\right|\cdot \1{\text{Case $i$}}}.
\]
In the rest of the proof, we upper bound the three terms separately.

\paragraph{Case 1: $B' < \lfloor k/2 \rfloor$.} In this case, $\findkth$ would set $x \gets -\infty$ on Line~\ref{line.case1}. Then, by Line~\ref{line:translate-infinity-to-smallest}, the algorithm would eventually output $x^* = \underline{x}$, the smallest element among $s_1, s_2, \ldots, s_n$. In other words, we always have $\rank_{1:n}(x^*; a) = |s_{1:n} \cap (-\infty, a)| = n'$ in this case. Recall that $B'$ follows the distribution $\Binomial(n', 1/2)$. Therefore, the contribution of this case, denoted by $A_1 \coloneqq \Ex{}{|\rank_{1:n}(x^*; a) - k| \cdot \1{\text{Case 1}}}$, is given by
\[
    A_1
=   \Ex{B' \sim \Binomial(n', 1/2)}{(n' - k) \cdot \1{B' < \lfloor k/2 \rfloor}}.
\]
For any $B' < \lfloor k / 2 \rfloor$, it holds that
$n' - k \le n' - 2B' \le |n' - 2B'|$, so we have the upper bound
\[
    A_1
\le \Ex{B'}{\left|n' - 2B'\right| \cdot \1{B'< \lfloor k/2\rfloor}}.
\]
Note that $B' = 0$ happens with probability $2^{-n'}$, so this case contributes at most $2^{-n'}\cdot n' \le 1$ to the right-hand side above. In other words,
\[
    A_1 \le \Ex{B'}{\left|n' - 2B'\right| \cdot \1{ 1\le B'< \lfloor k/2\rfloor}\}} + 1.
\]
For any $B' \in [1, \lfloor k/2 \rfloor)$, we have $|n' - 2B'|
=   B'\cdot|n' / B' - 2|
\le \frac{k}{2}\cdot |n' / B' - 2|$. Therefore,
\begin{align*}
        \Ex{B'}{\left|n' - 2B'\right| \cdot \1{1 \le B'< \lfloor k/2\rfloor}}
&\le    \frac{k}{2} \cdot \Ex{B'}{\left|\frac{n'}{B'}-2\right| \cdot \1{ 1\le B'< \lfloor k/2\rfloor}}\\
&\le    \frac{k}{2}\cdot\Ex{B'}{\left|\frac{n'}{B'}-2\right| \cdot \1{B' \ne 0}}\\
&\le    \frac{k}{2}\cdot \frac{14}{\sqrt{n'}} \tag{\Cref{lemma.integral1}}\\
&\le    7\sqrt{k}. \tag{$n' \ge k$}
\end{align*}

Therefore, we conclude that $A_1
\le 7\sqrt{k} + 1 \le 8\sqrt{k}$.
     
\paragraph{Case 2: $B' \ge \lfloor k/2\rfloor$ but $a'$ cannot be found.} In this case, by Line~\ref{line:no-a'} of the algorithm, the output $x^*$ is set to the largest element in the sequence that is strictly smaller than $a$. In other words, we always have $\rank_{1:n}(x^*; a) = 1$, which gives $|\rank_{1:n}(x^*; a) - k| = k - 1 \le k$. Furthermore, by \Cref{lemma:prob-find-a'}, this case happens with probability at most
\[
    \exp\left(-\frac{2}{9}\cdot m\right) + \exp\left(1-\frac{2C_0}{3}\cdot m\right).
\]
Therefore, the contribution of this case to $\Ex{}{|\rank_{1:n}(x^*; a) - k|}$ is at most
\begin{align*}
    A_2
&=  \Ex{}{|\rank_{1:n}(x^*; a) - k| \cdot \1{\text{Case 2}}}\\
&\le k\cdot \pr{}{\text{Case 2}}\\
&\le k\cdot\left[\exp\left(-\frac{2}{9}\cdot m\right) + \exp\left(1-\frac{2C_0}{3}\cdot m\right)\right]\\
&\le \sqrt{k},
\end{align*}
where the last step holds for all sufficiently large $m \ge \Omega(\log k)$.
     
\paragraph{Case 3: $B' \ge \lfloor k/2 \rfloor$ and $a'$ is found.} In this case, we make a recursive call on Line~\ref{line:recursive-call}, and use its result as the output $x^*$. By \Cref{lemma:inductive-step}, we have
\[
    A_3 = \Ex{}{\left|\rank_{1:n}(x^*; a) - k\right|\cdot\1{\text{Case 3}}} \le (C_2 - 9)\sqrt{k}.
\]       

\paragraph{Put everything together.} Combining the three cases, we have
\[
    \Ex{}{\left| \rank_{1:n}(x^*; a) - k\right|}
\le A_1 + A_2 + A_3
\le 8\sqrt{k} + \sqrt{k} + (C_2 - 9)\sqrt{k}
=   C_2\sqrt{k}.
\]
This completes the inductive step and finishes the proof.
\end{proof}

\section{Quantile Estimation: Exact Selection in Sublinear Memory}
In this section, we give an algorithm (\Cref{alg:exact-selection}) that, given a random-order string of $n$ elements, outputs the exact $k$-th largest element with high probability.

\begin{algorithm2e}
    \caption{Exact Selection}\label{alg:exact-selection}
    \KwIn{Stream length $n$, memory size $m$, target rank $k$, and access to random-order stream $s = (s_1, s_2, \ldots, s_n)$}
    $B_{ \lfloor \log_2 k \rfloor + 1} \leftarrow n$; $B_0 \leftarrow 0$\;
    \lFor{$i = \lfloor \log_2 k \rfloor, \ldots, 2, 1$}{sample $B_i \sim \Binomial\left(B_{i+1}, 1/2\right)$}
    \lFor{$i = 0, 1, \ldots, \lfloor \log_2 k\rfloor$} {$T_i \leftarrow \min\left\{B_{i+1},\left\lfloor\frac{k}{2^{ \lfloor\log_2 k\rfloor  - i}}\right\rfloor\right\}$}
    $i_0 \gets \min\{i \in \{0, 1, \ldots, \lfloor \log_2 k\rfloor\}: T_{i_0} > m/2$\}\;
    Let $M = (+\infty, -\infty, -\infty, \ldots, -\infty)$ be an array of length $3m$\;
    Read the first $B_{i_0 + 1}$ elements; store the largest $\min\{B_{i_0+1}, 3m-1\}$ elements among them in $M[2]$ through $M[\min\{B_{i_0+1}, 3m-1\}+1]$ in decreasing order\;\label{line:base-case}
    $M.\text{move}(3m/2 - T_{i_0} - 1)$\;\label{line:base-case-shift} 
    $\rank \gets T_{i_0}$\;
    \For{$i = i_0+1, \ldots, \lfloor \log_2 k\rfloor$} {
        \For{$j = B_i +1, B_{i} + 2, \ldots, B_{i+1}$} {
            Read the $j$-th element $s_j$\;
            \lIf{$s_j > M[3m/2]$}{$\rank \gets \rank + 1$}
            $M.\text{insert}(s_j)$\;
        }
        $M.\text{move}(\rank - T_i)$\; \label{line:inductive-step-shift}
        $\rank \gets T_i$\;
    }
    \Return $M[3m/2]$\;
\end{algorithm2e}

\begin{algorithm2e}
    \caption{$M.\text{insert}(q)$}
    Find the smallest $i^*$ such that $M[i^*] < q$\;
    \lIf{no such $i^*$ exists or $M[i^*-1] = \bot$}{\Return}
    \uIf{$i^* \le 3m/2$}{
        \lFor {$ i = 1, 2, \ldots, i^* - 2$} {$M[i] \gets M[i + 1]$}
        $M[i^* - 1]\gets q$\;
    } \Else {
        \lFor {$i = 3m, 3m-1, \ldots, i^* +1$} {$M[i] \gets M[i - 1]$}
        $M[i^*]\gets q$\;
    }
    \label{algorithm.insert}
\end{algorithm2e}

\begin{algorithm2e}
    \caption{$M$.move($d$)}
    \uIf{$d < 0$} {
        \lFor{$i = 1,  \ldots, 3m - |d|$}{$M[i] \leftarrow M\left[i + |d|\right]$}
        \For{$i =  3m - |d| + 1, 3m - |d| + 2, \ldots, 3m$}{
            \lIf{$M[3m - |d|] = -\infty$}{$M[i] \leftarrow -\infty$}
            \lElse {$M[i]\leftarrow \bot$}
        }
    } \Else {
        \lFor{$i = 3m, 3m-1, \ldots, d+1$}{$M[i] \leftarrow M\left[i - d\right]$}
        \For{$i = d, d-1, \ldots, 1$}{
            \lIf{$M[d+1] = +\infty$}{$M[i] \leftarrow +\infty$}
            \lElse{$M[i] \leftarrow \bot$}
        }
    }
  \label{algorithm.move}
\end{algorithm2e}

The algorithm maintains an array $M$ of length $3m$, divided into three blocks of length $m$. The first and the last $m$ entries are called \emph{buffers}. The algorithm tries to keep the $k$-th largest element in the middle $m$ entries of the array. In order to do so: (1) The array of length $3m$ always contains consecutive elements (among the elements that have been observed so far) in decreasing order; (2) We keep track of the rank of the middle element in the array (i.e., $M[3m/2]$) among the elements seen. Note that, by the first condition, we would know the rank of every element in the array.
 
The algorithm runs in $O(\log k)$ \emph{stages}. At the beginning of each stage $i$, the algorithm samples a string of length $B_i \sim \Binomial\left(B_{i+1}, \frac{1}{2}\right)$, while $B_{i+1}$ is the number of elements seen at the end of that stage, and tries to find the $2^i$-th largest element among the first $B_{i+1}$ elements seen.

In a typical execution of the algorithm, at the end of each stage $i$, the $\left\lfloor\frac{k}{2^{\lfloor\log_2 k\rfloor -i}}\right\rfloor$-th ($\approx 2^i$-th) largest element (among the elements that have been encountered so far) stays in the middle $m$ elements of the array. In this case, we shift the entries of the array to ensure that the $\left\lfloor\frac{k}{2^{\lfloor\log_2 k\rfloor -i}}\right\rfloor$-th largest element is at the center of the array (i.e., with index $3m/2$). All the other elements are shifted together, so that the array still stores consecutive elements in decreasing order. Doing so might cause one of the buffers to be half empty. We expect the algorithm, after reading the elements in the next stage, will be able to fill this half-full buffer.
 
We start with an easy observation: at the beginning of stage~$i$, the $B_i$ elements that have already been seen constitute a uniformly random subset of the first $B_{i+1}$ elements in the random-order stream.
\begin{lemma}\label{lemma:uniform-subset}
    Conditioning on the value of $B_{i+1}$ and the set $S$ of the first $B_{i+1}$ elements, the subset of the first $B_i$ elements is uniformly distributed over all subsets of $S$.
\end{lemma}

\begin{proof}
    This is a special case of \Cref{lemma:subsample} when $p = 1/2$.
\end{proof}

Next, we define a ``good event'' over the randomness in the ordering of the $n$ elements as well as the values $B_1$ through $B_{\lfloor \log_2 k\rfloor}$. Later, we will show that this good event happens with high probability, and implies the correctness of \Cref{alg:exact-selection}. Recall from \Cref{alg:exact-selection} that we define $T_i = \min\left\{B_{i+1},\left\lfloor\frac{k}{2^{ \lfloor\log_2 k\rfloor  - i}}\right\rfloor\right\}.$
\begin{definition}[Good event]\label{def.egood}
    Let $s_1, s_2, \ldots, s_n$ denote the $n$ elements in the order in which they arrive. Define $\goodevent$ as the event that the following conditions hold simultaneously:
    \begin{enumerate}
        \item For each stage $i$ such that $T_i \ge m / 2$, the $T_{i-1}$-th largest element among $s_{1:B_i}$ has a rank between $T_{i} - m/2$ and $T_i + m/2$ (inclusive) among $s_{1:B_{i+1}}$.
       \item For each stage $i\ge 2$, 
           either $T_{i-1} \le m$ or at least $m/2$ elements in $s_{B_i +1:B_{i+1}}$ are between the $(T_{i-1}-m+1)$-th and the $ T_{i-1}$-th largest elements in $s_{1:B_i}$.
        \item For each stage $i \ge 2$, 
         either $T_{i-1}  + m> B_i$ or at least $m / 2$ elements in $s_{B_i +1:B_{i+1}}$ are between the $T_{i-1}$-th and the $(T_{i-1} + m)$-th largest elements in $s_{1:B_i}$.
         \item For stage $i_0 \coloneqq \min\left\{i \in \{0, 1, \ldots, \lfloor \log_2 k \rfloor \}: T_{i} > \frac{m}{2} \right\}$, we have $T_{i_0} < 2m$.
    \end{enumerate}
\end{definition}

\begin{lemma}\label{lemma:exact-cond-1}
   For any $m \ge 4$,
Condition 1 in \Cref{def.egood} holds with probability at least
    \[
        1-   \lfloor\log_2 k \rfloor\cdot2^{-m/2} - 2\sum_{i=0}^{\lfloor\log_2k\rfloor-1}\exp\left(- \frac{m^2}{32(k/2^i)}\right).
    \]
\end{lemma}

\begin{proof}
We fix a stage $i$, and condition on the value of $B_{i+1}$ as well as the set $\{s_1, s_2, \ldots, s_{B_{i+1}}\}$ of the first $B_{i+1}$ elements in the stream (but not the exact order in which they arrive). We will upper bound the probability for Condition~1 to be violated in Stage~$i$. Note that after conditioning on $B_{i+1}$, the value of $T_i$ is determined, and we may assume that $T_i \ge m / 2$; otherwise, Condition~1 would hold vacuously.

Let $x_i^*$ denote the $T_{i-1}$-th largest element among $s_{1:B_i}$. Recall that Condition~1 is the intersection of the following two events:
\begin{enumerate}
    \item Event $E_1$: The rank of $x^*_i$ among $s_{1:B_{i+1}}$ is at least $T_i - m/2$.
    \item Event $E_2$: The rank of $x^*_i$ among $s_{1:B_{i+1}}$ is at most $T_i + m/2$.
\end{enumerate}
In the following, we lower bound the probability of each of the two events above.

\paragraph{Lower bound $\pr{}{E_1}$.} We note that the following is a sufficient condition for $E_1$:
\begin{itemize}
    \item Event $E'_1$: Among the largest $T_i - m/2$ elements in $s_{1:B_{i+1}}$, strictly fewer than $T_{i-1}$ of them appear in $s_{1:B_i}$.
\end{itemize}
To see why $E'_1$ implies $E_1$, we sort the elements $s_1, s_2, \ldots, s_{B_{i+1}}$ in decreasing order: $x_1 > x_2 > \cdots > x_{B_{i+1}}$. When $E'_1$ happens, strictly fewer than $T_{i-1}$ of the elements $x_1, x_2, \ldots, x_{T_i - m/2}$ appear in $s_{1:B_i}$. Then, $x_i^*$---being the $T_{i-1}$-th largest among $s_{1:B_i}$---must be smaller than $x_{T_i - m/2}$. In other words, the rank of $x^*_i$ among $s_{1:B_{i+1}}$ must be strictly higher than $T_i - m/2$, which implies event $E_1$. 

Next, we note that if event $E'_1$ does not happen, at least one of the following two must be true:
\begin{itemize}
    \item Event $E_{1, a}$: None of the $m/2$ elements $x_{B_{i+1} - m/2 + 1}, x_{B_{i+1} - m/2 + 2}, \ldots, x_{B_{i+1}}$ appear in $s_{1:B_i}$. (This is well-defined, since $B_{i+1} \ge T_i \ge m/2$.)
    \item Event $E_{1, b}$: At least $\lfloor k / 2^{\lfloor \log_2 k\rfloor - (i-1)} \rfloor$ elements among $x_1, x_2, \ldots, x_{T_i - m /2}$ appear in $s_{1:B_i}$.
\end{itemize}
To see this, recall that $T_{i-1} = \min\left\{B_i, \lfloor k/2^{\lfloor\log_2 k\rfloor - (i-1)}\rfloor\right\}$. If $B_i \ge \lfloor k/2^{\lfloor\log_2 k\rfloor - (i-1)}\rfloor$, we have $T_{i-1} = \lfloor k / 2^{\lfloor \log_2 k\rfloor - (i-1)} \rfloor$, so the negation of event $E'_1$ would exactly be event $E_{1,b}$. If $B_{i} < \lfloor k/2^{\lfloor\log_2 k\rfloor - (i-1)}\rfloor$, we have $T_{i-1} = B_i$, in which case event $E'_1$ is violated only if exactly $T_{i-1} = B_i$ elements among $x_1, \ldots, x_{T_i - m/2}$ appear in $s_{1:B_i}$. For this to happen, none of the elements $x_{T_i - m/2 + 1}, \ldots, x_{B_{i+1}}$ may appear in $s_{1:B_i}$. Since $B_{i+1} \ge T_i$, this implies the event $E_{1,a}$.

Now we lower bound $\pr{}{E'_1}$ by upper bounding $\pr{}{E_{1,a}}$ and $\pr{}{E_{1,b}}$. By \Cref{lemma:uniform-subset}, each of $x_{T_i - m/2+1}, \ldots, x_{B_{i+1}}$ appears in $s_{1:B_i}$ with probability $1/2$ independently. It follows that
\[
    \pr{}{E_{1,a}} \le 2^{-m/2}.
\]
To control $\pr{}{E_{1,b}}$, we note that, for any $j \in [B_{i+1}]$, \Cref{lemma:uniform-subset} implies that $s_{1:B_i} \cap \{x_1, \ldots , x_j\}$ is uniformly distributed among the $2^{j}$ subsets of $\{x_1, x_2, \ldots, x_j\}$. In particular, for $j = T_i - m / 2$, $\left|s_{1:B_i} \cap \{x_1, \ldots, x_{T_{i}-m/2}\}\right|$ follows the binomial distribution $\Binomial\left(T_{i}-m/2, 1/2\right)$. Since $\lfloor k/2^{\lfloor\log_2 k\rfloor - (i-1)}\rfloor \ge \frac{1}{2}\cdot (\lfloor k/2^{\lfloor\log_2 k\rfloor - i}\rfloor-1) \ge \frac{1}{2}\cdot (T_i-1)$, by a Chernoff bound, we have
\begin{align*}
    \pr{}{E_{1,b}}
&=  \pr{X \sim \Binomial(T_i - m/2, 1/2)}{X \ge \lfloor k / 2^{\lfloor \log_2 k\rfloor - (i-1)}\rfloor}\\
&\le\pr{X \sim \Binomial(T_i - m/2, 1/2)}{X \ge (T_i-1) / 2}\\
&=  \pr{X \sim \Binomial(T_i - m/2, 1/2)}{X - \left(T_i / 2 - m / 4\right) \ge m / 4 - 1/2}\\
&\le\exp\left(-\frac{2\cdot(m / 4 - 1/2)^2}{T_i - m/2}\right)\\
&\le \exp\left(-\frac{(m-2)^2}{8\cdot T_{i}}\right) \tag{$T_i - m / 2 \le T_i$} \\
&\le \exp\left(-\frac{m^2}{32k/2^{\lfloor\log_2k\rfloor-i}}\right). \tag{$m \ge 4$, $T_i \le k/2^{\lfloor\log_2k\rfloor-i}$}
\end{align*}

Therefore, we have
\[
    \pr{}{E_1} 
\ge \pr{}{E'_1}
\ge 1 - \pr{}{E_{1,a}} - \pr{}{E_{1, b}}
\ge 1 - \left[\exp\left(-\frac{m^2}{32k/2^{\lfloor\log_2k\rfloor-i}}\right) + 2^{-m/2}\right].
\]

\paragraph{Lower bound $\pr{}{E_2}$.} The analysis for $\pr{}{E_2}$ is almost symmetric. First, we note that the following is a sufficient condition for $E_2$:
\begin{itemize}
    \item Event $E_{2}'$: Among the largest $\min\{\left\lfloor k/2^{\lfloor\log_2 k\rfloor - i}\right\rfloor + m/2, B_{i+1}\}$ elements in $s_{1:B_{i+1}}$, at least $T_{i-1}$ of them appear in $s_{1:B_i}$.
\end{itemize}
To see why $E'_2$ implies $E_2$, again, sort the elements $s_1, s_2, \ldots, s_{B_{i+1}}$ in decreasing order: $x_1 > x_2 > \cdots > x_{B_{i+1}}$. When $E'_2$ happens, at least $T_{i-1}$ of the elements $x_1, x_2, \ldots, x_{\min\{\left\lfloor k/2^{\lfloor\log_2 k\rfloor - i}\right\rfloor + m/2, B_{i+1}\}}$ appear in $s_{1:B_i}$. Then, $x_i^*$---being the $T_{i-1}$-th largest among $s_{1:B_i}$---must be larger than or equal to $x_{\min\{\left\lfloor k/2^{\lfloor\log_2 k\rfloor - i}\right\rfloor + m/2, B_{i+1}\}}$. In other words, the rank of $x^*_i$ among $s_{1:B_{i+1}}$ must be at most
\[
    \min\left\{\left\lfloor k/2^{\lfloor\log_2 k\rfloor - i}\right\rfloor + m/2, B_{i+1}\right\}
\le \min\left\{\left\lfloor k/2^{\lfloor\log_2 k\rfloor - i}\right\rfloor, B_{i+1}\right\} + m/2
=   T_i + m/2.
\]
This implies event $E_2$. 

To violate event $E_2'$, the following must be true:
\begin{itemize}
    \item Event $E_{2, a}$: $B_{i+1} > \left\lfloor k/2^{\lfloor\log_2 k\rfloor - i}\right\rfloor + m/2$, and  strictly less than $T_{i-1}$ elements among $x_1, x_2, \ldots, x_{\left\lfloor k/2^{\lfloor\log_2 k\rfloor - i}\right\rfloor + m/2}$ appear in $s_{1:B_i}$.
\end{itemize}
To see this, if $B_{i+1} \le \left\lfloor k/2^{\lfloor\log_2 k\rfloor - i}\right\rfloor + m/2$, event $E'_2$ would (vacuously) happen: among the largest $B_{i+1}$ elements in $s_{1:B_{i+1}}$, at least $T_{i-1}$ of them appear in $s_{1:B_i}$, which follows from $T_{i-1} \le B_i$. In the remaining case that $B_{i+1} > \left\lfloor k/2^{\lfloor\log_2 k\rfloor - i}\right\rfloor + m/2$, the second part of $E_{2, a}$ is exactly the negation of $E_{2}'$.

Now, we upper bound $\pr{}{E_{2, a}}$. Again, \Cref{lemma:uniform-subset} implies that
\[
    \left|s_{1:B_i} \cap \left\{x_1, \ldots, x_{\left\lfloor k/2^{\lfloor\log_2 k\rfloor - i}\right\rfloor + m/2}\right\}\right|
\]
follows the binomial distribution $\Binomial\left(\left\lfloor k/2^{\lfloor\log_2 k\rfloor - i}\right\rfloor + m/2, 1/2\right)$.  By a Chernoff bound, we have
\begin{align*}
  \pr{}{E_{2, a}}
&=  \pr{X \sim \Binomial(\left\lfloor k/2^{\lfloor\log_2 k\rfloor - i}\right\rfloor + m/2, 1/2)}{X < T_{i-1} }\\
&\le\pr{X \sim \Binomial(\left\lfloor k/2^{\lfloor\log_2 k\rfloor - i}\right\rfloor + m/2, 1/2)}{X  \le \left\lfloor k/2^{\lfloor\log_2 k\rfloor - i}\right\rfloor / 2}\\
&=  \pr{X \sim \Binomial(\left\lfloor k/2^{\lfloor\log_2 k\rfloor - i}\right\rfloor + m/2, 1/2)}{\left(\left\lfloor k/2^{\lfloor\log_2 k\rfloor - i}\right\rfloor / 2 + m / 4\right) - X \ge m / 4 }\\
&\le\exp\left(-\frac{2\cdot (m / 4)^2}{\left\lfloor k/2^{\lfloor\log_2 k\rfloor - i}\right\rfloor+ m / 2}\right)\\
&\le \exp\left(-\frac{m^2 / 8}{2 \left\lfloor k/2^{\lfloor\log_2 k\rfloor - i}\right\rfloor}\right) \notag \\
&\le \exp\left(-\frac{m^2}{16k/2^{\lfloor\log_2k\rfloor-i}}\right),
\end{align*}
where the first inequality is due to $  \left\lfloor k/2^{\lfloor\log_2 k\rfloor - i}\right\rfloor  \ge 2\cdot\left\lfloor k/2^{\lfloor\log_2 k\rfloor - (i-1)}\right\rfloor \ge 2\cdot T_{i-1}$, and the third inequality is due to $m/2 \le T_i \le \left\lfloor k/2^{\lfloor\log_2 k\rfloor - i}\right\rfloor$.

Therefore, we have
\[
    \pr{}{E_2}
\ge \pr{}{E'_2}
\ge 1 - \pr{}{E_{2,a}}
\ge 1 - \exp\left(-\frac{m^2}{16k/2^{\lfloor\log_2k\rfloor-i}}\right).
\]

\paragraph{Put everything together.} By the union bound, for each stage $i$, the probability that  $x_i^*$'s rank among $s_{1:B_{i+1}}$ is outside the range $[T_i - m/2, T_i + m/2]$ is at most
\[
    (1 - \pr{}{\event_1}) + (1 - \pr{}{\event_2})
\le 2\exp\left(-\frac{m^2}{32k/2^{\lfloor\log_2k\rfloor-i}}\right) + 2^{-m/2}.
\]
Applying another union bound again over all stages $i \in \{1, 2, \ldots, \lfloor\log_2k\rfloor\}$ proves the lemma.
\end{proof}

\begin{lemma}\label{lemma:exact-cond-23}
    Condition~2 in \Cref{def.egood} holds with probability at least
    \[
        1 - \lfloor\log_2 k \rfloor\cdot 4e^{2/3}\cdot e^{-m/12},
    \]
   and Condition~3 holds with probability at least
    \[
        1 - \lfloor\log_2 k \rfloor\cdot 2e^{-m/12}.
    \]
\end{lemma}

\begin{proof}
We start with the claim regarding Condition~2 in \Cref{def.egood}. For Condition~2 to be violated, there must be a stage $i$ such that $T_{i-1} > m$, and strictly less than $m/2$ elements in $s_{B_i +1:B_{i+1}}$ are between the $(T_{i-1}-m+1)$-th and the $ T_{i-1}$-th largest elements in $s_{1:B_i}$. We fix a stage~$i$ and upper bound the probability of the event that Condition~$2$ is violated in stage~$i$. Note that we may assume that $\left\lfloor k/2^{\lfloor\log_2 k\rfloor - (i-1)}\right\rfloor > m$: since $T_{i-1} \le \left\lfloor k/2^{\lfloor\log_2 k\rfloor - (i-1)}\right\rfloor$, for $T_{i-1} > m$ to hold, we must have $\left\lfloor k/2^{\lfloor\log_2 k\rfloor - (i-1)}\right\rfloor > m$.

Recall that $T_{i-1}$ is the smaller value between $\left\lfloor k/2^{\lfloor\log_2 k\rfloor - (i-1)}\right\rfloor$ and $B_i$. For Condition~2 in \Cref{def.egood} to be violated at stage~$i$, at least one of the following two events must happen:
\begin{itemize}
    \item Event $E_a$: $B_i \ge \left\lfloor k/2^{\lfloor\log_2 k\rfloor - (i-1)}\right\rfloor$, and there are strictly less than $m/2$ elements in $s_{B_{i} + 1: B_{i+1}}$ between the $\left(\left\lfloor k/2^{\lfloor\log_2 k\rfloor - (i-1)}\right\rfloor - m+1\right)$-th largest  and the $\left\lfloor k/2^{\lfloor\log_2 k\rfloor - (i-1)}\right\rfloor$-th largest elements in $s_{1:B_i}$.
    \item Event $E_b$: $B_i > m$, and there are strictly less than $m/2$ elements in $s_{B_i + 1: B_{i+1}}$ between the $(B_i - m+1)$-th largest and the $B_i$-th largest elements in $s_{1:B_i}$.
\end{itemize}
In the following, we control the probabilities of events $E_a$ and $E_b$.

\paragraph{Upper bound $\pr{}{E_a}$.}
List $s_{1:B_{i+1}}$ in decreasing order:  $x_1> x_2 > \cdots> x_{B_{i+1}}$. For each $i_1 \in [B_{i+1}]$, conditioning on that $x_{i_1}$ is the $\left(\left\lfloor k/2^{\lfloor\log_2 k\rfloor - (i-1)}\right\rfloor - m+1\right)$-th largest element in $s_{1:B_i}$, event $E_a$ implies at least one of the following two:
\begin{itemize}
    \item Event $E_{a,1}$: $B_{i+1} - i_1 \ge \frac{3m}{2}$, and at least $(m-1)$ elements among $ x_{i_1 + 1}, \ldots, x_{i_1 + 3m/2}$ are in $s_{1:B_i}$.
    \item Event $E_{a,2}$: $B_{i+1} - i_1 < \frac{3m}{2}$, and at least $(m-1)$ elements among $ x_{i_1 + 1}, \ldots, x_{B_{i+1}}$ are in $s_{1:B_i}$.
\end{itemize}
To see this, let $x_{i_2}$ be the $\left\lfloor k/2^{\lfloor\log_2 k\rfloor - (i-1)}\right\rfloor$-th largest element among $s_{1:B_i}$; such an element must exist, since event $E_a$ requires $B_i \ge \left\lfloor k/2^{\lfloor\log_2 k\rfloor - (i-1)}\right\rfloor$. Among $x_{i_1}, x_{i_1 + 1}, \ldots, x_{i_2}$, there are exactly $m$ elements in $s_{1:B_i}$ along with $(i_2 - i_1 + 1) - m = i_2 - i_1 - m+1$ elements in $s_{B_i + 1: B_{i+1}}$. Event $E_a$ would then imply that $i_2 - i_1 - (m-1) < m/2$, which is equivalent to $i_2 < i_1 + 3m/2-1$. It follows that, among $x_{i_1 + 1}, x_{i_1 + 2}, \ldots, x_{\min\{i_1 + 3m/2, B_{i+1}\}}$, at least $m$ of them must be in $s_{1:B_i}$, which further implies either $E_{a,1}$ or $E_{a,2}$.

We start with the probability of event $E_{a,1}$. By \Cref{lemma:uniform-subset}, conditioning on that $x_{i_1}$ is the $\left(\left\lfloor k/2^{\lfloor\log_2 k\rfloor - (i-1)}\right\rfloor - m+1\right)$-th largest element in $s_{1:B_i}$, each element among $ x_{i_1 + 1}, \ldots, x_{i_1 + 3m/2}$ still independently appears in $s_{1:B_i}$ with probability $1/2$. It then follows from a Chernoff bound that
\begin{align*}
  \pr{}{E_{a,1}}
&=  \pr{X \sim \Binomial(3m/2, 1/2)}{X \ge m-1}\\
&\le\pr{X \sim \Binomial(3m/2, 1/2)}{X -3m/4 \ge m/4-1}\\
&\le\exp\left(-\frac{2\cdot(m / 4 -1)^2}{3m/2}\right)
\le e^{2/3}\cdot   e^{-m/12}.
\end{align*}
Similarly, to upper bound the probability of event $E_{a,2}$, we note that \Cref{lemma:uniform-subset} implies that, conditioning on that $x_{i_1}$ is the $\left(\left\lfloor k/2^{\lfloor\log_2 k\rfloor - (i-1)}\right\rfloor - m+1\right)$-th largest element in $s_{1:B_i}$, each element among $ x_{i_1 + 1}, \ldots, x_{B_{i+1}}$ independently appears in $s_{1:B_i}$ with probability $1/2$. It again follows from a Chernoff bound that
\begin{align*}
  \pr{}{E_{a,2}}
&=  \pr{X \sim \Binomial(B_{i+1}-i_1, 1/2)}{X \ge m-1}\\
&\le\pr{X \sim \Binomial(3m/2, 1/2)}{X -3m/4 \ge m/4-1} \tag{$B_{i+1} - i_1 < 3m/2$}\\
&\le\exp\left(-\frac{2\cdot(m / 4 -1)^2}{3m/2}\right)
\le e^{2/3}\cdot e^{-m/12}.
\end{align*}
By the union bound, we have $\pr{}{E_a} \le \pr{}{E_{a,1}} + \pr{}{E_{a,2}} \le 2e^{2/3}\cdot e^{-m/12}$.

\paragraph{Upper bound $\pr{}{E_b}$.}
List $s_{1:B_{i+1}}$ in decreasing order:  $x_1> x_2 > \cdots> x_{B_{i+1}}$.
Conditioning on that $x_{i_2}$ is the smallest element in $s_{1:B_i}$, for event $E_b$ to happen, at least one of the following two must be true:
\begin{itemize}
    \item Event $E_{b,1}$: $i_2 > \frac{3m}{2}$, and at least $m-1$ elements among $x_{i_2- 3m/2}, x_{i_2- 3m/2+1},  \ldots, x_{i_2-1}$ are in $s_{1:B_i}$.
    \item Event $E_{b,2}$: $i_2 \le \frac{3m}{2}$, and at least $m-1$ elements among $x_{1}, x_2, \ldots, x_{i_2-1}$ are in $s_{1:B_i}$.
\end{itemize}

To see the above argument, let $x_{i_1}$ be the $\left(B_i - m+1\right)$-th largest element among $s_{1:B_i}$; such an element must exist, since event $E_b$ requires $B_i > m$. Among $x_{i_1}, x_{i_1 + 1}, \ldots, x_{i_2}$, there are exactly $m $ elements in $s_{1:B_i}$ along with $(i_2 - i_1 + 1) - m = i_2 - i_1 - m+1$ elements in $s_{B_i + 1: B_{i+1}}$. Event $E_b$ would then imply that $i_2 - i_1 - m +1< m/2$, which is equivalent to $i_1 > i_2 - 3m/2 + 1$. It follows that, among $x_{i_2 -1}, x_{i_2-2}, \ldots, x_{\max\{i_2 - 3m/2, 1\}}$, at least $m$ of them must be in $s_{1:B_i}$, which further implies either $E_{b,1}$ or $E_{b,2}$.

Conditioning on that $x_{i_2}$ is the smallest element in $s_{1:B_i}$, by a similar application of \Cref{lemma:uniform-subset}, each element among $x_{i_2 - 1}, x_{i_2 - 2}, \ldots, x_1$ independently appears in $s_{1:B_i}$. It then follows from a Chernoff bound that both $\pr{}{E_{b,1}}$ and $\pr{}{E_{b,2}}$ are upper bounded by $e^{2/3}\cdot e^{-m/12}$.

Finally, applying the union bound over all the stages proves the claim for Condition~$2$.

\paragraph{Condition~3.} The analysis for Condition~3 is similar and simpler. For Condition~$3$ in \Cref{def.egood} to be violated, there must be a stage $i$ such that $T_{i-1} + m \le B_i$ and strictly less than $m/2$ elements in $s_{B_i+1:B_{i+1}}$ are between the $T_{i-1}$-th and the $(T_{i-1} + m)$-th largest element in $s_{1:B_i}$. Recall that $T_{i-1}$ is the smaller value between $B_i$ and $\left\lfloor k/2^{\lfloor\log_2 k\rfloor - (i-1)}\right\rfloor$. When $T_{i-1} = B_i$, the above cannot hold, since we would have $T_{i-1} + m = B_i + m > B_i$. Thus, we will focus on the case that $T_{i-1} = \left\lfloor k/2^{\lfloor\log_2 k\rfloor - (i-1)}\right\rfloor$ in the following.

For Condition~3 in \Cref{def.egood} to be violated at stage~$i$, the following event must happen:
\begin{itemize}
    \item Event $E_a'$: $B_i\ge \left\lfloor k/2^{\lfloor\log_2 k\rfloor - (i-1)}\right\rfloor + m$, and there are strictly less than $m/2$ elements in $s_{B_i + 1: B_{i+1}}$ between the $\left\lfloor k/2^{\lfloor\log_2 k\rfloor - (i-1)}\right\rfloor $-th largest  and the $\left(\left\lfloor k/2^{\lfloor\log_2 k\rfloor - (i-1)}\right\rfloor + m\right) $-th largest elements in $s_{1:B_i}$.
\end{itemize}
To upper bound $\pr{}{E'_a}$, list $s_{1:B_{i+1}}$ in decreasing order:  $x_1> x_2 > \cdots> x_{B_{i+1}}$. Conditioning on that $x_{i_1}$ is the $\left\lfloor k/2^{\lfloor\log_2 k\rfloor - (i-1)}\right\rfloor $-th largest element in $s_{1:B_i}$, event $E_{a'}$ implies that at least one of the following two events happens:
\begin{itemize}
    \item Event $E_{a,1}'$: $B_{i+1} - i_1 \ge \frac{3m}{2}$, and at least $m$ elements among $ x_{i_1 + 1}, \ldots, x_{i_1 + 3m/2}$ are in $s_{1:B_i}$.
    \item Event $E_{a,2}'$: $B_{i+1} - i_1 < \frac{3m}{2}$, and at least $m$ elements among $ x_{i_1 + 1}, \ldots, x_{B_{i+1}}$ are in $s_{1:B_i}$.
\end{itemize}
This follows from an argument identical to the one for proving $E_a \implies E_{a,1} \vee E_{a,2}$ for Condition~2. Then, by similar applications of Chernoff bounds,
\[
    \pr{}{E'_a} \le \pr{}{E'_{a,1}} + \pr{}{E'_{a,2}} \le 2e^{-m/12}.
\]
Applying the union bound over all stages proves the claim for Condition~3.
\end{proof}

\begin{lemma}\label{lemma:exact-cond-4}
    Condition 4 in \Cref{def.egood} holds with probability at least \[
    1 - \lfloor \log_2 k\rfloor \cdot e^{-m/4}.
    \]
\end{lemma}
\begin{proof}
    For Condition 4 to be false, there must exist a stage $i = i_0 $ such that $T_i \ge 2m$ and $T_{i-1}\le \frac{m}{2}$. Fix $i \in \{1, 2, \ldots, \lfloor \log_2k \rfloor\}$. We will upper bound the probability of the event that both $T_i \ge 2m$ and $T_{i-1} \le m/2$ hold.
    
    Recall that $T_i = \min\left\{\left\lfloor k/2^{\lfloor\log_2 k\rfloor -i}\right\rfloor, B_{i+1}\right\}$. For $T_i \ge 2m$ to hold, we must have $B_{i+1} \ge 2m$ and $\left\lfloor k/2^{\lfloor\log_2 k\rfloor -i}\right\rfloor \ge 2m$. It follows that
    \[
        \left\lfloor k/2^{\lfloor\log_2 k\rfloor -(i-1)}\right\rfloor \ge \frac{1}{2}\cdot \left(\left\lfloor k/2^{\lfloor\log_2 k\rfloor -i}\right\rfloor - 1\right) \ge \frac{1}{2}\cdot 2m = m.
    \]
    Then, we must have $B_i \le m / 2$; otherwise, we would have
    \[
        T_{i-1}
    =   \min\left\{\left\lfloor k/2^{\lfloor\log_2 k\rfloor -(i-1)}\right\rfloor, B_i\right\}
    \ge \min\{m, m/2 + 1\}
    =   m/2+1,
    \]
    which contradicts $T_{i-1} \le m/2$.

    So far, we have proved that $T_{i-1} \le m / 2$ and $T_i > 2m$ together imply $B_i \le m/2$ and $B_{i+1} \ge 2m$. Recall that, conditioning on the value of $B_{i+1}$, $B_i$ follows the binomial distribution $\Binomial(B_{i+1}, 1/2)$. Therefore, we have
    \begin{align*}
        \pr{}{T_{i-1} \le m/2 \wedge T_i \ge 2m}
    &\le\pr{}{B_i \le m/2 \wedge B_{i + 1} \ge 2m}\\
    &\le\pr{X \sim \Binomial(2m, 1/2)}{X \le m/2} \\
    &\le\exp\left(-\frac{2\cdot (m/2)^2}{2m}\right) 
    =   e^{-m/4}. \tag{Chernoff bound}
    \end{align*}
    Applying the union bound over $i = 1, 2, \ldots, \lfloor \log_2 k\rfloor$ proves the lemma.
\end{proof}

\begin{lemma}\label{lemma:exact-correctness}
    When $\goodevent$ happens, \Cref{alg:exact-selection} outputs the $k$-th largest element among $s_1, s_2, \ldots, s_n$.
\end{lemma}
\begin{proof}
  We first note that the following invariants hold for the array $M$ throughout \Cref{alg:exact-selection}:
    \begin{enumerate}
        \item[(1)] Finite entries (i.e., everything except $\pm\infty$ and $\bot$) form a contiguous subsequence in $M$ (denoted by $M[l], M[l + 1],  \ldots, M[r]$).
        \item[(2)] Finite entries are in decreasing order, and contain every element between $M[l]$ and $M[r]$ that has appeared in the stream so far.
        \item[(3)] On either side (i.e., among $M[1], M[2], \ldots, M[l-1]$ and $M[r+1], M[r+2], \ldots, M[3m]$), there might be multiple copies of $\bot$, $+\infty$, or $-\infty$, but only one of the three on each side.
        \item[(4)] There are $+\infty$ on the left side only if $M[l]$ is the largest element so far.
        \item[(5)] There are $-\infty$ on the right side only if $M[r]$ is the smallest element so far.
    \end{enumerate}

    In the following, we will prove by induction that, under event $\goodevent$, the following two additional invariants hold at the end of each stage $i \ge i_0 = \min\{i \in \{0, 1, \ldots, \lfloor \log_2k\rfloor\}: T_i > m/2\}$:
    \begin{itemize}
        \item $M[3m/2]$ holds the $T_i$-th largest element among $s_{1:B_{i+1}}$.
        \item On either side of $M$, if ``$\bot$'' appears, there are at most $m/2$ copies.
    \end{itemize}
    Note that the above implies that, at the end of the last stage $i = \lfloor \log_2 k \rfloor$, $M[3m/2]$ holds the $T_i$-th largest element among the entire stream, where $T_i = \min\{B_{i+1}, \lfloor k / 2^{\lfloor\log_2 k\rfloor - i}\rfloor\} = \min\{n, k\} = k$. Therefore, the algorithm outputs the correct answer.
    
    \paragraph{The base case.} We start with the base case that $i = i_0$. In Line~\ref{line:base-case}, before the start of stage~$i_0+1$, \Cref{alg:exact-selection} stores the largest $\min\{B_{i_0 + 1}, 3m-1\}$ elements in decreasing order in array $M$ (starting from $M[2]$). By Condition~4 in \Cref{def.egood}, we have $T_{i_0} < 2m$, which implies $T_{i_0} + 1 \le 2m \le 3m$. Thus, $M[T_{i_0} + 1]$ holds the $T_{i_0}$-th largest element among $s_{1:B_{i_0+1}}$. Then, in Line~\ref{line:base-case-shift}, the algorithm shifts the array by $3m/2 - T_{i_0}-1$, so that the $T_{i_0}$-th largest element is moved to $M[3m/2]$. This verifies the first invariant.
    
    For the second invariant, we consider the direction of the shifting in Line~\ref{line:base-case-shift}. If the shifting is to the right (i.e., $T_{i_0} + 1 < 3m/2$), since $M[1] = +\infty$, the left side of $M$ would be filled with more copies of $+\infty$. Furthermore, the right side of $M$ either contains finite entries exclusively, or consists of both finite entries and one or more copies of $-\infty$. In either case, the second invariant is satisfied.

    If the shifting is to the left (i.e., $T_{i_0} + 1 > 3m/2$), the left side of $M$ would only contain finite entries.  For the right side, if $M[3m] = -\infty$ before the shifting, the shifting would only introduce more copies of $-\infty$. If $M[3m]$ held a finite element before the shifting, at most $T_{i_0} + 1 - 3m/2 \le 2m - 3m/2 = m/2$ copies of ``$\bot$'' are introduced during the shifting. In either case, we would have the second invariant.
    
    \paragraph{The inductive step.} We prove the inductive step in the following way: First, we apply Conditions 2~and~3 of the good event $\goodevent$ to show that, before the shifting in Line~\ref{line:inductive-step-shift} is conducted, both sides of $M$ are filled with either finite entries exclusively, or both finite entries and copies of $\pm\infty$. In other words, $M$ does not contain the the empty entry ``$\bot$''. Then, we use Condition~1 of $\goodevent$ to show that, before the shifting, the rank of $M[3m/2]$ among $s_{1:B_{i+1}}$ is within $T_i \pm m/2$. Finally, we show that we have both invariants after the shifting is performed.
    \begin{enumerate}
        \item[(1)]  \textbf{The left side either is completely filled or contains $+\infty$.} Consider the state of the left side of $M$ at the beginning of stage~$i$. If it contained one or more $+\infty$ elements, after reading the elements in stage~$i$, it still only contains finite entries and $+\infty$ (but not $\bot$).
    
        If the left side did not have $+\infty$ at the beginning of stage~$i$, by the inductive hypothesis, it might have contained at most $m/2$ copies of $\bot$. Then, $M[m/2+1], M[m/2+2], \ldots, M[3m/2]$ must be consecutive entries among $s_{1:B_i}$, with the $T_{i-1}$-th largest element among $s_{1:B_i}$ at index $3m/2$. Thus, we have $T_{i-1} \ge 3m/2 - (m/2+1) + 1 = m$. By Condition~2 of $\goodevent$ (\Cref{def.egood}), there are at least $m/2$ elements at stage $i$ that are within the range of the $[T_{i-1}-m+1, T_{i-1}]$-th largest element at the end of stage $i-1$, so the left buffer is filled to full. 
    \item[(2)] \textbf{The right side either is completely filled or contains $-\infty$.}  Similarly, consider the state of the right side of $M$ at the beginning of stage~$i$. If it contained $-\infty$ elements, it would still only contain finite elements and $-\infty$ after reading the elements in stage~$i$.
    
    If the right side of $M$ did not contain $-\infty$ at the beginning, by the inductive hypothesis, it might have contained $\le m/2$ copies of ``$\bot$''. Then, $M[3m/2], M[3m/2+1], \ldots, M[5m/2]$ must be consecutive entries among $s_{1:B_i}$, with the $T_{i-1}$-th largest element at index $3m/2$. It follows that $T_{i-1} + m \le B_i$. By Condition~3 of $\goodevent$ (\Cref{def.egood}),  there are at least $m/2$ elements in $s_{B_i+1:B_{i+1}}$ that are between the $T_{i-1}$-th and the $(T_{i-1}+m)$-th largest elements among $s_{1:B_i}$. It follows that, after reading $s_{B_i+1}$ through $s_{B_{i+1}}$, the right side of $M$ is filled with finite elements. 
    \item[(3)] \textbf{The rank of $M[3m/2]$ among $s_{1:B_{i+1}}$ is in $[T_i - m/2, T_i + m/2]$.} By the induction hypothesis, at the beginning of stage~$i$, $M[3m/2]$ holds the $T_{i-1}$-th largest element among $s_{1:B_i}$. By Condition~1 in $\goodevent$ (\Cref{def.egood}), this element has a rank between $T_i - m/2$ and  $T_i + m/2$ (inclusive) among $s_{1:B_{i+1}}$. Therefore, after reading the $B_{i+1} - B_i$ new elements in stage~$i$, the value of the variable $\rank$ is in $[T_i - m/2, T_i + m/2]$. 
    \item[(4)] \textbf{The first invariant holds after shifting.} Now, we verify the invariants at the end of stage~$i$. Before the shifting in Line~\ref{line:inductive-step-shift}, the rank of $M[3m/2]$ among $s_{1:B_{i+1}}$ is tracked by $\rank \in [T_i - m/2, T_i + m/2]$. We claim that the $T_i$-th largest element must be held inside the array $M$. Assuming this claim, after the shifting in Line~\ref{line:inductive-step-shift}, that element will be moved to $M[3m/2]$. This proves the first invariant.
    
    To prove the claim, suppose that $\rank < T_i$; the case that $\rank > T_i$ is almost symmetric. Since $3m/2 + (T_i - \rank) \le 3m/2 + m/2 = 2m$, $3m/2 + (T_i - \rank)$ is a valid index for array $M$. Furthermore, $M[3m/2 + (T_i - \rank)]$ cannot be $\bot$; otherwise the right side of $M$ would contain at least $m + 1 > m/2$ copies of $\bot$. $M[3m/2 + (T_i - \rank)]$ cannot be $-\infty$ either; otherwise, the largest element among $s_{1:B_{i+1}}$ must be stored in $M[r]$ for some $r \in [3m/2, 3m/2 + (T_i - \rank))$. It follows that
    \[
        T_i - \rank
    =   3m/2 + (T_i - \rank) - 3m/2
    >   r - 3m/2
    =   B_{i+1} - \rank
    \ge T_i - \rank,
    \]
    a contradiction. Therefore, $M[3m/2 + (T_i - \rank)]$ must be a finite element, and its rank among $s_{1:B_{i+1}}$ is given by $\rank + (T_i - \rank) = T_i$. This proves the claim that the $T_i$-th largest element must be stored in $M$.

    \item[(5)] \textbf{The second invariant holds after shifting.} To verify the second invariant, recall that, before the shifting in Line~\ref{line:inductive-step-shift}, each side of $M$ either only contains finite elements, or consists of only finite elements and $\pm\infty$. In Line~\ref{line:inductive-step-shift}, the array $M$ gets shifted by at most $|\rank - T_i| \le m/2$. Then, on one of the two sides of $M$, we fill the empty spaces with either $\pm\infty$ or $\bot$. In either case, the second invariant is maintained.
    \end{enumerate}
    This concludes the inductive step and finishes the proof.
\end{proof}

\quantileexact*

\begin{proof}
By the union bound along with \Cref{lemma:exact-cond-1,lemma:exact-cond-23,lemma:exact-cond-4}, event $\goodevent$ happens with probability at least
\[
    1- 12\lfloor\log_2 k \rfloor\cdot \exp\left(-\min\left\{\frac{1}{12},\frac{\ln 2}{2}, \frac{1}{4}\right\}\cdot m\right) - 2\sum_{i=0}^{\lfloor\log_2k\rfloor-1}\exp\left(-\frac{m^2}{32k/2^i}\right). 
\]
The theorem then follows from \Cref{lemma:exact-correctness}.
\end{proof}

\section{From Quantile Estimation to $k$-Secretary}\label{section.3}
In this section, we prove \Cref{prop:reduction} (restated below) by showing that an algorithm for the quantile estimation problem (\Cref{def.quantile_estimation}) can be transformed into a competitive algorithm for $k$-secretary (\Cref{def.ksecretary}), with a mild additive increase in the space usage. Concretely, if the quantile estimation algorithm has an $O(k^{\alpha})$ rank error in expectation, the resulting $k$-secretary algorithm achieves a competitive ratio of $1 - O(k^{\alpha} / k)$.

\reduction*

In \Cref{sec:reduction-weaker}, we prove a weaker version of \Cref{prop:reduction}, in which the memory usage increases by a factor of $O(\log k)$. We then derive the actual version from the weaker reduction in \Cref{sec:reduction-stronger}.

\subsection{A Weaker Reduction}\label{sec:reduction-weaker}
We start by stating the weaker reduction below. Recall the definition of comparison-based algorithms from \Cref{def:comparison-based}.

\begin{proposition}[Weaker version of \Cref{prop:reduction}]\label{prop:reduction-weaker}
    Suppose that, for some $\alpha \in [1/2, 1]$, there is a comparison-based quantile estimation algorithm with memory usage $m$ and an error of $O(k^{\alpha})$ in expectation. Then, there is a $k$-secretary algorithm that uses $O(m \log k)$ memory and achieves a competitive ratio of
    \[
        1 -O\left(\frac{1}{k^{1 - \alpha}}\right).
    \]
\end{proposition}

Compared to \Cref{prop:reduction}, the only change in the proposition above is that the memory bound gets relaxed from $m + O(1)$ to $O(m\log k)$.

We prove \Cref{prop:reduction-weaker} by constructing a $k$-secretary algorithm (\Cref{algo:k-secretary}) that applies the quantile estimator as a black box. Note that this algorithm is almost identical to the algorithm of~\cite{Kleinberg05}, except that the straightforward algorithm for finding the $k$-th largest element---which requires $\Omega(k)$ memory---is replaced by algorithm $\A$, which uses much less memory but only finds an approximately $k$-th largest element.

Another minor change is that we take the first $B \coloneqq \lfloor n/2\rfloor$ elements of the sequence as ``the first half''. In contrast, Kleinberg's algorithm draws $B$ from $\Binomial(n, 1/2)$ randomly, so that $\{s_1, s_2, \ldots, s_B\}$ would be uniformly distributed among all subsets of $\{s_1, s_2, \ldots, s_B\}$. This makes the subsequent concentration argument a bit easier. In our proof, we decided against choosing $B$ randomly because, in that case, the value of $B$ factors into the realization of $\rank_{1:B}(x^*)$ (via the black box quantile estimator). Then, in our analysis, the conditioning on the value of $\rank_{1:B}(x^*)$ would then bias the distribution of $B$ and renders the conditional distribution of $\{s_1, s_2, \ldots, s_B\}$ non-uniform.

\begin{algorithm2e}
    \caption{$\choosetopk(n, k, s)$}
    \label{algo:k-secretary}
    \KwIn{String length $n$, target rank $k$, access to random-order sequence $s = (s_1, s_2, \ldots, s_n)$. Quantile estimation algorithm $\A$.}
    \KwOut{At most $k$ accepted elements in $s$.}
    \lIf{$k = 1$}
    {Run the $(1/e)$-competitive algorithm for $1$-secretary}
    $B \leftarrow \lfloor n/2\rfloor$\;
    Run $\choosetopk(B, \lfloor k/2\rfloor, s_{1:B})$ on the first $B$ elements\label{algoksec.recursive}\;
    In parallel with the line above, run $x^* \leftarrow \A\left(B, \left\lfloor k/2\right\rfloor, s_{1:B}\right)$\label{algoksec.findkth}\;
    $\counter \gets 0$\;
    \For{$i = B+1, B+2, \ldots, n$}{
        Read $s_i$\;
        \If{$s_i > x^*$ and $\counter < \lfloor k / 2\rfloor$} {
            \textbf{Accept} $s_i$\;
            $\counter \gets \counter + 1$\;
        }
    }
\end{algorithm2e}

We will analyze $\choosetopk$ (\Cref{algo:k-secretary}) using an inductive proof. The inductive step in the analysis is stated in the following lemma.

\begin{lemma}\label{lemma:k-secretary-inductive}
    Suppose that the following are true for some constant $C_3 \ge 1$: (1) $n \ge k \ge 10$; (2) \Cref{algo:k-secretary} leads to a competitive ratio of at least 
    \[
        1 - \frac{200C_3}{(k')^{1 - \alpha}}
    \]
    in the recursive call $\choosetopk(B, k', s_{1:B})$ on Line~\ref{algoksec.recursive} for $k' = \lfloor k/2 \rfloor$ and $B = \lfloor n/2 \rfloor$; (3) Line~\ref{algoksec.findkth} returns an element $x^*$ that satisfies 
    \[
    \Ex{}{\left|\rank_{1:B}(x^*) - k'\right|} \le C_3\cdot (k')^{\alpha}.
    \]
    Then, on the instance with parameters $n$ and $k$, \Cref{algo:k-secretary} has a competitive ratio of
    \[
    1 - \frac{200C_3}{k^{1-\alpha}}.
    \]
\end{lemma}

We first show how the lemma implies \Cref{prop:reduction-weaker}.

\begin{proof}[Proof of \Cref{prop:reduction-weaker} assuming \Cref{lemma:k-secretary-inductive}]
    On a $k$-secretary instance, \Cref{algo:k-secretary} involves $O(\log k)$ levels of recursion, since each recursive call shrinks the parameter $k$ by a factor of $2$. Furthermore, each recursive call of $\choosetopk$ uses memory $m$ to call the quantile estimation algorithm $\A$ and an additional $O(1)$ words for storing the remaining variables. Therefore, the algorithm uses $O(m\log k)$ memory in total.
    
    Regarding the competitive ratio, when $k \le 10$, since any algorithm is trivially $0$-competitive, the competitive ratio is indeed lower bounded by
    \[
        0
    \ge 1 - \frac{\sqrt{10}}{\sqrt{k}}
    \ge 1 -200C_3\cdot \frac{1}{k^{1 - \alpha}}.
    \]
    The proposition then follows from \Cref{lemma:k-secretary-inductive} and an induction on $k$.
\end{proof}

Now we turn to the proof of \Cref{lemma:k-secretary-inductive}.

\begin{proof}[Proof of \Cref{lemma:k-secretary-inductive}]
We analyze the execution of $\choosetopk(n, k, s_{1:n})$ on the random-order sequence $s_{1:n}$. Let $x_1 > x_2 > \cdots > x_n$ denote the elements of $s_{1:n}$ when sorted in descending order. Let $\SOL$ denote the expected sum of the elements accepted by \Cref{algo:k-secretary}, over the randomness in both the algorithm and the random ordering. We can decompose $\SOL$ into two parts, $\SOL_1$ and $\SOL_2$, defined as the expected sums of the accepted elements in the first and the second halves, respectively:
\begin{align*}
        \SOL
&=      \SOL_1 + \SOL_2\\
&\ge    \CR_{k'} \cdot\sum_{i=1}^k x_i\cdot \pr{}{\text{$x_i$ is among the $k'$ largest elements in the first half}}\\
&\quad + \sum_{i=1}^k x_i\cdot \pr{}{\text{$x_i$ is among the second half and accepted}},
\end{align*}
where $k' = \lfloor k/2\rfloor$, and $\CR_{k'} \coloneqq 1 - 200C_3 / (k')^{1 - \alpha}$ denotes the lower bound on the competitive ratio of the algorithm on an instance of $k'$-secretary.

In the rest of the proof, we will show that both $\SOL_1$ and $\SOL_2$ are lower bounded by $\OPT \cdot (1/2 - O(1/k^{1 - \alpha}))$, so the entire algorithm has a competitive ratio of $1 - O(1/k^{1-\alpha})$.

\paragraph{Lower bound the first part.} Define binary random variables $X_1, X_2, \ldots, X_n$ such that $X_i = 1$ if and only if $x_i$ is in the first half $s_{1:B}$. Note that $X_1$ through $X_n$ can be viewed as being drawn without replacement from the size-$n$ population $(0, 0, \ldots, 0, 1, 1, \ldots, 1)$ with $B$ copies of $1$ and $n - B$ copies of $0$.

Then, the event ``$x_i$ is among the $k'$ largest elements in the first half'' can be equivalently written as
\[
    X_i = 1 \wedge X_1 + X_2 + \cdots + X_i \le k'.
\]
Let $p_i \coloneqq \pr{}{X_i = 1 \wedge X_1 + X_2 + \cdots + X_i \le k'}$ denote the probability of the event above over the randomness in $X_1, \ldots, X_n$. Note that $p_1 \ge p_2 \ge \cdots \ge p_n$. The lower bound on $\SOL_1$ can then be simplified into
\[
    \SOL_1
\ge \CR_{k'}\cdot\sum_{i=1}^{k}(x_i\cdot p_i)
\ge \CR_{k'}\cdot\frac{1}{k}\cdot\left(\sum_{i=1}^{k}x_i\right)\cdot\left(\sum_{i=1}^{k}p_i\right)
=   \OPT\cdot \CR_{k'}\cdot\frac{1}{k}\sum_{i=1}^{k}p_i,
\]
where $\OPT \coloneqq \sum_{i=1}^{k}x_i$ is the benchmark with which the algorithm competes. The second step above holds since both $x$ and $p$ are non-negative and monotone non-increasing.

It remains to lower bound the term $\sum_{i=1}^{k}p_i$. Note that for any realization of $X_1, \ldots, X_n$, we have
\[
    \sum_{i=1}^k \1{X_i = 1 \wedge X_1 + X_2 + \cdots + X_i \le k'} = \min\left\{\sum_{i=1}^k X_i, k'\right\},
\]
which follows from a case analysis on whether $\sum_{i=1}^k X_i \ge k'$. Taking an expectation on both sides gives
\[
    \sum_{i=1}^k p_i
=   \Ex{X}{\min\left\{\sum_{i=1}^{k}X_i, k'\right\}}.
\]

The right-hand side above can be further lower bounded as follows:
\begin{align*}
        \Ex{X}{\min\left\{\sum_{i=1}^k X_i, k'\right\}}
&\ge    \Ex{X}{\sum_{i=1}^k X_i} - \Ex{X}{\left|\sum_{i=1}^k X_i - k'\right|} \tag{$\min\{a, b\} \ge a - |a-b|$}\\
&\ge    k\cdot \frac{B}{n} - \sqrt{\Ex{X}{\left(\sum_{i=1}^kX_i - k'\right)^2}}. \tag{Jensen's Inequality}
\end{align*}
To control the remaining expectation in the above, we note that the function $x \mapsto (x - k')^2$ is convex, so \Cref{lemma.cite} implies that
\[
    \Ex{X}{\left(\sum_{i=1}^kX_i - k'\right)^2}
\le \Ex{Y}{\left(\sum_{i=1}^kY_i - k'\right)^2},
\]
where $Y_1, Y_2, \ldots, Y_k$ are sampled \emph{with} replacement from the same population (with $B = \lfloor n/2\rfloor$ copies of $1$ and $n - B$ copies of $0$). Equivalently, $Y_i$s are independent samples from $\Bern(B / n)$. It follows that
\[
    \Ex{Y}{\left(\sum_{i=1}^kY_i - k'\right)^2}
=   \Var{Y}{\sum_{i=1}^kY_i} + \left(\Ex{Y}{\sum_{i=1}^kY_i} - k'\right)^2
=   k\cdot\frac{B}{n}\cdot\left(1 - \frac{B}{n}\right) + \left(k\cdot \frac{B}{n} - k'\right)^2.
\]
Note that $(B/n)\cdot (1 - B/n) \le 1/4$. Furthermore, it can be verified that, for all integers $n \ge k \ge 1$,
\[
    \left(k\cdot \frac{B}{n} - k'\right)^2 \le \frac{1}{4}.
\]
Plugging the above back to the lower bound on $\Ex{X}{\min\left\{\sum_{i=1}^{k}X_i, k'\right\}}$ gives
\[
    \Ex{X}{\min\left\{\sum_{i=1}^{k}X_i, k'\right\}}
\ge k\cdot\frac{B}{n} - \sqrt{\frac{k+1}{4}}.
\]
It follows that
\[
    \SOL_1
\ge \OPT\cdot \CR_{k'} \cdot \frac{1}{k}\sum_{i=1}^{k}p_i
\ge \OPT\cdot\CR_{k'}\cdot \left(\frac{B}{n} - \sqrt{\frac{k+1}{4k^2}}\right).
\]
Plugging $B = \lfloor n / 2\rfloor \ge n / 2 - 1/2$ and $\sqrt{\frac{k+1}{4k^2}} \le \sqrt{\frac{4k}{4k^2}} = 1 / \sqrt{k}$ into the above gives
\begin{equation}\label{eq:SOL-1-bound}
    \SOL_1
\ge \OPT\cdot\CR_{k'}\cdot \left(\frac{1}{2} - \frac{1}{2n} - \frac{1}{\sqrt{k}}\right)
\ge \OPT\cdot\CR_{k'}\cdot \left(\frac{1}{2} - \frac{2}{\sqrt{k}}\right).
\end{equation}

\paragraph{The second part after conditioning.} Let $l\coloneqq\rank_{1:B} (x^*)$ denote the actual rank of $x^*$---the output of the quantile estimator on Line~\ref{algoksec.findkth}---among the first half $s_{1:B}$. Let $Z$ denote the number of elements in the second half $s_{(B+1):n}$ that are larger than $x^*$. Note that conditioning on the values of $l$ and $Z$, we have $\rank_{1:n}(x^*) = l + Z$. Equivalently, $x^*$ is given by $x_{l+Z}$.

In the following, we condition on the realization of $(l, Z)$ and examine the contribution to $\SOL_2$, namely, the conditional expectation
\[
    \sum_{i=1}^{k}x_i\cdot\pr{}{x_i\text{ is among the second half and accepted} \mid l, Z}.
\]
Later in the proof, we will take an expectation over the randomness in $(l, Z)$. 

Specifically, we consider four cases, depending on whether $Z \le k'$ and $l + Z - 1 \ge k$ hold:
\begin{enumerate}
    \item[\textbf{Case 1.}] $Z \le k'$. In this case, at most $k' = \lfloor k/2\rfloor$ elements among $s_{B+1}, s_{B+2}, \ldots, s_n$ exceed the threshold $x^*$. Then, by \Cref{algo:k-secretary}, all those elements are accepted. In other words, an element $x_i$ counts towards $\SOL_2$ as long as: (1) $i \le k$; (2) $i \le l + Z - 1$, so that $x_i > x_{l+Z} = x^*$. In the following, we consider two sub-cases, depending on whether $k$ or $l + Z - 1$ is larger.

    \begin{enumerate}
        \item[\textbf{Case 1a.}] $Z \le k'$ and $l + Z - 1 \ge k$. In this case, each element $x_i$ contributes to $\SOL_2$ as long as $x_i$ is among the second half of the sequence. Note that, given the values of $l$ and $Z$, it holds that $\rank_{1:B}(x_{l+Z}) = \rank_{1:B}(x^*) = l$, which implies $|\{s_1, s_2, \ldots, s_B\} \cap \{x_1, x_2, \ldots, x_{l+Z-1}\}| = l - 1$ and $|s_{(B+1):n} \cap \{x_1, x_2, \ldots, x_{l+Z-1}\}| = Z$. Furthermore, conditioning on the realization of $(l, Z)$, $s_{(B+1):n} \cap \{x_1, x_2, \ldots, x_{l+Z-1}\}$ is still uniformly distributed among all size-$Z$ subsets of $\{x_1, x_2, \ldots, x_{l+Z-1}\}$. In particular, each of $x_1, x_2, \ldots, x_{l+Z-1}$ appears in the second half with a conditional probability of $\frac{Z}{l+Z-1}$. It follows that the conditional contribution to $\SOL_2$ is given by
        \begin{align*}
            &~\sum_{i=1}^{k}x_i\cdot\pr{}{x_i\text{ is among the second half and accepted} \mid l, Z}\\
        =   &~\sum_{i=1}^{k}x_i\cdot\frac{Z}{l+Z-1}
        =   \OPT \cdot \frac{Z}{l+Z-1}.
        \end{align*}
        \item[\textbf{Case 1b.}] $Z \le k'$ and $l + Z -1 < k$. This case is similar to Case~1a, except that only the elements $x_1, x_2, \ldots, x_{l+Z-1}$ can contribute to $\SOL_2$, since an element $x_i$ needs to exceed $x^* = x_{l+Z}$ to be accepted. Since $x_1 > x_2 > \cdots > x_n$ are in descending order, we have the inequality
        \[
            \frac{1}{l+Z-1}\sum_{i=1}^{l+Z-1}x_i
        \ge \frac{1}{k}\sum_{i=1}^{k}x_i
        =   \frac{1}{k}\OPT.
        \]
        It follows that the conditional contribution in this case is at least
        \[
            \sum_{i=1}^{l + Z - 1}x_i\cdot \frac{Z}{l + Z - 1}
        =   Z\cdot\frac{1}{l+Z-1}\sum_{i=1}^{l+Z-1}x_i
        \ge \OPT\cdot \frac{Z}{k}.
        \]
    \end{enumerate}
    \item[\textbf{Case 2.}] $Z > k'$. In this case, the second half of the sequence, $s_{(B+1):n}$, contains $Z > k' = \lfloor k/2 \rfloor$ elements that exceed the threshold $x^*$ (namely, the elements in $\{s_{B+1}, s_{B+2}, \ldots, s_n\} \cap \{x_1, x_2, \ldots, x_{l+Z-1}\}$). Then, $\choosetopk$ would accept the $k'$ such elements that arrive first (out of the $Z$ elements).

    Note that, even after conditioning on the values of $(l, Z)$ as well as the size-$Z$ set
    \[
        \{s_{B+1}, s_{B+2}, \ldots, s_n\} \cap \{x_1, x_2, \ldots, x_{l+Z-1}\},
    \]
    the ordering of these $Z$ elements in the sequence is still uniformly distributed. Therefore, each of the $Z$ elements gets accepted with a probability of $k' / Z$.
    
    Again, to calculate the contribution of this case to $\SOL_2$, we need to separately consider the two cases $l + Z - 1 \ge k$ and $l + Z - 1 < k$.
    \begin{enumerate}
        \item[\textbf{Case 2a.}] $Z > k'$ and $l + Z - 1 \ge k$. In this case, each of $x_1, x_2, \ldots, x_k$ can potentially contribute to $\SOL_2$. For each $i \in [k]$, $x_i$ contributes to $\SOL_2$ if: (1) $x_i$ appears in the second half, which happens with probability $\frac{Z}{l + Z - 1}$, by the same argument as in Case~1a; (2) $x_i$ is among the $k'$ elements in the second half (out of $Z$ in total) that the algorithm accepts. This happens with probability $k'/Z$. Therefore, the conditional contribution to $\SOL_2$ is given by
        \[
            \sum_{i=1}^{k}x_i\cdot\frac{Z}{l + Z - 1}\cdot\frac{k'}{Z}
        =   \frac{k'}{l+Z-1}\cdot\OPT.
        \]
        
        \item[\textbf{Case 2b.}] $Z > k'$ and $l + Z - 1 < k$. Finally, compared to Case~2a, only elements $x_1, x_2, \ldots, x_{l+Z-1}$ can contribute to $\SOL_2$. Applying the inequality
        \[
            \frac{1}{l + Z - 1}\sum_{i=1}^{l+Z-1}x_i
        \ge \frac{1}{k}\sum_{i=1}^{k}x_i
        =   \frac{1}{k}\OPT
        \]
        again shows that the conditional contribution to $\SOL_2$ is at least
        \[
             \sum_{i=1}^{l+Z-1}x_i\cdot\frac{Z}{l + Z - 1}\cdot\frac{k'}{Z}
        =   k'\cdot\frac{1}{l+Z-1}\sum_{i=1}^{l+Z-1}x_i
        \ge \frac{k'}{k}\cdot\OPT.
        \]
    \end{enumerate}  
\end{enumerate}
Summarizing the four cases above, we have that, conditioning on the realization of $(l, Z)$, the conditional contribution to the expectation in $\SOL_2$ is lower bounded by
\[
    \frac{\min\{Z, k'\}}{\max\{l+Z-1, k\}}\cdot \OPT.
\]

\paragraph{Lower bound the second part.} It remains to lower bound the expectation of the ratio
\begin{equation}\label{eq:ratio-of-interest}
    \frac{\min\{Z, k'\}}{\max\{l+Z-1, k\}}
\end{equation}
over the randomness in $l \coloneqq \rank_{1:B}(x^*)$ and $Z \coloneqq |s_{(B+1):n} \cap (-\infty, x^*)|$. Equivalently, $Z$ can be defined as the value such that $l + Z = \rank_{1:n}(x^*)$.

Note that, conditioning on the value of $l$, $l + Z$ is identically distributed as the $l$-th smallest number in a size-$B$ subset of $[n] = \{1, 2, \ldots, n\}$ chosen uniformly at random. Thus, the distribution of $Z \mid l$ can be analyzed using \Cref{lemma:subset-rank-concentration}. Concretely, applying the lemma with parameters
\[
    \tilde n = n, \quad \tilde k = B = \lfloor n/2\rfloor, \quad \tilde i = l
\]
gives
\[
    \Ex{}{\left|l + Z-l\cdot \frac{n}{\lfloor n/2\rfloor}\right|}
\le 2\cdot\frac{n}{\lfloor n/2\rfloor}\cdot \sqrt{l\cdot \frac{n}{\lfloor n/2\rfloor}} + \left(\frac{n}{\lfloor n/2\rfloor}\right)^2
\le 7\sqrt{l} + 5,
\]
where the last step applies $n / \lfloor n/2\rfloor \le 11/5$, which holds for any integer $n \ge 10$.
It follows that
\[
    \Ex{}{\left|Z -l\right|}
\le \Ex{}{\left|l + Z-l\cdot \frac{n}{\lfloor n/2\rfloor}\right|} + \Ex{}{\left|l\cdot \frac{n}{\lfloor n/2\rfloor} - 2l\right|} 
\le 7\sqrt{l} + 5 + \frac{2l}{n-1}
\le 7\sqrt{l} + 7.
\]
The last step above applies $l \le B = \lfloor n/2\rfloor \le n/2$.

To further control the $\sqrt{l}$ term (after taking an expectation over $l$), we note that
\[
    \Ex{}{\sqrt{l}}
\le \sqrt{\Ex{}{l}}
\le \sqrt{k' + \Ex{}{|l - k'|}}
\le \sqrt{k' + C_3\cdot(k')^{\alpha}},
\]
where the last step applies the assumption of the lemma. We can further simplify the above into
\[
    \Ex{}{\sqrt{l}}
\le \sqrt{(C_3 + 1)k'}
\le \sqrt{2C_3k'},
\]
which gives the upper bound
\[
    \Ex{}{|Z - l|}
\le \Ex{}{7\sqrt{l} + 7}
\le 7\cdot\sqrt{2C_3k'} + 7.
\]

For each realization of $(l, Z)$, we can lower bound \Cref{eq:ratio-of-interest} as follows:
\begin{align*}
        &~\frac{\min\{Z, k'\}}{\max\{l+Z-1, k\}}\\
\ge     &~\frac{k' - |k'-Z|}{\max\{l+Z-1, k\}} \tag{$\min\{a, b\} \ge a - |a - b|$}\\
\ge     &~\frac{k' - |Z - l| - |l - k'|}{\max\{l+Z-1, k\}} \tag{triangle inequality}\\
\ge     &~\frac{k'}{\max\{l+Z-1, k\}}-\frac{|Z - l| + |l - k'|}{k} \tag{$\max\{l+Z-1, k\} \ge k$}\\
\ge     &~\frac{k' - |l+Z-1 -k|}{k}-\frac{|Z - l| + |l - k'|}{k}\\
\ge     &~\frac{k'}{k} - \frac{|Z - l| + |2l - 2k'| + |2k' - k - 1|}{k}-\frac{|Z - l| + |l - k'|}{k} \tag{triangle inequality}\\
=       &~\frac{k'}{ k}-\frac{2\cdot|Z - l| + 3\cdot|l - k'|+2}{k}.
\end{align*}
The fourth step above applies the inequality
\[
    \frac{a}{\max\{b, c\}}
\ge \frac{a - |b - c|}{c}
\]
for $a = k'$, $b = l + Z - 1$ and $c = k$. The above, in turn, follows from a case analysis: (1) If $b \le c$, the left-hand side is given by $a / c \ge (a - |b - c|) / c$; (2) If $b > c$, we have $a < c < b$, which gives
\[
    \frac{a}{\max\{b, c\}}
=   \frac{a}{b}
\ge \frac{a - |b - c|}{b - |b - c|}
=   \frac{a - |b - c|}{c}.
\]

Now, we analyze the expectation of the above using the inequalities
\[
    \Ex{}{|Z - l|} \le 7\cdot\sqrt{2C_3k'} + 7
\quad\text{and}\quad
    \Ex{}{|l - k'|} \le C_3 \cdot (k')^{\alpha},
\]
the second of which follows from the assumption of the lemma. We obtain
\begin{align*}
    \Ex{}{\frac{\min\{Z, k'\}}{\max\{l+Z-1, k\}}}
\ge &~\frac{k'-2}{k} - \frac{2}{k}\Ex{}{|Z-l|} - \frac{3}{k}\Ex{}{|l-k'|}\\
\ge &~\frac{k/2 - 5/2}{k} - \frac{14\sqrt{2C_3k'} + 14}{k} - \frac{3C_3(k')^{\alpha}}{k} \tag{$k' \ge (k-1)/2$}\\
\ge &~\frac{1}{2} - \frac{(3C_3 + 14\sqrt{2C_3})(k')^{\alpha} + 16.5}{k} \tag{$\alpha \ge 1/2$}.
\end{align*}
It follows that
\begin{equation}\label{eq:SOL-2-bound}
    \SOL_2
\ge \OPT\cdot \left(\frac{1}{2} - \frac{(3C_3 + 14\sqrt{2C_3})(k')^{\alpha} + 16.5}{k}\right).
\end{equation}

\paragraph{Adding the two terms.} Combining the bounds in \Cref{eq:SOL-1-bound,eq:SOL-2-bound} gives
\[
    \frac{\SOL}{\OPT}
\ge \frac{1}{2}\cdot\left(1 - \frac{200C_3}{(k')^{1-\alpha}}\right)\cdot\left(1 - \frac{4}{\sqrt{k}}\right) + \left(\frac{1}{2} - \frac{(3C_3 + 14\sqrt{2C_3})(k')^{\alpha} + 16.5}{k}\right).
\]
Rearranging and applying the relaxation $(1-a)(1-b) \ge 1 - a - b$ gives
\[
    1 - \frac{\SOL}{\OPT}
\le \frac{100C_3}{(k')^{1-\alpha}} + \frac{2}{\sqrt{k}} + \frac{(3C_3 + 14\sqrt{2C_3})(k')^{\alpha} + 16.5}{k}.
\]
For the first term, we note that $k' = \lfloor k/2\rfloor \ge \frac{5}{11}k$ holds for every integer $k \ge 10$. It follows that
\[
    \frac{100C_3}{(k')^{1-\alpha}}
\le \frac{100C_3 \cdot (11/5)^{1 - \alpha}}{k^{1-\alpha}}
\le \frac{149C_3}{k^{1 - \alpha}}.
\]
The second term, $2/\sqrt{k}$, is easily upper bounded by $2C_3/k^{1 - \alpha}$, since $C_3 \ge 1$ and $\alpha \in [1/2, 1]$. Finally, since $k' \le k/2$, the last term can be relaxed to
\[
    \frac{(3 + 14\sqrt{2})C_3\cdot(k/2)^{\alpha} + 16.5}{k}
\le \frac{(3/\sqrt{2} + 14)C_3\cdot k^{\alpha} + 16.5}{k}
\le \frac{33C_3}{k^{1 - \alpha}}.
\]

Therefore, we conclude that the competitive ratio of the algorithm is lower bounded by
\[
    1 - (149 + 2 + 33)\cdot \frac{C_3}{k^{1 - \alpha}}
\ge 1 - \frac{200C_3}{k^{1 - \alpha}}.
\]
\end{proof}

\subsection{Proof of \Cref{prop:reduction}}\label{sec:reduction-stronger}
Now, we derive \Cref{prop:reduction} from \Cref{prop:reduction-weaker}. The only missing piece is the following simple reduction that transforms a quantile estimation algorithm into one that only reads the second half of the stream, at a moderate cost on the accuracy.

\begin{lemma}\label{lemma:only-read-second-half}
    Suppose that, for some $\alpha \in [1/2, 1]$, there is a comparison-based quantile estimation algorithm $\A$ with memory usage $m(k)$ and an error of $O(k^{\alpha})$ in expectation. Then, there is a comparison-based quantile estimation algorithm $\A'$ with memory usage $m(\lfloor k/2\rfloor)$ and an $O(k^{\alpha})$ expected error. Furthermore, $\A'$ ignores all but the last $\lfloor n/2\rfloor$ elements in the length-$n$ stream.
\end{lemma}

The construction of $\A'$ is very simple: it ignores the first half of the stream, and finds the $(k/2)$-th largest element among the second half using $\A$. We can then translate the error of $\A'$ into the error of $\A$ using the concentration bound from \Cref{lemma:subset-rank-concentration}. This follows from a calculation similar to (but simpler than) that in \Cref{lemma:rank-in-second-half-to-overall-simplified}.

\begin{proof}
    Let $B \coloneqq \lfloor n / 2\rfloor$. We define the alternative algorithm $\A'$ as follows: If $k = 1$, we simply find the largest element using $O(1)$ memory. Otherwise, $\A'$ ignores the first $n - B$ elements and simulates algorithm $\A$ on the last $B$ elements with parameter $k' \coloneqq \lfloor k/2\rfloor$.

    It suffices to analyze $\A'$ in the non-trivial case that $n \ge k \ge 2$. Let $x^*$ denote the output of $\A$ (and thus the output of $\A'$). Let $l \coloneqq \rank_{(n-B+1):n}(x^*)$ denote its rank among the second half.

    By the triangle inequality, the expected error can be upper bounded as follows:
    \[
        \Ex{}{|\rank_{1:n}(x^*) - k|}
    \le \Ex{}{\left|\rank_{1:n}(x^*) - l \cdot \frac{n}{B}\right|} + \Ex{}{\left|l \cdot \frac{n}{B} - k\right|}.
    \]
    In the rest of the proof, we upper bound both terms on the right-hand side by $O(k^{\alpha})$.

    \paragraph{Upper bound the first term.} We first analyze the conditional expectation given the realization of $l$. Since $\A$ is comparison-based, conditioning on the value of $l$, $\{s_{n-B+1}, s_{n-B+2}, \ldots, s_n\}$ is still uniformly distributed among all size-$B$ subsets of $\{s_1, s_2, \ldots, s_n\}$. Applying \Cref{lemma:subset-rank-concentration} with parameters
    \[
        \tilde n = n, \quad \tilde k = B, \quad \tilde i = l
    \]
    gives
    \[
        \Ex{}{\left|\rank_{1:n}(x^*) - l \cdot \frac{n}{B}\right|}
    \le 2\sqrt{l \cdot \frac{n}{B}}\cdot\frac{n}{B} + \frac{n^2}{B^2}.
    \]
    Since $B = \lfloor n / 2\rfloor$ and $n \ge 2$, we have $n / B \le 3$. The above can then be simplified into
    \[
        \Ex{}{\left|\rank_{1:n}(x^*) - l \cdot \frac{n}{B}\right|}
    \le 6\sqrt{3}\cdot\sqrt{l} + 9.
    \]
    It remains to take an expectation over $l$. By the assumption on quantile estimator $\A$, we have
    \begin{equation}\label{eq:l-minus-k-prime-upper-bound}
        \Ex{}{|l - k'|} \le C(k')^{\alpha} \le Ck^{\alpha}
    \end{equation}
    for some universal constant $C$. Applying the triangle inequality gives
    \begin{equation}\label{eq:l-upper-bound}
        \Ex{}{l}
    \le \Ex{}{k' + |l - k'|}
    \le k' + Ck^{\alpha}
    \le (C + 1)k.
    \end{equation}
    By Jensen's inequality, we have
    \[
        \Ex{}{\sqrt{l}}
    \le \sqrt{\Ex{}{l}}
    \le \sqrt{C + 1}\cdot\sqrt{k},
    \]
    and
    \[
        \Ex{}{\left|\rank_{1:n}(x^*) - l \cdot \frac{n}{B}\right|}
    \le 6\sqrt{3(C+1)}\cdot\sqrt{k} + 9
    \le O(\sqrt{k})
    \le O(k^{\alpha}),
    \]
    where the $O(\cdot)$ notation hides a universal constant that only depends on $C$. The last step above holds since $\alpha \ge 1/2$.

    \paragraph{Upper bound the second term.} We can upper bound $\left|l\cdot \frac{n}{B} - k\right|$ by
    \[
        \left|l\cdot \frac{n}{B} - 2l\right| + |2l - 2k'| + |2k' - k|.
    \]

    For the expectation of the first term, note that $B = \lfloor n/2 \rfloor$ and $n \ge 2$ implies $n / B = 2 + O(1/n)$. It follows that
    \[
        \Ex{}{\left|l\cdot \frac{n}{B} - 2l\right|}
    =   \left|\frac{n}{B} - 2\right|\cdot \Ex{}{l}
    =   O(1/n) \cdot \Ex{}{l}
    \le O(1/n) \cdot (C + 1)k
    =   O(1),
    \]
    where the third step applies \Cref{eq:l-upper-bound}, and the last step follows from $k \le n$.
    
    For the second term $|2l - 2k'|$, \Cref{eq:l-minus-k-prime-upper-bound} gives
    \[
        \Ex{}{|2l - 2k'|}
    \le 2Ck^{\alpha}.
    \]
    
    Finally, the last term $|2k' - k|$ is either $0$ or $1$. Therefore, we conclude that
    \[
        \Ex{}{\left|l\cdot \frac{n}{B} - k\right|}
    \le O(1) + 2Ck^{\alpha} + 1
    =   O(k^{\alpha}).
    \]
    Again, the $O(\cdot)$ notation hides a constant factor that only depends on $C$.

    Therefore, we conclude that the expected error of $\A'$ is $O(k^{\alpha})$.
\end{proof}

We end by deriving \Cref{prop:reduction} from \Cref{prop:reduction-weaker} and \Cref{lemma:only-read-second-half}.

\begin{proof}[Proof of \Cref{prop:reduction}]
    Let $\A$ be the quantile estimator with memory usage $m$ and an $O(k^{\alpha})$ expected error. By \Cref{lemma:only-read-second-half}, we have an algorithm $\A'$ with the same memory usage and expected error. By \Cref{prop:reduction-weaker}, using $\A'$ as the quantile estimation algorithm, $\choosetopk$ (\Cref{algo:k-secretary}) has a competitive ratio of $1 - O(1/k^{1-\alpha})$.

    It remains to show that \Cref{algo:k-secretary} can be implemented using $m + O(1)$ memory (instead of $O(m\log k)$ memory). To see this, note that when we call $\choosetopk(n, k, s)$, the algorithm makes a recursive call $\choosetopk(\lfloor n / 2\rfloor, \lfloor k/2\rfloor, s_{1:\lfloor n / 2\rfloor})$, and also runs the quantile estimation algorithm $\A'(\lfloor n / 2\rfloor, \lfloor k/2\rfloor, s_{1:\lfloor n / 2\rfloor})$ in parallel. By the construction of $\A'$ (from \Cref{lemma:only-read-second-half}), $\A'(\lfloor n / 2\rfloor, \lfloor k/2\rfloor, s_{1:\lfloor n / 2\rfloor})$ only accesses the elements 
    \[
        s_{\lfloor n / 2\rfloor - \lfloor n / 4\rfloor + 1}, s_{\lfloor n / 2\rfloor - \lfloor n / 4\rfloor + 2}, \ldots, s_{\lfloor n / 2\rfloor}.
    \]
    More generally, the $i$-th recursive call (where $i \in \{1, 2, \ldots, \lfloor \log_2 k\rfloor\})$) of $\choosetopk$ calls $\A'$ on the first $\lfloor n/2^i\rfloor$ elements. By \Cref{lemma:only-read-second-half}, $\A'$ only reads the elements with indices between
    \[
        \lfloor n / 2^i\rfloor - \lfloor n / 2^{i+1}\rfloor + 1
    \]
    and $\lfloor n / 2^i\rfloor$ (inclusive). It follows that the $\Theta(\log k)$ calls to $\A'$ do not overlap in terms of the (contiguous) subsequence of $s_1, s_2, \ldots, s_n$ that they access. Therefore, we only need to allocate a memory of $m$ for procedure $\A'$. Apart from this, $\choosetopk$ performs $O(\log k)$ levels of recursion, each of which takes $O(1)$ words of memory. If we further expand the recursion into a loop implementation, at any time, we only need to store $O(1)$ different values of $x^*$ (either being computed by algorithm $\A'$, or being used as a threshold for accepting elements). Therefore, we can implement the algorithm with the desired competitive ratio using $m + O(1)$ space.
\end{proof}

\bibliographystyle{alpha}
\bibliography{sample}

\appendix

\section{Technical Lemmas}\label{sec:technical-lemmas}
\lemmasubsample*

\begin{proof}[Proof of \Cref{lemma:subsample}]
    It suffices to show that, for every subset $S' \subseteq S$, $\{s_1, s_2, \ldots, s_B\} = S'$ holds with probability exactly $p^{m}(1-p)^{n-m}$, where $m = |S'|$.

    For $\{s_1, s_2, \ldots, s_B\}$ to be equal to $S'$, the following two conditions must hold: (1) $B = m$ is sampled from $\Binomial(n, p)$; (2) The first $m$ elements of the stream constitute the set $S'$. The former happens with probability $p^{m}(1-p)^{n-m} \cdot \binom{n}{m}$. Since $s$ and $B$ are independent, conditioning on $B = m$, $s_1, s_2, \ldots, s_n$ is still uniformly distributed among all permutations. In particular, $\{s_1, s_2, \ldots, s_m\}$ is uniformly distributed among all size-$m$ subsets of $S$. Therefore, the latter condition holds with probability $1 / \binom{n}{m}$. Therefore, the overall probability is given by
    \[
        \left[p^{m}(1-p)^{n-m} \cdot \binom{n}{m}\right] \cdot\frac{1}{\binom{n}{m}} = p^{m}(1-p)^{n-m}.
    \]
\end{proof}

\subsetrankconcentration*

The proof of \Cref{lemma:subset-rank-concentration} is based on the following lemma, which relates sampling without replacement to sampling with replacement.
\begin{lemma}[\cite{bardenet2015concentration}]\label{lemma.cite}
    Let $\mathcal{X} = (x_1, \ldots,  x_N)$ be a finite population of $N$ real points, $X_1, \ldots , X_n$ denote a random sample without replacement from $\mathcal{X}$ and $Y_1, \ldots , Y_n$ denote a random sample with replacement from $\mathcal{X}$. If $f : \mathbbm{R} \rightarrow \mathbbm{R}$ is continuous and convex, then
    \[
        \Ex{}{f\left(\sum_{i=1}^n X_i\right)}
    \le \Ex{}{f\left(\sum_{i=1}^n Y_i\right)}.
    \]
\end{lemma}
 
\begin{proof}[Proof of \Cref{lemma:subset-rank-concentration}]
Define binary random variables $X_1, X_2, \ldots, X_n$ such that $X_j = 1$ if  element $j \in [n]$ is included in the size-$k$ subset, and $X_j = 0$ otherwise. Then, $X_1$ through $X_n$ can be viewed as being sampled from the size-$n$ population
\[
    \mathcal{X} = (\underbrace{0, 0, \ldots, 0}_{n-k\text{ copies}}, \underbrace{1, 1, \ldots, 1}_{k\text{ copies}})
\]
without replacement. Towards applying \Cref{lemma.cite}, we consider random variables $Y_1, Y_2, \ldots, Y_n$ that are sampled from $\mathcal{X}$ \emph{with} replacement. Equivalently, $Y_1$ through $Y_n$ are independently sampled from $\Bern(k/n)$. Since the function $x \mapsto \exp(tx)$ is convex for any $t \in \mathbb{R}$, \Cref{lemma.cite} implies that, for any $n' \in [n]$,
\begin{equation}\label{eq:without-to-with-replacement}
    \Ex{X}{\exp\left(t\cdot \sum_{j=1}^{n'}X_j\right)}
\le \Ex{Y}{\exp\left(t\cdot \sum_{j=1}^{n'}Y_j\right)}.
\end{equation}

In the rest of the proof, we first control the tail probabilities of $x - i\cdot \frac{n}{k}$ on both sides. We then get an upper bound on the expectation via integration.

\paragraph{Control the left-tail.} Fix $m \ge 0$. Recall that $x$ denotes the $i$-th smallest element in the random size-$k$ subset. Our goal is to upper bound the tail probability $\pr{}{x \le i \cdot \frac{n}{k} - m\cdot\sqrt{i \cdot \frac{n}{k}}}$. When $i \cdot \frac{n}{k} - m\cdot\sqrt{i \cdot \frac{n}{k}} < 1$, this probability is trivially $0$. Otherwise, we note that, for $x \le j$ to hold, we must have $\sum_{r = 1}^j X_r \ge i$. Applying this observation to $j = \left\lfloor i\cdot \frac{n}{k} - m\cdot \sqrt{i\cdot \frac{n}{k}}\right\rfloor$ gives
\begin{align*}
        \pr{}{x \le i\cdot \frac{n}{k} - m\cdot \sqrt{i\cdot \frac{n}{k}}}
    &=  \pr{}{x \le \left\lfloor i\cdot \frac{n}{k} - m\cdot \sqrt{i\cdot \frac{n}{k}}\right\rfloor}\\
    &\le\pr{X}{\sum_{r=1}^{\left \lfloor i\cdot \frac{n}{k} - m\cdot \sqrt{i\cdot \frac{n}{k}}\right\rfloor } X_r \ge i}
    =   \inf_{t > 0}\pr{X}{e^{t\sum_{r=1}^{\left \lfloor i\cdot \frac{n}{k} - m\cdot \sqrt{i\cdot \frac{n}{k}}\right \rfloor }X_r} \ge e^{ti}}\\
    &\le\inf_{t>0}\frac{\Ex{X}{e^{t\sum_{r=1}^{\left \lfloor i\cdot \frac{n}{k} - m\cdot \sqrt{i\cdot \frac{n}{k}}\right \rfloor }X_r}}}{e^{ti}} \tag{Markov's inequality} \\
    &\le\inf_{t>0}\frac{\Ex{Y}{e^{t\sum_{r=1}^{\left \lfloor i\cdot \frac{n}{k} - m\cdot \sqrt{i\cdot \frac{n}{k}}\right \rfloor }Y_r}}}{e^{ti}} \tag{\Cref{eq:without-to-with-replacement}}\\
    &=  \inf_{t>0}\frac{\Pi_{r=1}^{\left \lfloor i\cdot \frac{n}{k} - m\cdot \sqrt{i\cdot \frac{n}{k}} \right \rfloor }\Ex{Y_r}{e^{t \cdot Y_r}}}{e^{ti}},
\end{align*}
where the last step holds since $Y_1, Y_2, \ldots$ are independent.

Recall that each $Y_r$ follows $\Bern(k/n)$. By Hoeffding's lemma, we have $\Ex{Y_r}{e^{t\cdot \left(Y_r - \frac{k}{n}\right)}} \le e^{t^2/8}$, which further implies
\[
    \Ex{Y_r}{e^{t\cdot Y_r}} \le \exp\left(\frac{t^2}{8} + t\cdot \frac{k}{n}\right).
\]
Plugging the above into the tail bound gives
\begin{align*}
    \pr{}{x \le i\cdot \frac{n}{k} - m\cdot \sqrt{i\cdot \frac{n}{k}}}
    &\le\inf_{t>0}\exp\left(\left(\frac{t^2}{8} + t\cdot \frac{k}{n}\right)\cdot \left\lfloor i\cdot \frac{n}{k} - m\cdot \sqrt{i\cdot \frac{n}{k}}\right\rfloor - ti\right)\\
    &\le\inf_{t>0}\exp\left(\left(\frac{t^2}{8} + t\cdot \frac{k}{n}\right)\cdot \left( i\cdot \frac{n}{k} - m\cdot \sqrt{i\cdot \frac{n}{k}}\right) - ti\right)\\
    &=  \inf_{t > 0}\exp\left(\frac{t^2}{8}\cdot\left(i\cdot\frac{n}{k} - m\cdot\sqrt{i \cdot \frac{n}{k}}\right) - t\cdot m\cdot\sqrt{i\cdot \frac{k}{n}}\right)\\
    &=  \exp\left(-\frac{2m^2 \cdot i\cdot \frac{k}{n}}{i\cdot \frac{n}{k} - m\cdot \sqrt{i\cdot \frac{n}{k}}}\right)\\
    &\le\exp\left(-\frac{2m^2 \cdot i\cdot \frac{k}{n}}{i\cdot \frac{n}{k}}\right)
    =   \exp\left(-2m^2\cdot \frac{k^2}{n^2}\right), \tag{$m \ge 0$}
\end{align*}
where the fourth step applies $\inf_{t > 0}(at^2 - bt) = -\frac{b^2}{4a}$ for $a, b > 0$.

We conclude that, for all $m \ge 0$,
\begin{equation}\label{eq:left-tail-bound}
    \pr{}{x \le i\cdot \frac{n}{k} - m\cdot \sqrt{i\cdot \frac{n}{k}}}
\le \exp\left(-2m^2\cdot \frac{k^2}{n^2}\right).
\end{equation}

\paragraph{Control the right-tail.} Similarly, we fix $m \ge 0$. If $i\cdot\frac{n}{k} + m\cdot\sqrt{i \cdot\frac{n}{k}} > n$, the tail probability $\pr{}{x \ge i\cdot\frac{n}{k} + m\cdot\sqrt{i \cdot\frac{n}{k}}}$ is trivially $0$. Otherwise, note that, for the $i$-th smallest element in the size-$k$ subset (namely, $x$) to be larger than or equal to $j$, we must have $\sum_{r = 1}^j X_r \le i$. This gives
\begin{align*}
    \pr{}{x \ge i\cdot\frac{n}{k} + m\cdot\sqrt{i \cdot\frac{n}{k}}}
&=  \pr{}{x \ge \left\lceil i\cdot \frac{n}{k} + m\cdot \sqrt{i\cdot \frac{n}{k}}\right\rceil}\\
&\le\pr{X}{\sum_{r=1}^{\left\lceil i\cdot \frac{n}{k} + m\cdot \sqrt{i\cdot \frac{n}{k}}\right \rceil } X_r \le i}
=   \inf_{t > 0}\pr{X}{e^{-t\cdot \sum_{r=1}^{\left\lceil i\cdot \frac{n}{k} + m\cdot \sqrt{i\cdot \frac{n}{k}}\right \rceil }X_r} \ge e^{-ti}}\\
&\le\inf_{t>0}\frac{\Ex{X}{e^{-t\sum_{r=1}^{\left\lceil i\cdot \frac{n}{k} + m\cdot \sqrt{i\cdot \frac{n}{k}} \right \rceil}X_r}}}{e^{-ti}} \tag{Markov's inequality}\\
&\le\inf_{t>0}\frac{\Ex{Y}{e^{-t\sum_{r=1}^{\left\lceil i\cdot \frac{n}{k} + m\cdot \sqrt{i\cdot \frac{n}{k}}\right \rceil }Y_r}}}{e^{-ti}} \tag{\Cref{eq:without-to-with-replacement}}\\
&=  \inf_{t>0}\frac{\Ex{Y}{\Pi_{r=1}^{\left\lceil i\cdot \frac{n}{k} + m\cdot \sqrt{i\cdot \frac{n}{k}}\right \rceil } e^{-t \cdot Y_r}}}{e^{-ti}}\\
&=  \inf_{t>0}\frac{\Pi_{r=1}^{\left\lceil i\cdot \frac{n}{k} + m\cdot \sqrt{i\cdot \frac{n}{k}}\right \rceil }\Ex{Y_r}{e^{-t \cdot Y_r}}}{e^{-ti}},
\end{align*}
where the last step follows from the independence among $Y_1, Y_2, \ldots, Y_n$.

Again, recall that each $Y_r$ follows $\Bern(k/n)$, so Hoeffding's lemma gives $\Ex{Y_r}{[e^{-t\cdot \left(Y_r - \frac{k}{n}\right)}} \le e^{t^2/8}$. Thus,
\[
    \Ex{Y_r}{e^{-t\cdot Y_r}} \le \exp\left(\frac{t^2}{8} - t\cdot \frac{k}{n}\right).
\]
Plugging the above into the upper bound on the right-tail gives
\[
    \pr{}{x \ge i\cdot\frac{n}{k} + m\cdot\sqrt{i \cdot\frac{n}{k}}}
\le \inf_{t>0}\exp\left(\left(\frac{t^2}{8} - t\cdot \frac{k}{n}\right)\cdot \left\lceil i\cdot \frac{n}{k} + m\cdot \sqrt{i\cdot \frac{n}{k}}\right\rceil + ti\right).
\]
Let $\overline{m} \in \mathbb{R}$ be the unique value such that $i \cdot \frac{n}{k} + \overline{m}\cdot\sqrt{i \cdot \frac{n}{k}} = \left\lceil i\cdot \frac{n}{k} + m\cdot \sqrt{i\cdot \frac{n}{k}}\right\rceil$. Clearly, we have $\overline{m} \ge m$. Then, we can re-write the above into
\begin{align*}
        \pr{}{x \ge i\cdot\frac{n}{k} + m\cdot\sqrt{i \cdot\frac{n}{k}}}
&\le    \inf_{t>0}\exp\left(\left(\frac{t^2}{8} - t\cdot \frac{k}{n}\right)\cdot \left( i\cdot \frac{n}{k} + \overline{m}\cdot \sqrt{i\cdot \frac{n}{k}}\right) + ti\right)\\
&=      \inf_{t>0}\exp\left(\frac{t^2}{8}\cdot \left( i\cdot \frac{n}{k} + \overline{m}\cdot \sqrt{i\cdot \frac{n}{k}}\right) - t \cdot \overline{m} \cdot \sqrt{i \cdot \frac{k}{n}}\right)\\
&=      \exp\left(-\frac{2\cdot  \overline{m}^2\cdot i\cdot \frac{k}{n}}{i\cdot \frac{n}{k} + \overline{m}\cdot \sqrt{i\cdot \frac{n}{k}}}\right).
\end{align*}
Again, the last step above follows from $\inf_{t > 0}(at^2 - bt) = -\frac{b^2}{4a}$, which holds for $a, b > 0$.

Note that for any $a, b > 0$, the function $x \mapsto \frac{x^2}{ax + b}$ is monotone increasing on $(0, +\infty)$. Since $\overline{m} \ge m$, we have
\[
    \frac{2\cdot  \overline{m}^2\cdot i\cdot \frac{k}{n}}{i\cdot \frac{n}{k} + \overline{m}\cdot \sqrt{i\cdot \frac{n}{k}}}
\ge \frac{2\cdot  m^2\cdot i\cdot \frac{k}{n}}{i\cdot \frac{n}{k} + m\cdot \sqrt{i\cdot \frac{n}{k}}},
\]
which further implies
\begin{align*}
        \pr{}{x \ge i\cdot\frac{n}{k} + m\cdot\sqrt{i \cdot\frac{n}{k}}}
&\le    \exp\left(-\frac{2\cdot  m^2\cdot i\cdot \frac{k}{n}}{i\cdot \frac{n}{k} + m\cdot \sqrt{i\cdot \frac{n}{k}}}\right)\\
&\le    \exp\left(-m^2\cdot \frac{k^2}{n^2}\right) + \exp\left(-m\cdot \sqrt{i\cdot \frac{k}{n}}\cdot \frac{k}{n}\right).
\end{align*}

We conclude that, for all $m \ge 0$,
\begin{equation}\label{eq:right-tail-bound}
    \pr{}{x \ge i\cdot \frac{n}{k} + m\cdot \sqrt{i\cdot \frac{n}{k}}}
\le \exp\left(-m^2\cdot \frac{k^2}{n^2}\right) + \exp\left(-m\cdot \sqrt{i\cdot \frac{k}{n}}\cdot \frac{k}{n}\right).
\end{equation}

    \paragraph{Put everything together.} Using the fact that $\Ex{}{X} = \int_{0}^{+\infty}\pr{}{X \ge \tau}~\rmd\tau$ holds for any non-negative random variable $X$, we can write the expectation of interest into the following:
    \[
        \Ex{}{\left|x - i\cdot\frac{n}{k}\right|}
    =   \int_{0}^{+\infty}\pr{}{x \le i\cdot\frac{n}{k} - \tau}~\rmd\tau + \int_{0}^{+\infty}\pr{}{x \ge i\cdot\frac{n}{k} + \tau}~\rmd\tau.
    \]

    Let $I_{-}$ and $I_{+}$ denote the two integrals above. By a change of variables, we have
    \begin{align*}
        I_{-}
    &=  \sqrt{i \cdot\frac{n}{k}}\cdot\int_{0}^{+\infty}\pr{}{x \le i\cdot\frac{n}{k} - \tau\cdot \sqrt{i \cdot\frac{n}{k}}}~\rmd\tau\\
    &\le\sqrt{i \cdot\frac{n}{k}}\cdot\int_{0}^{+\infty}\exp\left(-2\tau^2\cdot\frac{k^2}{n^2}\right)~\rmd\tau \tag{\Cref{eq:left-tail-bound}}\\
    &=  \sqrt{i \cdot\frac{n}{k}}\cdot\frac{n}{k}\cdot\int_{0}^{+\infty}\exp\left(-2\tau^2\right)~\rmd\tau.
    \end{align*}

    Similarly, we have
    \begin{align*}
        I_{+}
    &=  \sqrt{i \cdot\frac{n}{k}}\cdot\int_{0}^{+\infty}\pr{}{x \ge i\cdot\frac{n}{k} + \tau\cdot \sqrt{i \cdot\frac{n}{k}}}~\rmd\tau\\
    &\le\sqrt{i \cdot\frac{n}{k}}\cdot\left[\int_{0}^{+\infty}\exp\left(-\tau^2\cdot\frac{k^2}{n^2}\right)~\rmd\tau + \int_{0}^{+\infty}\exp\left(-\tau\cdot\sqrt{i \cdot \frac{k}{n}}\cdot\frac{k}{n}\right)~\rmd\tau\right] \tag{\Cref{eq:right-tail-bound}}\\
    &=  \sqrt{i \cdot\frac{n}{k}}\cdot\left[\frac{n}{k}\cdot\int_{0}^{+\infty}\exp\left(-\tau^2\right)~\rmd\tau + \frac{1}{\sqrt{i \cdot \frac{k}{n}}\cdot\frac{k}{n}}\right]\\
    &=  \sqrt{i\cdot\frac{n}{k}}\cdot\frac{n}{k}\cdot \int_{0}^{+\infty}\exp\left(-\tau^2\right)~\rmd\tau + \frac{n^2}{k^2}.
    \end{align*}

    Finally, since
    \[
        \int_{0}^{+\infty}e^{-\tau^2} + e^{-2\tau^2}~\rmd\tau
    =   \frac{2 + \sqrt{2}}{4}\cdot\sqrt{\pi}
    \le 1.513
    <   2,
    \]
    we have the desired upper bound:
    \[
        \Ex{}{\left|x - i\cdot\frac{n}{k}\right|}
    =   I_{-} + I_{+}
    \le 2\sqrt{i\cdot\frac{n}{k}}\cdot\frac{n}{k} + \frac{n^2}{k^2}.
    \]
\end{proof}

\expectednoverb*

\begin{proof}[Proof of \Cref{lemma.integral1}]
    By a Chernoff bound,
    \[
        \pr{B \sim \Binomial\left(n, 1/2\right)}{B \le \frac{n}{4}}  \le \exp\left(-2n \cdot  (1/4)^2\right) = e^{-n/8}.
    \]
    Then, since $|n/B - 2| \le n$ holds for every $B \in \{1, 2, \ldots, n\}$, we have
    \[
        \Ex{B \sim \Binomial(n, 1/2)}{\left|\frac{n}{B} - 2\right| \cdot \1{1 \le B \le n/4}}
    \le n \cdot \pr{B}{1 \le B \le n/4}
    \le n\cdot e^{-n/8}.
    \]
    
    It remains to upper bound the contribution from the $B > n / 4$ case. Note that, for every $B > n/4$, we have
    \[
        \left|\frac{n}{B} - 2\right|
    =   \frac{2}{B}\left|\frac{n}{2} - B\right|
    \le \frac{8}{n}\left|\frac{n}{2} - B\right|.
    \]
    Taking an expectation over $B \sim \Binomial(n, 1/2)$ gives
    \[
        \Ex{B \sim \Binomial(n, 1/2)}{\left|\frac{n}{B} - 2\right| \cdot \1{B > n/4}}\\
    \le \frac{8}{n}\cdot\Ex{B \sim \Binomial(n, 1/2)}{|B - n/2|}.
    \]
    Then, by Jensen's inequality,
    \[
        \Ex{B \sim \Binomial(n, 1/2)}{|B - n/2|}
    \le \sqrt{\Ex{B}{(B - n/2)^2}}
    =   \sqrt{\Var{}{B}}
    =   \sqrt{n/4} = \frac{\sqrt{n}}{2}.
    \]
    It follows that
    \[
        \Ex{B \sim \Binomial(n, 1/2)}{\left|\frac{n}{B} - 2\right| \cdot \1{B > n/4}}
    \le \frac{8}{n}\cdot\frac{\sqrt{n}}{2} = \frac{4}{\sqrt{n}}.
    \]
    
    In total, we have
    \begin{align*}
        &~\Ex{B \sim \Binomial(n, 1/2)}{\left|\frac{n}{B} - 2\right|\cdot\1{B \ne 0}}\\
    =   &~\Ex{B \sim \Binomial(n, 1/2)}{\left|\frac{n}{B} - 2\right|\cdot\1{1 \le B \le n/4}} + \Ex{B \sim \Binomial(n, 1/2)}{\left|\frac{n}{B} - 2\right|\cdot\1{B > n/4}}\\
    \le &~\frac{4}{\sqrt{n}} + n\cdot e^{-n/8}
    \le \frac{14}{\sqrt{n}},
    \end{align*}
    where the last step applies $x^{3/2} \cdot e^{-x/8} \le 10$, which holds for all $x \ge 0$.
\end{proof}

\end{document}